\documentclass[ejsv2,preprint,noshowframe]{imsart}

\RequirePackage[authoryear]{natbib}
\RequirePackage[colorlinks,citecolor=blue,urlcolor=blue]{hyperref}
\RequirePackage{graphicx}

\graphicspath{{./figure/}}

\usepackage{amsmath}
\usepackage[plain,noend,ruled]{algorithm2e}

\makeatletter
\renewcommand{\algocf@captiontext}[2]{#1\algocf@typo. \AlCapFnt{}#2} 
\def\@algocf@capt@plain{top}
\renewcommand{\algocf@makecaption}[2]{%
  \addtolength{\hsize}{\algomargin}%
  \sbox\@tempboxa{\algocf@captiontext{#1}{#2}}%
  \ifdim\wd\@tempboxa >\hsize
    \hskip .5\algomargin%
    \parbox[t]{\hsize}{\algocf@captiontext{#1}{#2}}
  \else%
    \global\@minipagefalse%
    \hbox to\hsize{\box\@tempboxa}
  \fi%
  \addtolength{\hsize}{-\algomargin}%
}
\makeatother

\startlocaldefs
\usepackage{thmtools}
\declaretheorem[name=Theorem, Refname={Theorem,Theorems},
                numberwithin=section]{theorem}

\declaretheorem[name=Corollary, Refname={Corollary,Corollaries},
                sibling=theorem]{corollary}
\declaretheorem[name=Assumption, Refname={Assumption,Assumptions},
                sibling=theorem, style=definition]{assumption}
\declaretheorem[name=Definition, Refname={Definition,Definitions},
                sibling=theorem, style=definition]{definition}
\declaretheorem[name=Remark, Refname={Remark,Remarks},
                sibling=theorem, style=definition, qed=\qedsymbol]{remark}
\declaretheorem[name=Theorem, Refname={Theorem,Theorems},
                numberwithin=section]{thm}
\declaretheorem[name=Lemma, Refname={Lemma,Lemmas}, sibling=thm]{lem}
\declaretheorem[name=Assumption, Refname={Assumption,Assumptions},
                numberwithin=section]{asu}
\declaretheorem[name=Remark, Refname={Remark,Remarks},
                numberwithin=section, qed=\qedsymbol]{rmk}
\declaretheorem[name=Proposition, Refname={Proposition,Propositions},
                sibling=thm]{prop}

\usepackage[capitalize, nameinlink, sort]{cleveref}
\crefname{theorem}{Theorem}{Theorems}
\crefname{lemma}{Lemma}{Lemmas}
\crefname{corollary}{Corollary}{Corollaries}
\crefname{assumption}{Assumption}{Assumptions}
\crefname{definition}{Definition}{Definitions}
\crefname{remark}{Remark}{Remarks}
\crefname{thm}{Theorem}{Theorems}
\crefname{lem}{Lemma}{Lemmas}
\crefname{prop}{Proposition}{Propositions}
\crefname{asu}{Assumption}{Assumptions}
\crefname{rmk}{Remark}{Remarks}
\crefname{thmlisti}{Theorem}{Theorems}
\crefname{asulisti}{Assumption}{Assumptions}
\crefname{algorithm}{Algorithm}{Algorithms}
\Crefname{algorithm}{Algorithm}{Algorithms}

\usepackage{enumitem}
\usepackage{url}
\usepackage{booktabs}

\newlist{thmlist}{enumerate}{1}
\setlist[thmlist]{label=(\roman*), ref=\thetheorem\,(\roman*)}

\newlist{asulist}{enumerate}{1}
\setlist[asulist]{label=(\roman*), ref=\theassumption\,(\roman*)}

\usepackage{dsfont}
\newcommand{\indic}{{\mathds 1}}

\newcommand{\Var}{\mathrm{Var}}
\newcommand{\Cov}{\mathrm{Cov}}

\newcommand{\ba}{\bar{a}}
\newcommand{\bA}{\bar{A}}
\newcommand{\bX}{\bar{X}}

\newcommand{\cM}{\mathcal{M}}
\newcommand{\cT}{\mathcal{T}}
\newcommand{\cO}{\mathcal{O}}

\newcommand{\tp}{\tilde{p}}
\newcommand{\td}{\tilde{d}}
\newcommand{\PP}{\mathbb{P}}
\newcommand{\EE}{\mathbb{E}}
\newcommand{\RR}{\mathbb{R}}

\newcommand{\pto}{\stackrel{p}{\to}}
\newcommand{\dto}{\stackrel{d}{\to}}

\newcommand{\mb}{m^\mathrm{(b)}}
\newcommand{\mc}{m^\mathrm{(c)}}
\newcommand{\phic}{\phi^\mathrm{(c)}}

\newcommand{\cee}{\mathrm{CEE}}
\newcommand{\betainit}{\beta^{\mathrm{init}}}
\newcommand{\betainithat}{\hat\beta^{\mathrm{init}}}

\endlocaldefs

\begin{document}

\begin{frontmatter}
\title{Efficient and Globally Robust Causal Excursion Effect Estimation\thanksref{T1}}
\thankstext{T1}{This is the authors' final version of the article published in the \emph{Electronic Journal of Statistics}, Vol.~20, No.~2 (2026), \url{https://doi.org/10.1214/26-EJS2606}.}
\runtitle{Efficient and Globally Robust Causal Excursion Effect Estimation}

\begin{aug}
\author[A]{\fnms{Zhaoxi}~\snm{Cheng}\ead[label=e1]{carrie.cheng@duke.edu}},
\author[B]{\fnms{Lauren}~\snm{Bell}\ead[label=e2]{lauren.m.bell@kcl.ac.uk}}
\and
\author[C]{\fnms{Tianchen}~\snm{Qian}\ead[label=e3]{t.qian@uci.edu}}

\address[A]{Department of Statistics, Duke University, Durham, NC, USA\printead[presep={,\ }]{e1}}

\address[B]{Department of Biostatistics at the Institute of Psychiatry, Psychology and Neuroscience, Kings College London, London, UK\printead[presep={,\ }]{e2}}

\address[C]{Department of Statistics, University of California, Irvine, Irvine, CA, USA\printead[presep={,\ }]{e3}}

\runauthor{Cheng et al.}
\end{aug}

\begin{abstract}
Causal excursion effect (CEE) characterizes the effect of an intervention under policies that deviate from the experimental policy. It is widely used to study the effect of time-varying interventions that have the potential to be frequently adaptive, such as those delivered through smartphones. We study the semiparametric efficient estimation of CEE and we derive a semiparametric efficiency bound for CEE under working assumptions, in the context of micro-randomized trials. We propose a class of two-stage estimators that achieve the efficiency bound and are robust to misspecified nuisance models. In deriving the asymptotic property of the estimators, we establish a collection of easy-to-use theorems for globally robust Z-estimators with either cross-fitted or non-cross-fitted nuisance parameters that are applicable beyond the proposed estimators for CEE. We demonstrate efficiency gain of the proposed estimator compared to existing ones through simulations and a real data application using the Drink Less micro-randomized trial.
\end{abstract}



\end{frontmatter}

\section{Introduction}
\label{sec:introduction}

Digital behavior change interventions, such as push notifications delivered through smartphones and wearables, are being developed in various domains including mobile health, education, and social justice \citep{bell2020notifications,pham2016effects,bhatti2017moving}. Micro-randomized trial (MRT) is the state-of-the-art experimental design for developing and optimizing such interventions \citep{klasnja2015,dempsey2015,liao2016sample}. In an MRT, each individual is repeatedly randomized among treatment options hundreds or thousands of times, giving rise to longitudinal data with time-varying treatments.

Causal excursion effect (CEE) is the key quantity of interest for primary and secondary analyses of MRTs \citep{boruvka2018,qian2021estimating,dempsey2020stratified}. CEE captures whether an intervention is effective and how the intervention effect changes over time or interacts with the individual's context, which are central questions to the optimization of the interventions. Existing semiparametric models for CEE typically involve multiple nuisance functions such as outcome regression models and stabilizing weights \citep{boruvka2018,qian2021estimating}. It is shown that the choice of the nuisance models can substantially impact the efficiency of the CEE estimation \citep{qian2022microrandomized}. However, the current practice of modeling the nuisance functions is ad hoc and is often based on statistical significance from exploratory regression analysis \citep{klasnja2018efficacy,bell2023notifications}. 

Semiparametric efficiency theory characterizes the minimum achievable asymptotic variance of a regular estimator through the efficient influence function \citep{newey1990,tsiatis2007semiparametric}. Modern methods such as targeted minimum loss based estimation \citep{van2011targeted} and double/debiased machine learning \citep{chernozhukov2018double} can attain this bound in certain problems while accommodating flexible nonparametric nuisance estimators. In MRTs, efficiency gains translate to smaller required sample sizes or greater power to detect treatment effect moderation. Yet the MRT setting introduces distinctive structure not covered by existing results: sequentially randomized treatments and a causal estimand marginalized over all history information except for a subset of time-varying variables. This calls for the semiparametric efficiency theory for CEE estimation.

We comprehensively investigate the semiparametric efficient estimation for CEE, propose new estimators, and establish supporting asymptotic theory for a class of globally robust Z-estimators. We study CEE with identity or log link, and derive the semiparametric efficiency bound for estimating the moderated CEE under working assumptions. We propose a class of two-stage estimators that achieve the efficiency bound: in the first stage, optimal nuisance functions including outcome regressions and inverse variance weights are identified and fitted; in the second stage, the causal parameter is estimated. We establish consistency, asymptotic normality, and robustness to nuisance-model misspecification. In addition, we establish a collection of easy-to-use theorems for a class of two-stage globally robust Z-estimators with cross-fitted or non-cross-fitted infinite-dimensional nuisance parameters, where the first-stage estimator may converge at an arbitrarily slow rate to an arbitrary limit, facilitating the use of flexible machine learning algorithms. Since the preprint version of the manuscript \citep{cheng2023efficient}, these theorems have been applied to other settings where nuisance parameters are used for efficiency improvement but not parameter identification \citep[e.g.,][]{yu2024doubly,bao2025estimating,qian2025distal}.

In simulation studies, the proposed estimators achieve substantial efficiency gain over standard alternatives such as the WCLS estimator \citep{boruvka2018}. The gain is most pronounced when the outcome regression is highly nonlinear or when the outcome variance varies across decision points. In a real data application to the Drink Less MRT, we use the proposed estimators to assess the effect of push notifications on user engagement, finding relative efficiencies up to 2.15 against standard alternatives. The proposed method will enable domain scientists to more precisely answer critical questions regarding time-varying effect moderation with fewer samples.

Concurrent with the preprint version of our manuscript \citep{cheng2023efficient}, \citet{shi2023meta} proposed a related meta-learning approach aimed at improving the efficiency of CEE estimators. Our contribution differs in two fundamental ways. First, we derive the semiparametric efficiency bounds for CEE and use these results to motivate estimator construction; \citet{shi2023meta} do not study efficiency bounds. Second, the resulting class of estimating equations we develop is broader: the DR-WCLS estimator of \citet{shi2023meta} arises as a special case obtained by choosing possibly suboptimal weights. Our semiparametric efficiency theory indicates that the optimal weights under working assumptions should involve the inverse conditional variance of the decision-point-specific estimating functions. We further elaborate on these distinctions and connections in \cref{subsec:efficiency-comparison}.

The paper is organized as follows. \Cref{sec:notation} defines notation and the causal quantity of interest. \Cref{sec:methods} presents the semiparametric efficiency bound and the efficient estimators. \Cref{sec:theory} establishes the general asymptotic theory for globally robust Z-estimators. \Cref{sec:simulation} presents simulation studies to illustrate the robustness and improved efficiency. \Cref{sec:application} presents the real data illustration. \Cref{sec:discussion} concludes with discussion.

\section{Notation and Definition}
\label{sec:notation}

\subsection{MRT Data Structure}
\label{subsec:mrt-data-structure}

Assume an MRT has $n$ individuals, each enrolled for $T$ decision points at which treatments will be randomized. Variables without subscript $i$ represent observations from a generic individual. $A_t$ represents treatment assignment at decision point $t$, where $A_t = 1$ indicates treatment, and $A_t = 0$ indicates no treatment. $X_t$ represents observations between decision points $t-1$ and $t$. $Y_{t+1}$ denotes the proximal outcome following decision point $t$, which is typically the direct target of the intervention. An individual's trajectory of observed information is $O = (X_1,A_1,Y_2, X_2,A_2,Y_3,\ldots,X_T,A_T,Y_{T+1})$. Data from individuals are assumed to be independent, identically distributed samples from distribution $P_0$. Expectations are taken with respect to $P_0$, unless otherwise stated.

In intervention design, researchers may consider it unsafe or unethical to send push notifications at certain decision points, such as during driving. Following the MRT literature, we consider settings where an individual may be unavailable for randomization at certain decision points \citep{boruvka2018}. Formally, $X_t$ includes an indicator $I_t$, with $I_t = 1$ denoting being available for randomization at decision point $t$, and $I_t = 0$ otherwise. If $I_t = 0$, then $A_t = 0$ deterministically.

An overbar denotes a sequence of variables, e.g., $\bA_t=(A_1,A_2,\ldots,A_t)$. Information gathered from an individual up to decision point $t$ is $H_t = (X_1,A_1,Y_2, \ldots,X_{t-1},A_{t-1},Y_t, X_t)$. At each $t$, $A_t$ is randomized with probability $p_t(H_t) := P(A_t = 1 \mid H_t)$, and we sometimes notationally omit $H_t$ to write $p_t$. We use policy to refer to any decision rule of assigning $A_t$, and we refer to the particular dynamic rule of randomly assigning $A_t$ as the MRT treatment policy.

We use $\PP$ to denote the population average: $\PP(f) = \int f(o) dP(o)$, which is equivalent to the expectation $\EE$ notation when $f(\cdot)$ is non-random. We use $\PP_n$ to denote the empirical average over all individuals. For a positive integer $k$, define $[k] := \{1, 2, \ldots, k\}$. For a column vector $\alpha$, define $\alpha^{\otimes 2} := \alpha \alpha^T$. We use superscript $\star$ to denote quantities corresponding to the true data generating distribution, $P_0$. We use $\| \cdot \|$ to denote the $L_2$ norm, i.e., $\| f(O) \| = \{\int f(o)^2 dP(o)\}^{1/2}$ for any function $f$ of the observed data $O$. For a vector $\alpha$ and a vector-valued function $f(\alpha)$, $\partial_\alpha f(\alpha) := \partial f(\alpha) / \partial \alpha^T$ denotes the matrix where the $(i,j)$-th entry is the partial derivative of the $i$-th entry of $f$ with respect to the $j$-th entry of $\alpha$. When $f(\alpha)$ is scalar-valued, $\partial_\alpha f(\alpha)$ denotes the column gradient $\partial f(\alpha) / \partial \alpha \in \RR^p$.

\subsection{Potential Outcome and Causal Excursion Effect}
\label{subsec:def-cee}

We focus on the estimation of the causal excursion effect (CEE). CEE is the quantity of interest in most MRT primary and secondary analyses, because it provides a parsimonious summary for whether, when, and under what context an intervention is effective, which can be used to further optimize the intervention \citep{klasnja2018efficacy,battalio2021sense2stop,nahum2021translating,bell2023notifications}. CEE characterizes the contrast between the potential outcomes under two hypothetical treatment policies, which are called ``excursions'' from the MRT treatment policy. To define CEE precisely we use the potential outcomes framework \citep{rubin1974estimating, robins1986new}. Lower case letters represent instantiations (non-random values) of the corresponding capital letter random variable. For example, $a_t$ is an instantiation of treatment $A_t$. For every individual, denote by $X_t(\ba_{t-1})$ (or $Y_t(\ba_{t-1})$) the $X_t$ (or $Y_t$) that would have been observed at decision point $t$ if the individual were assigned a treatment sequence of $\ba_{t-1}$ prior to $t$. The potential outcome of $H_t$ under $\ba_{t-1}$ is
\begin{align*}
    H_t(\ba_{t-1}) = \{X_1, a_1, Y_2(a_1), X_2(a_1), a_2, Y_3(\ba_2), \ldots, X_{t-1}(\ba_{t-2}), a_{t-1}, Y_t(\ba_{t-1}), X_t(\ba_{t-1})\}.
\end{align*}

Let $S_t$ denote a subset of $H_t$, which are the moderators of interest based on the scientific question. For example, one can set $S_t = \emptyset$ for a fully marginal effect that averages over all moderators, $S_t = A_{t-1}$ for effect moderation by previous intervention, and $S_t = Y_t$ for effect moderation by previous outcome. Given a link function $g(\cdot)$, define the CEE of $A_t$ on $Y_{t+1}$ moderated by $S_t$ as
\begin{align}
    \cee_t \left\{S_t(\bA_{t-1}) \right\} & ~ := ~ g[\EE\{ Y_{t+1}(\bA_{t-1},1) \mid S_t(\bA_{t-1}), I_t(\bA_{t-1}) = 1 \}] \nonumber \\ & ~~~~ - g[\EE\{ Y_{t+1}(\bA_{t-1},0) \mid S_t(\bA_{t-1}), I_t(\bA_{t-1}) = 1 \}].
    \label{eq:def-cee}
\end{align}
When $g(\cdot)$ is the identity function or the log function, \eqref{eq:def-cee} agrees with the CEE for continuous outcome or binary outcome, respectively \citep{boruvka2018,qian2021estimating}.
In this paper, we restrict the link function $g(\cdot)$ to be either the identity function or the log function, but we do not restrict the type of the outcome variable. This yields a more general CEE definition than considered in \citet{boruvka2018} and \citet{qian2021estimating}: for example, \eqref{eq:def-cee} can be used for count outcomes or other non-negative outcomes. The theory and methodology will be derived in a unified way for the two link functions. Generalization to other link functions and survival outcomes is nontrivial and is left for future work; we provide a discussion in \cref{sec:discussion}.

Expression \eqref{eq:def-cee} exhibits three defining characteristics that distinguish CEE from other causal effects for time-varying treatments: (i) the excursion aspect, (ii) the marginalization aspect, and (iii) the conditioning on availability. For (i), the CEE is a contrast in the potential outcome under two ``excursions'', i.e., policies that deviate from the MRT policy. In particular, $(\bA_{t-1},1)$ and $(\bA_{t-1},0)$ represent following the MRT policy up to decision point $t-1$ then deviating to assign deterministically $A_t = 1$ and $A_t = 0$, respectively. For (ii), the CEE is marginalized over variables in $H_t \setminus \{S_t, I_t\}$, such as past treatment assignments $\bA_{t-1}$ unless they are included in $S_t$. For (iii), the CEE is conditional on a participant being available in the moment ($I_t = 1$), because randomization is only feasible at available decision points.

These characteristics make the CEE particularly suitable for MRT analysis. First, the marginalization aspect, combined with the researcher-chosen $S_t$, allows a hierarchy of analyses from the fully marginal primary analysis to secondary analyses that explore moderation by various contextual variables. Second, the CEE allows more efficient use of the data because the treatment trajectories being compared are not static. As discussed in \cref{rmk:cee-interpretation-and-use}, the CEE's marginalization over past treatment assignments also makes it depend on the MRT randomization policy. One potential criticism for CEE is that this dependence may hinder generalizability; however, the concern is alleviated by the fact that MRT policies are typically designed with real-world considerations in mind (such as availability and burden), not just for experimentation. Thus, the CEE resembles the effect of the intervention in a reasonable real-world implementation.

\begin{remark}[Interpretation and use of the CEE]
\label{rmk:cee-interpretation-and-use}
We comment further on two aspects of the CEE: its conditioning on availability ($I_t = 1$) and its dependence on the MRT randomization policy.

The availability indicator $I_t$ marks decision points at which intervention can be delivered based on feasibility or ethical considerations. In implemented MRTs, examples of $I_t = 0$ include when the participant is driving, when there is no internet connection, or when the participant has not completed a real-time survey which is required by the tailoring of the intervention \citep{klasnja2015,qian2026study}. When $I_t = 0$, no intervention is the only feasible option, so a contrast between intervention and no intervention is not of scientific interest. \citet{guo2021discussion} also pointed out that among static excursion regimes, only those that assign no intervention whenever $I_t = 0$ are identified from MRT data without additional assumptions.

The marginalization in \eqref{eq:def-cee} makes the CEE depend on the MRT randomization policy. \citet{qian2021rejoinder} provided a policy-improvement interpretation of the CEE: where the CEE at moderator value $S_t = s$ is positive, the expected outcome under the MRT policy can be increased by intervening more often at decision points where the moderator equals $s$; where the CEE is negative, by intervening less often. The sign and magnitude of the CEE thus directly indicate how to improve the policy locally. The formal version of this interpretation, in which the CEE is interpreted as a directional derivative of the expected outcome with respect to the policy, is detailed in \citet{qian2021rejoinder}. This makes the CEE useful for iteratively improving digital interventions: an MRT is run under an initial policy, the CEE and its moderation by $S_t$ are estimated, and the moderation analysis results guide the design of the next iteration of the implementation. For example, in designing the reinforcement-learning algorithm used in HeartSteps V2, \citet{liao2020personalized} used moderation analysis from the prior HeartSteps V1 MRT to select the contextual moderators that enter the treatment-effect model in the policy and to form informative Gaussian priors on the treatment-effect coefficients. The dependence of a marginalized estimand on the design distribution is a familiar phenomenon in optimization trials: for example, in factorial designs, when treatment-by-treatment interactions are present, the marginal main effect of one factor depends on the distribution over the levels of the other factors.
\end{remark}

CEE can be viewed as the marginalization of a treatment blip in a structural nested mean model \citep[SNMM,][]{robins1994snmm}, and is related to the optimal SNMM with a coarsened history \citep{robins2004optimal}, except that the latter considers increasing, nested conditional sets. This marginal aspect is also related to two variations of marginal structural models (MSM) for the mean potential outcome: the history-restricted MSM \citep{neugebauer2007causal} and the history-adjusted MSM \citep{joffe2001partially,van2005history,petersen2007history}, where the former considers all possible future treatment trajectories and the latter further considers all possible time window lengths. \citet{guo2021discussion} provides an illustration for the connections and differences between CEE and other causal effects with time-varying treatments.

To express the CEE using observed data, we make standard causal assumptions.
\begin{assumption}
	\label{asu:causal-assumptions}
	\begin{asulist}
		\item \label{asu:consistency} (SUTVA.) There is no interference across individuals and the observed data equals the potential outcome under the observed treatment. Specifically, $X_t = X_t(\bA_{t-1})$ and $Y_{t+1} = Y_{t+1}(\bA_t)$ for $t\in[T]$.
		\item \label{asu:positivity} (Positivity.) There exists a positive constant $\tau > 0$, such that if $P(H_t = h_t, I_t = 1) > 0$ then $\tau < P(A_t = a \mid H_t = h_t, I_t = 1) < 1-\tau$ for $a \in \{0, 1\}$.
		\item \label{asu:sequential-ignorability} (Sequential ignorability.) For $1 \leq t \leq T$, the potential outcomes $\{X_{t+1}(\ba_t)$, $X_{t+2}(\ba_{t+1})$, $\ldots, X_{T+1}(\ba_T):$ $\ba_T \in \{0,1\}^T\}$ are independent of $A_t$ conditional on $H_t$.
	\end{asulist}
\end{assumption}
Positivity and sequential ignorability are guaranteed by the MRT design. SUTVA is violated if interference is present, i.e., if the treatment assigned to one participant affects the potential outcome of another participant. In such settings, a framework that incorporates causal interference is needed \citep{hudgens2008toward,shi2022assessing}. We do not consider such settings here. Under \cref{asu:causal-assumptions}, the CEE can be expressed as
\begin{align}
	\cee_t (S_t) & = g[\EE\{ \EE(Y_{t+1} \mid H_t, A_t = 1) \mid S_t, I_t = 1\}] \nonumber \\
	& ~~~~~~~~~~~~~~~~ - g[\EE\{ \EE(Y_{t+1} \mid H_t, A_t = 0) \mid S_t, I_t = 1\}]. \label{eq:identify-cee}
\end{align}

\section{Robust and Efficient Estimators for CEE}
\label{sec:methods}

\subsection{The Efficiency Bound Under Working Assumptions}
\label{subsec:efficiency-bound}

Consider the estimation of an unknown parameter $\beta$ whose true value $\beta^\star \in \Theta \subset \RR^p$ satisfies
\begin{align}
	& g[\EE\{ \EE(Y_{t+1} \mid H_t, A_t = 1)  \mid S_t, I_t = 1\}] - g[\EE\{ \EE(Y_{t+1} \mid H_t, A_t = 0) \mid S_t, I_t = 1\}] \nonumber \\
	& \qquad \qquad \qquad = \gamma_t(S_t; \beta^\star) \label{eq:linear-model-cee}
\end{align}
for $t \in [T]$, where $\gamma_t(S_t; \beta)$ is a pre-specified function. This model allows for linear and nonlinear effects; for example, $\gamma_t(S_t; \beta)$ may include a linear combination of basis functions of $t$ or $S_t$. Estimating $\beta^\star$ is the focus of primary and secondary analyses for MRTs \citep{klasnja2018efficacy,klasnja2020barifit,nahum2021translating,bell2023notifications}.

Let $U_t(\beta) := Y_{t+1} - A_t \gamma_t(S_t;\beta)$ if $g$ is the identity link and $U_t(\beta) := Y_{t+1}\exp\{-A_t \gamma_t(S_t;\beta)\}$ if $g$ is the log link. $U_t$ is similar to the exposure-free outcome in the SNMM literature \citep{robins1994snmm,vansteelandt2003causal}, except that here we are removing only the immediate marginal effect from $Y_{t+1}$.
\Cref{thm:efficiency-bound} gives the semiparametric efficiency bound for $\beta$ under working assumptions, with proof in \cref{A-subsec:proof-thm-efficiency-bound}.

\begin{theorem}
	\label{thm:efficiency-bound}
	Consider the semiparametric model \eqref{eq:linear-model-cee}. Let $\psi(\beta) = \sum_{t=1}^T d_t^\star(S_t) \phi_t(\beta)$ with
	\begin{align}
		& \phi_t(\beta) := \frac{A_t - p_t}{p_t(1-p_t)} I_t  \nonumber \\
		& ~~~~~~~~~~~~~ \times \Big[U_t(\beta) - (1-p_t)\EE\{ U_t(\beta) \mid H_t, A_t = 1 \} - p_t\EE\{ U_t(\beta) \mid H_t, A_t = 0 \}\Big] \partial_\beta \gamma_t(S_t;\beta), \label{eq:phi-definition} \\
		& d_t^\star(S_t) = \EE\{\partial_\beta \phi_t(\beta^\star) \mid S_t \} ~ [\EE\{ \phi_t(\beta^\star)^{\otimes 2} \mid S_t \}]^{-1}. \label{eq:dt-definition}
	\end{align}
	Consider the following two working assumptions:
	\begin{align*}
		& \text{(WA-1)} \quad \{S_t, I_t, A_t, p_t(H_t), Y_{t+1}\} \perp A_u \mid H_u \quad \text{for all } 1 \leq u < t \leq T, \\
		& \text{(WA-2)} \quad \EE\{ \phi_t(\beta^\star) \phi_u(\beta^\star)^T\} = 0 \quad \text{for all } 1 \leq u < t \leq T.
	\end{align*}
	If it happens that (WA-1) and (WA-2) hold, then the semiparametric efficiency bound for $\beta^\star$ is $\EE \{\partial_\beta \psi(\beta^\star)\} [\EE\{ \psi(\beta^\star)^{\otimes 2} \}]^{-1} \EE \{\partial_\beta \psi(\beta^\star)\}^T$, the inverse of the variance matrix of the semiparametric efficient score.
\end{theorem}

\begin{remark}
\label{rmk:phi-psi-interpretation}
	The function $\phi_t(\beta)$ is a decision-point-$t$ estimating-function contribution restricted to available decision points (via the indicator $I_t$). It combines (i) treatment centering through $(A_t - p_t)$, which has mean zero given $H_t$; and (ii) an augmentation term, where a weighted average of the conditional means of $U_t(\beta)$ is subtracted to reduce estimation variance. The weights on the two treatment-specific conditional means are $p_t^{-1}$ and $(1-p_t)^{-1}$ after absorbing the common denominator $p_t(1-p_t)$, and this variance-minimization combination puts more weight on the rarer treatment level's conditional mean. A similarly weighted augmentation was used by \citet{tsiatis2008covariate} in the randomized controlled trial context. The pooled estimating function $\psi(\beta) = \sum_{t=1}^T d_t^\star(S_t) \phi_t(\beta)$ aggregates these per-decision-point contributions with the data-dependent weight $d_t^\star(S_t)$, which down-weights decision points whose $\phi_t$ are noisy given $S_t$ and up-weights those that are informative for $\beta$.
\end{remark}

\begin{remark}
	\Cref{thm:efficiency-bound} presents a tractable efficiency bound when working assumptions happen to hold. A more general efficiency bound without the working assumptions may potentially be derived by combining the multinomial distribution approximation technique and Gram-Schmidt orthogonalization \citep{chamberlain1987asymptotic,chamberlain1992comment}. However, we found the orthogonalization step not straightforward here due to the non-nested $\{S_t\}_{t\in[T]}$. Furthermore, we conjecture that the form of the general efficiency bound will involve across-time correlation for many $(t_1,t_2)$ combinations, and the need to model these correlations will make the corresponding estimator implementation challenging with small to moderate sample sizes.
\end{remark}

\begin{remark}
\label{rmk:sufficient-conditions-WA}
	The primary purpose of \cref{thm:efficiency-bound} is to motivate the proposed estimators in \cref{subsec:proposed-estimator}. The asymptotic normality and global robustness of those estimators (\cref{thm:normality-ncf,A-thm:normality-cf}) and their empirical efficiency gains (\cref{sec:simulation}) do \emph{not} require (WA-1) or (WA-2). The discussion below clarifies when $\EE \{\partial_\beta \psi(\beta^\star)\} [\EE\{ \psi(\beta^\star)^{\otimes 2} \}]^{-1} \EE \{\partial_\beta \psi(\beta^\star)\}^T$ is the efficiency bound, not when the proposed estimator is applicable.

	Consider the following three conditions on the data-generating process and the working model:
	\begin{itemize}[leftmargin=2.4em]
	\item[(C1)] (No carryover effects.) For each $t$, $X_t(\ba_{t-1}) = X_t$ does not depend on $\ba_{t-1}$, and $Y_t(\ba_{t-1}) = Y_t(a_{t-1})$ depends on $\ba_{t-1}$ only through $a_{t-1}$.
	\item[(C2)] (Randomization probability not depending on past treatment.) $p_t(H_t)$ depends on $H_t$ only through $\bX_t = (X_1, \ldots, X_t)$.
	\item[(C3)] (Moderator specified through covariate process capturing all moderation.) $S_t$ is a function of $\bX_t = (X_1, \ldots, X_t)$ alone, and for almost every $H_t$ with $I_t = 1$,
	\[
	g\{\EE(Y_{t+1} \mid H_t, A_t = 1)\} - g\{\EE(Y_{t+1} \mid H_t, A_t = 0)\} = \gamma_t(S_t; \beta^\star).
	\]
	\end{itemize}
	C1, C2, and C3 jointly imply both (WA-1) and (WA-2); we prove this for both the identity and log links in \cref{A-subsec:sufficient-conditions-WA}.

	In MRTs, C1 holds when the intervention acts in the moment so that past treatments do not alter later covariates or outcomes; C1 is violated when the intervention is expected to induce persistent state changes such as engagement, fatigue, habituation, or learning. C2 holds whenever the randomization probability is a function of covariates only or a constant; it is violated by adaptive randomization in which $p_t$ depends on past treatments $\bA_{t-1}$. C3 holds when (i) $S_t$ is a function of context covariates, and (ii) $S_t$ is rich enough to capture all heterogeneity in the contrast $g\{\EE(Y_{t+1} \mid H_t, A_t = 1)\} - g\{\EE(Y_{t+1} \mid H_t, A_t = 0)\}$; it is violated when $S_t$ involves past treatments (e.g., $S_t = A_{t-1}$), or when $S_t$ is chosen for parsimony and omits variables in $H_t$ that modify the effect. A more detailed discussion is given in \cref{A-subsec:sufficient-conditions-WA}.

\end{remark}

\subsection{Estimators Achieving The Efficiency Bound}
\label{subsec:proposed-estimator}

Before describing the estimators, we introduce additional notation to emphasize the dependence of the estimating function on nuisance parameters. Let $\mu_t^\star(H_t, A_t) = \EE(Y_{t+1} \mid H_t, A_t, I_t = 1)$. $\phi_t(\beta)$ defined in \eqref{eq:phi-definition} equals
\begin{align*}
	\phi_t(\beta) = \begin{cases}
		\frac{A_t - p_t}{p_t(1-p_t)} I_t \Big\{Y_{t+1} - (A_t + p_t - 1)\gamma_t(S_t;\beta) - (1-p_t)\mu_t^\star(H_t, 1) \\ ~~~~~~~~~~~~~~~ - p_t \mu_t^\star(H_t, 0) \Big\} \partial_\beta \gamma_t(S_t;\beta) & \text{if $g$ is identity}, \\
		\frac{A_t - p_t}{p_t(1-p_t)} I_t \Big\{e^{-A_t \gamma_t(S_t;\beta)}Y_{t+1} - (1-p_t)e^{-\gamma_t(S_t;\beta)}\mu_t^\star(H_t, 1) \\~~~~~~~~~~~~~~~ - p_t \mu_t^\star(H_t, 0) \Big\} \partial_\beta \gamma_t(S_t;\beta) & \text{if $g$ is log}.
	\end{cases} \nonumber
\end{align*}
We use $\phi_t(\beta, \mu_t)$ to denote the above $\phi_t(\beta)$ with $\mu_t^\star$ replaced by some function $\mu_t(h_t, a_t)$. In addition, define
\begin{align}
	d_t(S_t; \mu_t, \beta) = \EE\{\partial_\beta \phi_t(\beta, \mu_t) \mid S_t \} ~ [\EE\{ \phi_t(\beta, \mu_t)^{\otimes 2} \mid S_t \}]^{-1}. \label{eq:dt-definition-working}
\end{align}
It follows from \eqref{eq:dt-definition} that $d_t(S_t; \mu_t^\star, \beta^\star) = d_t^\star(S_t)$. We construct an estimator $\hat\beta$ for $\beta$ in \cref{algo:estimator-ncf}. A cross-fitted variant $\check\beta$, which removes the Donsker condition required by $\hat\beta$ and thereby admits a wider class of (possibly non-Donsker) machine-learning methods for the nuisance estimators, is presented in \cref{A-subsec:crossfit-estimator} of the Appendix. Additional details about the implementation, such as the modeling choices of $\hat\mu_t$ and $\hat{d}_t$, are included in \cref{A-sec:implementation}.

\begin{algorithm}
	\caption{Estimator $\hat\beta$}
	\label{algo:estimator-ncf}
	\textbf{Step 1:} Obtain an estimator $\hat\mu_t(H_t, A_t)$ of $\EE(Y_{t+1} \mid H_t, A_t, I_t = 1)$ for $t \in [T]$.

	\textbf{Step 2:} Obtain $\betainithat$ by solving $\PP_n \{\sum_{t=1}^T \phi_t (\beta, \hat\mu_t)\} = 0$ for $\beta$.

	\textbf{Step 3:} Obtain estimators of $\EE\{\partial_\beta \phi_t(\betainithat, \hat\mu_t) \mid S_t \}$ and $\EE\{ \phi_t(\betainithat, \hat\mu_t)^{\otimes 2} \mid S_t \}$ for $t \in [T]$, and plug them into \eqref{eq:dt-definition-working} to obtain $\hat{d}_t(S_t; \hat\mu_t, \betainithat)$.

	\textbf{Step 4:} Obtain $\hat\beta$ by solving $\PP_n \sum_{t=1}^T \hat{d}_t(S_t; \hat\mu_t, \betainithat) \phi_t (\beta, \hat\mu_t) = 0$ for $\beta$.

	\medskip
	\hrule
	\smallskip
	{\footnotesize\textit{Implementation note.} The estimators in Steps~1 and~3 are each indexed by $t \in [T]$. They may be obtained by fitting separately for each $t$, or by pooling across $t$ (e.g., a single regression with $t$ entered as a covariate). Pooling borrows strength across decision points and is typically more scalable when $T$ is large; per-$t$ fitting may be preferable when the conditional structure differs markedly across $t$ or with small $T$.}
\end{algorithm}

The asymptotic normality of $\hat\beta$ is established in \cref{thm:normality-ncf} below, with proof in \cref{A-subsec:normality-continuous}. The analogue for the cross-fitted $\check\beta$ is given in \cref{A-thm:normality-cf} of the Appendix and yields the same asymptotic variance.

\begin{theorem}[Asymptotic normality of $\hat\beta$]
	\label{thm:normality-ncf}
	Suppose \eqref{eq:linear-model-cee} holds. In addition, suppose the following hold.
	\begin{enumerate}
		\item[(a)] For each $t \in [T]$, $\hat\mu_t$ converges to some limit in $L_2$. In other words, there exists a function $\mu_t'(h_t, a_t)$ such that $\|\hat\mu_t - \mu_t'\| = o_P(1)$.
		\item[(b)] For each $t \in [T]$, $\hat{d}_t(s_t; \mu_t', \beta^\star)$ converges to some limit in $L_2$. In other words, there exists a function $d_t'(s_t)$ such that $\|\hat{d}_t - d_t'\| = o_P(1)$.
		\item[(c)] Given the limit functions $\{\mu_t'(h_t, a_t): t \in [T]\}$ and $\{d_t'(s_t): t \in [T]\}$, \[ \PP \sum_{t=1}^T d_t'(S_t) \phi_t(\beta, \mu_t') = 0 \] as a function of $\beta$ has a unique zero.
		\item[(d)] Suppose the parameter space $\Theta$ of $\beta$ is compact, the support of $O$ is bounded, $\hat\mu_t$ and $\mu_t'$ are bounded for each $t \in [T]$, and $\PP\{\sum_{t=1}^T d_t'(S_t) \partial_\beta \phi_t(\beta^\star, \mu_t') \}$ is invertible.
		\item[(e)] For each $t$, the estimators $\hat\mu_t$ and $\hat{d}_t$ take values in a Donsker class.
	\end{enumerate}
	Conditions (a)--(e) are informal restatements of the formal assumptions in \cref{A-subsec:asu-continuous}.
	Then we have $\sqrt{n}(\hat\beta - \beta^\star) \dto N ( 0, V)$ as $n\to\infty$, where $V$ equals
	\begin{align}
		\PP\Big\{\sum_{t=1}^T d_t'(S_t) \partial_\beta \phi_t(\beta^\star, \mu_t') \Big\}^{-1} \PP \Big[ \Big\{ \sum_{t=1}^T d_t'(S_t) \phi_t(\beta^\star, \mu_t') \Big\}^{\otimes 2} \Big] \PP\Big\{\sum_{t=1}^T d_t'(S_t) \partial_\beta \phi_t(\beta^\star, \mu_t') \Big\}^{-1, T},\label{eq:avar}
	\end{align}
	and $V$ can be consistently estimated by
	\begin{align*}
		& ~~ \bigg[\PP_n\Big\{\sum_{t=1}^T \hat{d}_t(S_t; \hat\mu_t, \betainithat) \partial_\beta \phi_t(\hat\beta, \hat\mu_t) \Big\}\bigg]^{-1} \PP_n \Big[ \Big\{ \sum_{t=1}^T \hat{d}_t(S_t; \hat\mu_t, \betainithat) \phi_t(\hat\beta, \hat\mu_t) \Big\}^{\otimes 2} \Big] \\
		& \times \bigg[\PP_n\Big\{\sum_{t=1}^T \hat{d}_t(S_t; \hat\mu_t, \betainithat) \partial_\beta \phi_t(\hat\beta, \hat\mu_t) \Big\}\bigg]^{-1, T}.
	\end{align*}
\end{theorem}

\begin{corollary}
  \label{cor:achieving-efficiency-bound}
  If $\hat\mu_t(h_t, a_t)$ and $\hat{d}_t(s_t; \mu_t^\star, \beta^\star)$ converge to $\mu_t^\star(h_t, a_t)$ and $d_t^\star(s_t)$ in $L_2$, respectively, then $\hat\beta$ achieves the efficiency bound given in \cref{thm:efficiency-bound}. The same conclusion holds for $\check\beta$ when $\hat\mu_{kt}$ and $\hat{d}_{kt}$ converge in $L_2$ to $\mu_t^\star$ and $d_t^\star$ for all $k \in [K]$.
\end{corollary}

\begin{remark}[Robustness]
	The estimator $\hat\beta$ is robust to misspecification of the nuisance parameter models in that the consistency and asymptotic normality do not require the nuisance parameter estimators $\hat\mu_t$ and $\hat{d}_t$ to converge to $\mu_t^\star$ and $d_t^\star$. Furthermore, there is no requirement on the rate of convergence for the nuisance parameter estimators, which allows a wider range of nonparametric methods to be used such as kernel or nearest neighbor estimators. This is less stringent than the literature on doubly-robust average treatment effect estimation, where a faster than $n^{-1/4}$ rate is typically required for nuisance parameter estimators \citep{van2000asymptotic,kennedy2016semiparametric}. The reason is that the estimating equations we constructed are ``globally robust'' (see \cref{sec:theory}) and the randomization probabilities are known in MRTs.
\end{remark}

\begin{remark}[Double robustness in observational settings]
\label{rmk:double-robustness-observational}
The estimator is applicable to observational settings, where $p_t(H_t)$ is unknown and must be estimated. In that case $\hat\beta$ is doubly-robust: $\hat\beta$ remains consistent if for each $t \in [T]$ at least one of the two nuisance models, $p_t$ and $\mu_t$, is correctly specified, where the choice of which model is correct may differ across $t$. We prove this in \cref{A-sec:dr-observational}. When nuisance models are fitted nonparametrically, asymptotic normality is guaranteed under the standard double-robustness rate condition $\|\hat p_t-p_t^\star\|\cdot\|\hat\mu_t-\mu_t^\star\|=o_P(n^{-1/2})$ \citep{chernozhukov2018double,kennedy2016semiparametric}.
\end{remark}

\begin{remark}[Probability limit of the estimators when CEE model is wrong]
	\label{rmk:dt-determines-projection}
	The CEE model \eqref{eq:linear-model-cee} itself may be misspecified: for example, one may set $S_t = \emptyset$ and $\gamma_t(S_t;\beta) = \beta_0$ for $t\in[T]$ in order to estimate a single marginal effect, but the truth CEE may vary across $t$. In such settings, the form of $\hat d_t$ in \cref{algo:estimator-ncf} affects the probability limit of $\hat\beta$ (and $\check\beta$). In the case of misspecified CEE, setting $\hat d_t = 1$ for $t\in[T]$ results in an average-over-time CEE as the probability limit. Modeling $\hat d_t$ differently may result in different probability limits. Therefore, assessing whether different choices of $\hat d_t$ lead to drastically different $\hat\beta$ is a way to assess whether the CEE model is misspecified. If CEE misspecification is a concern, setting $\hat d_t = 1$ in Step 3 of \cref{algo:estimator-ncf} can lead to a more interpretable estimand while still achieving efficiency improvement over WCLS and EMEE. Specifically, with $\hat d_t = 1$, $\hat\beta$ targets the $L_2$ projection of $\{\cee_t(S_t)\}_{t=1}^T$ onto $\gamma_t(S_t;\beta)$ with equal weights across $t$.
\end{remark}

\subsection{Comparison with Existing Estimators}
\label{subsec:efficiency-comparison}

The proposed estimating equation $\PP_n \sum_{t=1}^T d_t(S_t;\mu_t,\beta)\,\phi_t(\beta,\mu_t) = 0$ unifies and extends several existing estimators for CEE in MRTs. \Cref{tab:estimator-comparison} positions the proposed estimators ($\hat\beta$, $\check\beta$) alongside DR-WCLS, DR-EMEE, WCLS, EMEE, and the unadjusted estimator as special cases of this framework; \cref{A-subsec:table-equivalences} verifies each row.

\begin{table}[htbp]
\centering
\footnotesize
\caption{Comparison of CEE estimators as special cases of the proposed framework, $\PP_n \sum_{t=1}^T d_t(S_t;\mu_t,\beta)\,\phi_t(\beta,\mu_t) = 0$, with $\phi_t(\beta,\mu_t) = \{(A_t-p_t)/[p_t(1-p_t)]\}\, I_t\, R_t\, \partial_\beta\gamma_t(S_t;\beta)$. Estimators differ only in the inner residual $R_t$ and the cross-time weight $d_t(S_t)$. Here $\gamma_t = \gamma_t(S_t;\beta)$, $\mu_t(H_t,a)$ is an outcome regression working model, $c_t(H_t)^\top\alpha$ is a parametric working model linear in the control variables $c_t(H_t)$, $\tilde p_t = \tilde p_t(1\mid S_t)$ is a numerator probability for stabilized weights, and $d_t^\star(S_t)$ is the variance-minimizing weight in~\eqref{eq:dt-definition}. WCLS and EMEE rows correspond to the immediate-effect versions in the original papers (lag $k=1$ and $\Delta=1$, respectively).}
\label{tab:estimator-comparison}
\setlength{\tabcolsep}{2pt}
\begin{tabular}{@{}lll@{}}
\toprule
Estimator & Inner residual $R_t$ in $\phi_t$ & Weight $d_t(S_t)$ \\
\midrule
\multicolumn{3}{@{}l}{\textit{Identity link}} \\
Proposed ($\hat\beta$, $\check\beta$) & $Y_{t+1} - (A_t+p_t-1)\gamma_t - (1-p_t)\mu_t(H_t,1) - p_t\mu_t(H_t,0)$ & $d_t^\star(S_t)$ \\
DR-WCLS \citep{shi2023meta} & same as Proposed & $\tilde p_t(1-\tilde p_t)$ \\
WCLS \citep{boruvka2018} & $Y_{t+1} - (A_t - \tilde p_t)\gamma_t - c_t(H_t)^\top\alpha$ & $\tilde p_t(1-\tilde p_t)$ \\
Unadjusted ($\hat\beta_\text{unadj}$) & $Y_{t+1} - A_t \gamma_t$ & $1$ \\[3pt]
\multicolumn{3}{@{}l}{\textit{Log link}} \\
Proposed ($\hat\beta$, $\check\beta$) & $e^{-A_t \gamma_t}\, Y_{t+1} - (1-p_t)e^{-\gamma_t}\mu_t(H_t,1) - p_t\mu_t(H_t,0)$ & $d_t^\star(S_t)$ \\
DR-EMEE \citep{liu2024incorporating} & same as Proposed & $\tilde p_t(1-\tilde p_t)$ \\
EMEE \citep{qian2021estimating} & $e^{-A_t \gamma_t}\, Y_{t+1} - e^{c_t(H_t)^\top\alpha}$ & $\tilde p_t(1-\tilde p_t)$ \\
Unadjusted ($\hat\beta_\text{unadj}$) & $e^{-A_t \gamma_t}\, Y_{t+1}$ & $1$ \\
\bottomrule
\end{tabular}
\end{table}

\citet{shi2023meta} proposed DR-WCLS for continuous outcomes within a meta-learning framework, and \citet{liu2024incorporating} proposed DR-EMEE for binary and count outcomes (log link). Both use the same nuisance combination as ours, but with a heuristic decision-point weight $d_t(S_t) = \tilde p_t(1-\tilde p_t)$, where $\tilde p_t = \tilde p_t(1\mid S_t)$ is a numerator probability for stabilized inverse probability weights. Our weight $d_t^\star(S_t)$ depends on the inverse conditional variance of $\phi_t$ given $S_t$ and generally differs from this heuristic choice.

\citet{boruvka2018} proposed WCLS for continuous outcomes (with general lag $k$), and \citet{qian2021estimating} proposed EMEE for binary outcomes (with general window length $\Delta$); we focus the comparison on the immediate-effect special cases ($k=1$ for WCLS, $\Delta=1$ for EMEE) since those are the targets of our framework. Both use the same heuristic weight $d_t = \tilde p_t(1-\tilde p_t)$ as DR-WCLS/DR-EMEE, and they replace our augmentation in $\phi_t$ with a parametric control $c_t(H_t)^T\alpha$, linear for the identity link and log-linear for the log link (denoted $g_t(H_t)^T\alpha$ in \citealp{boruvka2018} and \citealp{qian2021estimating}). Our augmentation arises as the projection of $U_t(\beta^\star)$ onto the score for the treatment-assignment mechanism (see the proof of \cref{thm:efficiency-bound} in \cref{A-subsec:proof-thm-efficiency-bound}), so plugging in $\mu_t = \mu_t^\star$ minimizes the variance of $\phi_t$ over all $H_t$-measurable augmentations.

A natural question is whether covariate adjustment via the proposed estimator can ever inflate asymptotic variance relative to an unadjusted baseline, in analogy with classical covariate-adjustment results for randomized clinical trials \citep{yang2001efficiency,tsiatis2008covariate,lin2013agnostic,ye2023anhecova}. Define the unadjusted estimator $\hat\beta_\text{unadj}$ as the solution to $\PP_n \sum_{t=1}^T I_t \frac{A_t - p_t}{p_t(1-p_t)} U_t(\beta) \, \partial_\beta \gamma_t(S_t;\beta) = 0$, namely the proposed estimating equation with augmentation set to zero and $d_t \equiv 1$. \Cref{thm:efficiency-comparison} below shows that, for the marginal CEE under the assumption of no carryover effects, $\hat\beta$ has no larger asymptotic variance than $\hat\beta_\text{unadj}$, even when the working outcome model is misspecified.

\begin{theorem}[Asymptotic variance improvement under misspecified $\mu_t$]
\label{thm:efficiency-comparison}
Consider the estimation of the marginal CEE, $\gamma_t(S_t;\beta) = \beta_0 \in \RR$, with either the identity or the log link. Suppose $p_t = p \in (0,1)$ (constant randomization probability) and $I_t \equiv 1$ (always available). For the proposed $\hat\beta$ in \cref{algo:estimator-ncf}, suppose $\hat\mu_t$ is fitted by ordinary least squares of $Y_{t+1}$ on the linear working model
$\mathcal{M}_t = \bigl\{ c_1 + c_2 A_t + g_1(H_t) + A_t\, g_2(H_t) \::\: c_1, c_2 \in \RR,\; g_1, g_2 \in \mathcal{G}_t \bigr\}$,
with $\mathcal{G}_t$ a linear space of $H_t$-measurable functions, and $\mathcal{M}_t$ possibly not containing $\mu_t^\star$.
Suppose there is no carryover effect: for each $t$, $X_t(\ba_{t-1}) = X_t$ does not depend on $\ba_{t-1}$, and $Y_t(\ba_{t-1}) = Y_t(a_{t-1})$ depends on $\ba_{t-1}$ only through $a_{t-1}$.
Then under the regularity conditions of \cref{thm:normality-ncf}, $V(\hat\beta) \le V(\hat\beta_\text{unadj})$, where $V(\hat\beta)$ and $V(\hat\beta_\text{unadj})$ denote the asymptotic variances of $\hat\beta$ and $\hat\beta_\text{unadj}$.
\end{theorem}

The proof of \cref{thm:efficiency-comparison} is given in \cref{A-subsec:efficiency-comparison-proof}.

\begin{remark}[Connection with covariate adjustment in randomized trials]
\label{rmk:rct-no-harm-connection}
\Cref{thm:efficiency-comparison} is the MRT analogue of a classical result for single-time-point randomized trials: in a randomized trial, adjusting for baseline covariates via a linear working model never asymptotically inflates the variance of the treatment effect estimator, even when the working model is misspecified. \citet{yang2001efficiency,tsiatis2008covariate,lin2013agnostic} established this in the ANCOVA framework, and \citet{ye2023anhecova} extended it to stratified randomization. \Cref{thm:efficiency-comparison} shows that the same guarantee for the MRT setting holds when there are no carryover effects, for both the identity and log links.
\end{remark}

\begin{remark}[Necessity of the no-carryover assumption]
\label{rmk:no-harm-fails-without-c1}
When the no-carryover assumption is dropped, \Cref{thm:efficiency-comparison} does not hold generally. In \cref{A-subsec:counter-example-carryover} we construct a $T = 2$ MRT in which $A_1$ has a delayed effect on $Y_3$, and the proposed estimator $\hat\beta$ (with the optimal weight $d_t^\star$ and the true conditional mean $\mu_t^\star$), DR-WCLS, and WCLS all have strictly larger asymptotic variance than $\hat\beta_\text{unadj}$.
\end{remark}

\section{Asymptotic Theory for Globally Robust Z-Estimators}
\label{sec:theory}

We establish a collection of easy-to-use asymptotic theory for a general class of two-stage estimators that encompasses the estimators in \cref{sec:methods} but also other estimators that arise in randomized and micro-randomized trials. Consider an estimating function $m(\theta, \eta)$, where $\theta \in \Theta$ is a finite-dimensional parameter of interest with an unknown true value $\theta^\star$, and $\eta \in \cT$ is an infinite-dimensional nuisance parameter. We consider a two-stage estimator for $\theta$, where in the first stage an estimator $\hat\eta$ for the nuisance parameter $\eta$ is obtained, and in the second stage $\PP_n\{ m(\theta, \hat\eta) \} = 0$ is solved to obtain $\hat\theta$. The $\hat\beta$ in \cref{algo:estimator-ncf} is such an estimator with Steps 1-3 being the first stage and Step 4 being the second stage. Later we will also consider a cross-fitting version of the two-stage estimator, of which $\check\beta$ in \cref{A-algo:estimator-cf} is a special case.

We focus on the scenario where the estimating function $m(\theta, \eta)$ is globally robust:

\begin{definition}
	\label{def:globally-robust}
	$m(\theta, \eta)$ is globally robust if $\PP\{ m(\theta^\star, \eta) \} = 0$ for all $\eta \in \cT$.
\end{definition}

\begin{remark}
	Globally robust estimating functions are a subset of Neyman orthogonalized estimating functions---the latter require that the Gateaux derivative of $\PP\{ m(\theta^\star, \eta) \}$ with respect to $\eta$ evaluates to 0 at some $\eta^\star$ \citep{chernozhukov2018double}. Globally robust estimating functions often arise when the nuisance function $\eta$ is not needed for identifiability and is only included to increase precision \citep{lok2024estimating}. The estimating function $\sum_{t=1}^T d_t(S_t; \mu_t, \beta) \phi_t(\beta, \mu_t)$ in \cref{sec:methods} is such an example, where both $d_t$ and $\mu_t$ are included only to increase precision and can be arbitrarily misspecified without hindering the asymptotic normality of the estimator for $\beta$ (\cref{thm:normality-ncf,A-thm:normality-cf}). The global robustness of $\sum_{t=1}^T d_t(S_t; \mu_t, \beta) \phi_t(\beta, \mu_t)$ is proven in \cref{A-sec:general-lemmas}. In the MRT literature, WCLS \citep{boruvka2018}, EMEE \citep{qian2021estimating}, DR-WCLS \citep{shi2023meta}, and the distal CEE estimator \citep{qian2025distal} are all globally robust in MRTs with known randomization probabilities. Other notable examples include regression adjustment estimators for randomized experiments \citep{lin2013agnostic}.
\end{remark}

We make the following assumptions for establishing the asymptotic normality of the two-stage estimator $\hat\theta$.

\begin{assumption}
	\label{asu:unique-zero}
	For any $\eta \in \cT$, $\theta^\star$ is the unique solution to $\PP\{m(\theta,\eta)\} = 0$.
\end{assumption}

\begin{assumption}
	\label{asu:nuisance-conv-general}
	Suppose there exists $\eta' \in \cT$ such that the following hold.

	\begin{asulist}
		\item \label{asu:nuisance-conv-PPee-sup} $\sup_{\theta \in \Theta} | \PP m(\theta, \hat\eta) - \PP m(\theta, \eta') | = o_P(1)$;
		\item \label{asu:nuisance-conv-ee-l2} $\| m(\theta^\star, \hat\eta) - m(\theta^\star, \eta') \|^2  := \int |m(\theta^\star, \hat\eta) - m(\theta^\star, \eta')|^2 dP = o_P(1)$;
		\item \label{asu:nuisance-conv-PPee-deriv} $| \PP \{\partial_\theta m(\theta^\star, \hat\eta)\} - \PP \{\partial_\theta m(\theta^\star, \eta')\} | = o_P(1)$;
		\item \label{asu:nuisance-conv-PPee-meat} $| \PP \{m(\theta^\star, \hat\eta) m(\theta^\star, \hat\eta)^T\} - \PP \{m(\theta^\star, \eta') m(\theta^\star, \eta')^T \} | = o_P(1)$.
	\end{asulist}
\end{assumption}

Rather than specifying the convergence in terms of $\hat\eta$ itself, \cref{asu:nuisance-conv-general} specifies the convergence of $\hat\eta$ in terms of functions of $m(\theta, \eta)$. Notably, as we establish in \cref{thm:general-normality} below, as long as the convergence rates are $o_P(1)$, the fact that $\hat\eta$ is estimated does not affect the asymptotic variance. In fact, the two-step estimator $\hat\theta$ is asymptotically as efficient as its counterpart that solves $\PP_n \{ m(\theta, \eta')\} = 0$, i.e., with the limit $\eta'$ plugged into the estimating equation. 

\begin{assumption}
	\label{asu:reg-general}
	Suppose the following hold: (i) The parameter space $\Theta$ of $\theta$ is compact. (ii) The support of $O$ is bounded. (iii) $\PP\{m(\theta, \eta')\}$ is a continuous function in $\theta$. (iv) Suppose $m(\theta,\eta)$ is continuously differentiable in $\theta$, and the class $\{m(\theta,\eta): \theta \in \Theta, \eta \in \cT\}$ is uniformly bounded and bounded by an integrable function. (v) $\partial_\theta m(\theta,\eta)$ is bounded by an integrable function. (vi) $m(\theta,\eta) m(\theta,\eta)^T$ is bounded by an integrable function. (vii) $\PP\{\partial_\theta m(\theta^\star,\eta')\}$ is invertible. (viii) $\partial_\theta m(\theta,\eta)$ and $m(\theta,\eta) m(\theta,\eta)^T$ are uniformly bounded.
\end{assumption}

\begin{assumption}
	\label{asu:donsker}
	Suppose $\{m(\theta, \eta): \theta \in \Theta, \eta \in \cT\}$ and $\{\partial_\theta m(\theta,\eta): \theta \in \Theta, \eta \in \cT\}$ are $P$-Donsker classes.
\end{assumption}

\Cref{asu:unique-zero} implies that $m(\theta,\eta)$ is globally robust. \Cref{asu:nuisance-conv-general} characterizes the limiting behavior of the nuisance parameter estimator through the rates that are connected to the particular statistical problem, which will be further discussed in \cref{rmk:estimating-eta-not-affecting-variance}. \Cref{asu:reg-general,asu:donsker} are standard regularity assumptions for Z-estimators \citep{van2000asymptotic}.

\begin{theorem}[Asymptotic normality of $\hat\theta$]
	\label{thm:general-normality}
	Suppose \cref{asu:unique-zero,asu:nuisance-conv-general,asu:reg-general,asu:donsker} hold. Then $\sqrt{n}(\hat\theta - \theta^\star) \dto N ( 0, \Sigma)$ as $n\to\infty$, where the asymptotic variance is
	\begin{align}
		\Sigma = \PP\{\partial_\theta m(\theta^\star, \eta')\}^{-1} ~ \PP\{ m(\theta^\star, \eta') m(\theta^\star, \eta')^T \} ~ \PP\{\partial_\theta m(\theta^\star, \eta')\}^{-1, T}. \label{eq:def-asymp-var}
	\end{align}
	In addition, $\Sigma$ can be consistently estimated by
	\begin{align*}
		\PP_n\{\partial_\theta m(\hat\theta, \hat\eta)\}^{-1} ~ \PP_n\{ m(\hat\theta, \hat\eta) m(\hat\theta, \hat\eta)^T \} ~ \PP_n\{\partial_\theta m(\hat\theta, \hat\eta)\}^{-1, T}.
	\end{align*}
\end{theorem}

Next we describe a cross-fitted version of the two-step estimator, $\check\theta$, and establish its asymptotic normality. In the first step, we take a $K$-fold equally-sized random partition $(B_k)_{k=1}^K$ of observation indices $[n]$, and denote by $B_k^c = [n] \setminus B_k$ for $k \in [K]$. For each $k \in [K]$, an estimator $\hat\eta_k$ for the nuisance parameter $\eta$ is obtained using $(O_i)_{i \in B_k^c}$. In the second step, $K^{-1} \sum_{k=1}^K \PP_{n,k} m(\theta, \hat\eta_k) = 0$ is solved for $\check\theta$. The $\check\beta$ in \cref{A-algo:estimator-cf} is such an estimator with Steps 1-2 being the first stage and Step 3 being the second stage. We make the following assumptions.

\begin{assumption}
	\label{asu:nuisance-conv-general-cf}
	Suppose there exists $\eta' \in \cT$ such that the following hold.
	\begin{asulist}
		\item \label{asu:nuisance-conv-PPee-sup-cf} For each $k\in[K]$, $\sup_{\theta \in \Theta} | \PP m(\theta, \hat\eta_k) - \PP m(\theta, \eta') | = o_P(1)$;
		\item \label{asu:nuisance-conv-ee-l2-cf} For each $k\in[K]$, $ \| m(\theta^\star, \hat\eta_k)  - m(\theta^\star, \eta') \|^2  = o_P(1)$.
		\item \label{asu:nuisance-conv-PPee-deriv-cf} For each $k\in[K]$, $| \PP \{\partial_\theta m(\theta^\star, \hat\eta_k)\} - \PP \{\partial_\theta m(\theta^\star, \eta')\} | = o_P(1)$;
		\item \label{asu:nuisance-conv-PPee-meat-cf} For each $k\in[K]$, $| \PP \{m(\theta^\star, \hat\eta_k) m(\theta^\star, \hat\eta_k)^T\} - \PP \{m(\theta^\star, \eta') m(\theta^\star, \eta')^T \} | = o_P(1)$.
	\end{asulist}
\end{assumption}

\begin{theorem}[Asymptotic normality of $\check\theta$]
	\label{thm:general-normality-cf}
	Suppose \cref{asu:unique-zero,asu:reg-general,asu:nuisance-conv-general-cf} hold. Then $\sqrt{n}(\check\theta - \theta^\star) \dto N (0, \Sigma)$ as $n\to\infty$ with $\Sigma$ defined in \eqref{eq:def-asymp-var}. In addition, $\Sigma$ can be consistently estimated by
	\begin{align*}
		\Big[\frac{1}{K}\sum_{k=1}^K\PP_{n,k} \{ \partial_\theta m(\check\theta, \hat\eta_k) \} \Big]^{-1} \Big[ \frac{1}{K}\sum_{k=1}^K\PP_{n,k}  \{ m(\check\theta, \hat\eta_k)^{\otimes 2} \} \Big] \Big[\frac{1}{K}\sum_{k=1}^K\PP_{n,k} \{ \partial_\theta m(\check\theta, \hat\eta_k) \} \Big]^{-1, T}.
	\end{align*}
\end{theorem}

It is evident from comparing \cref{thm:general-normality,thm:general-normality-cf} that cross-fitting removes the Donsker class requirement.

\begin{remark}
  \label{rmk:estimating-eta-not-affecting-variance}
  Related results for non-cross-fitted estimators with plug-in infinite-dimensional nuisance parameters appear in the literature under various forms and assumptions (\citet{newey1994large}; \citet{newey1994asymptotic}; \citet{van2000asymptotic}; \citet{lok2024estimating}). Our \cref{thm:general-normality} adds to this literature by focusing on the easily verifiable condition of global robustness (\cref{def:globally-robust}), and as a result it requires almost no assumptions about the nuisance parameter estimates. For cross-fitted estimators, \citet{chernozhukov2018double} proved asymptotic normality of Neyman-orthogonalized estimators, a more general class than globally robust estimators. However, in settings such as randomized or micro-randomized trials, verifying global robustness and thus applying \cref{thm:general-normality-cf} is often simpler than verifying Neyman-orthogonality. Moreover, because we focus on globally robust estimators, our regularity conditions (\cref{asu:nuisance-conv-PPee-sup-cf,asu:nuisance-conv-ee-l2-cf,asu:nuisance-conv-PPee-deriv-cf,asu:nuisance-conv-PPee-meat-cf}) are each weaker than or comparable to the corresponding terms in Assumptions 3.2(c) and 3.4(c) of \citet{chernozhukov2018double}. We emphasize, however, that this ease of verification applies specifically to globally robust estimating equations; in observational studies or settings with missing data, one may return to the broader theory such as \citet{chernozhukov2018double} when no globally robust estimating equations are available.
\end{remark}

\section{Simulation}
\label{sec:simulation}

We conducted simulations to validate (i) the efficiency gain of the proposed estimators compared to the ones in the literature, (ii) the efficiency gain when the serial independence working assumption fails, and (iii) the benefits of cross-fitting when the ML estimator is potentially non-Donsker. Throughout the simulations, we implemented an additional small-sample correction technique for the variance estimator similar to \citet{mancl2001} and \citet{qian2021estimating} to improve the confidence interval coverage probability when the sample size is small. Details of this correction are in \cref{A-sec:small-sample-correction}. Each setting was simulated for 1000 replicates. All simulations in this section use continuous outcomes; analogous results for binary and count outcomes are reported in \cref{A-sec:appendix-additional-sims}.

\subsection{Marginal Causal Excursion Effect}
\label{subsec:simulation-continuous}

We set the total number of decision points to $T = 10$.
For each individual, their observations $(Z_t, A_t, Y_{t+1})$ are generated sequentially for $t \in [T]$. $Z_t \sim \text{Unif}[-2, 2]$ is generated independently of the past observations. $A_t$ is generated from a Bernoulli distribution with a constant randomization probability $p_t(H_t) = 0.5$. The proximal outcome is generated by $Y_{t+1} = \mu_t(H_t, A_t) + \epsilon_t$, where $\mu_t(H_t, A_t) = A_t (\beta_0 + \beta_1 Z_t) + \mu_t(H_t, 0)$ and $\mu_t(H_t, 0)$ can take one of four forms: Linear, where $\mu_t(H_t, 0) = \alpha_0 + \alpha_1 t + \alpha_2 Z_t$; simple nonlinear, where $\mu_t(H_t, 0) = \alpha_0 + \lambda_1 \{ q_{2,2}(Z_t/6 + 1/2) + q_{2,2}(t/T) \}$ with $q_{2,2}$ being the probability density function of $\text{Beta}(2,2)$ distribution; periodic, where $\mu_t(H_t, 0) = \alpha_0 + \lambda_1 \{\sin(t) + \sin(Z_t)\}$; step function, where $\mu_t(H_t, 0) = \alpha_0 + \lambda_1 \{ \indic(t \text{ is even}) + \indic(\lfloor 10Z_t \rfloor \text{ is even}) \}$ with $\lfloor \cdot \rfloor$ being the floor function. The error terms $(\epsilon_t)_{t \in [T]}$ are generated from a multivariate normal distribution with mean zero, an exponential correlation matrix $\text{Corr}(\epsilon_t, \epsilon_u) = \rho^{|t - u| / 2}$, and time-varying variance $\text{Var}(\epsilon_t) = (t-1)\lambda_2 + \lambda_3$. We set $\beta_0 = 0.5$, $\beta_1 = 0.2$, and $\alpha_0 = \alpha_1 = \alpha_2 = 1$. The tuning parameters $\lambda_1$, $\lambda_2$, $\lambda_3$, and $\rho$ characterize the nonlinearity in $\mu_t(H_t, 0)$, the nonstationarity over time in the variance of the outcome, the magnitude of the variance of the outcome, and within-subject correlation in the outcome, respectively.

We aim to estimate the marginal CEE with $S_t = \emptyset$ and $g$ being the identity link, with true value equal to $\beta_0$. We set $\rho = 0.5$, $\lambda_3 = 1$, and varied $\lambda_1 \in [0,3]$, and $\lambda_2 \in [0,3]$. We set the control variables used in fitting $\hat\mu_t$ as either $t$ alone or $(t, Z_t)$. In addition to WCLS, we included three variations of $\hat\beta$: GAM, where Step 1 of \cref{algo:estimator-ncf} uses a generalized additive model with penalized splines \citep{wood2011fast}; RF, where Step 1 uses a fast implementation of random forests \citep{wright2017ranger}; SL, where Step 1 uses the super learner ensemble where the individual learners include sample mean, generalized linear model, GAM, multivariate spline, random forest, XGBoost, and neural network \citep{polley2021superlearner}. Throughout, the working model in Step 1 is fit by pooling across decision points: it is fit on the long-format dataset $\{(t, Z_{it}, A_{it}, Y_{i,t+1}): i \in [n], t \in [T]\}$ with $t, Z_t$ included as covariates, separately for $A_t = 1$ and $A_t = 0$ to obtain $\hat\mu_t(H_t, 1)$ and $\hat\mu_t(H_t, 0)$. This pooling, rather than fitting a separate model per $t$, keeps the computation feasible when $T$ is large; further implementation details are in \cref{A-sec:implementation}. For all three estimators, Step 3 of \cref{algo:estimator-ncf} uses $t$-specific empirical averages to approximate the two expectations. We also included the cross-fitting versions of GAM, RF, and SL (i.e., $\check\beta$ from \cref{A-algo:estimator-cf} where the nuisance parameters are fitted with GAM, RF, SL, respectively). Finally, to investigate the theoretical limit of efficiency gain, we included an oracle estimator, which is the $\beta$ that solves $\PP_n \sum_{t=1}^T d_t^\star(S_t) \phi_t(\beta, \mu_t^\star) = 0$.

The left panels of \cref{fig:simulation-cont-control_pattern,fig:simulation-cont-error_var_pattern} show that the mean squared error (MSE) of all estimators decreases to 0 as sample size $n$ increases in all scenarios. This demonstrates the robustness of the estimators because they are consistent even when the working model for Step 1 of \cref{algo:estimator-ncf} is misspecified (e.g., when $t$ alone is included in the control variable).

The middle panels of \cref{fig:simulation-cont-control_pattern,fig:simulation-cont-error_var_pattern} show that all estimators have close to nominal confidence interval coverage (95\%) except for the non-cross-fitting version of RF and SL. The lower-than-nominal coverage of RF and SL is due to random forests potentially not satisfying the Donsker condition required by \cref{thm:general-normality}, and the coverage issue is fixed by using cross-fitting.

In the right panel of \cref{fig:simulation-cont-control_pattern}, as the nonlinearity in $\mu_t(H_t, 0)$ increases (i.e., larger $\lambda_1$), the relative efficiency of the proposed estimators against WCLS increases. When the true $\mu_t(H_t, 0)$ takes a linear form, the performance of GAM, RF, and SL is almost identical to WCLS. When the true $\mu_t(H_t, 0)$ takes a nonlinear form, GAM, RF, and SL are more efficient than WCLS, even when the working model for Step 1 of \cref{algo:estimator-ncf} is misspecified. When the true $\mu_t(H_t, 0)$ takes a simple nonlinear and periodic form, GAM, RF, and SL are close to Oracle in terms of relative efficiency when $(t,Z_t)$ are included in the control variables. In the right panel of \cref{fig:simulation-cont-error_var_pattern}, as the nonstationarity over time in variance of the outcome increases (i.e., larger $\lambda_2$), the relative efficiency of the proposed estimators against WCLS increases. Even when the true $\mu_t(H_t, 0)$ is linear, the proposed estimator shows a substantial efficiency gain. Across all settings, the proposed estimator shows meaningful efficiency gain with a relative efficiency up to 3. The efficiency gain is partially due to the proposed estimators approximating the nonlinear $\mu_t(H_t, A_t)$ with ML methods, and partially due to the estimating equations being optimally weighted based on the decision-point-specific variance (which is more pronounced in \cref{fig:simulation-cont-error_var_pattern}).

Based on the simulation results, SL with cross-fitting is the best performing across various scenarios, which is intuitive because the ensemble combines the strengths of different ML methods for different scenarios, and cross-fitting ensures adequate confidence interval coverage without requiring the Donsker condition, allowing a larger class of ML methods to be used in the ensemble. Computation-wise, SL with cross-fitting is the slowest, but the speed is tolerable for data applications where there is no need for Monte Carlo repetitions. Notably, substantial efficiency gain is obtained by the proposed estimator even when the serial independence condition in \cref{thm:efficiency-bound} does not hold (as $\rho = 0.5$). This demonstrates the efficiency gain of the proposed estimator in realistic scenarios.

\begin{figure}
        \centering
        \includegraphics[width = \textwidth]{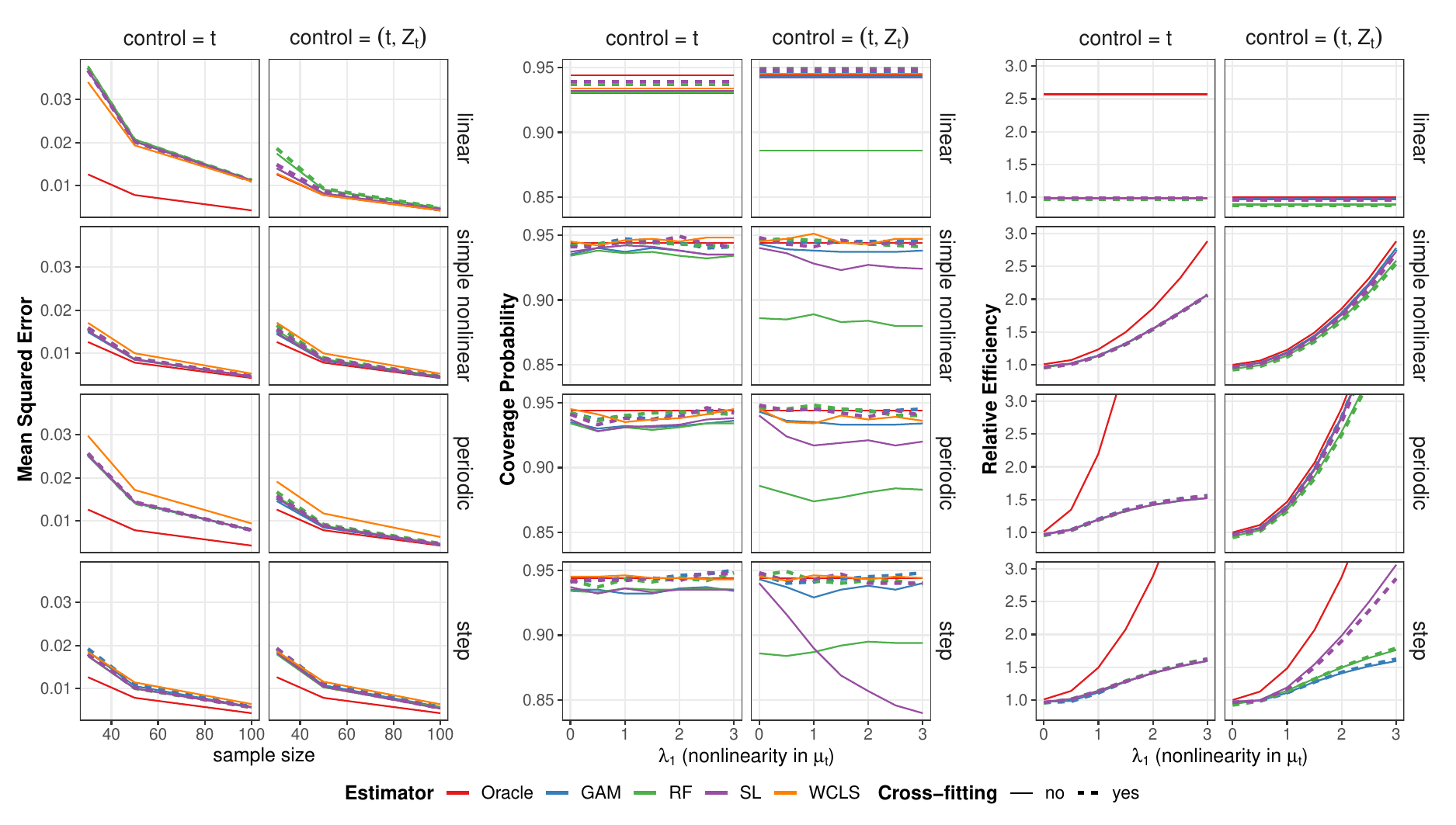}
	\caption{Simulation results for continuous outcome where we varied $\lambda_1$, the nonlinearity in $\mu_t$. The other tuning parameters were fixed at $\lambda_2 = 0$, $\lambda_3 = 1$, $\rho = 0.5$. In the left panel, $\lambda_1$ is fixed at 1 and the sample size $n$ varies from 30 to 100. In the middle and the right panels, the sample size $n$ is fixed at 100 and $\lambda_1$ varies from 0 to 3. When $\mu_t(H_t, 0)$ takes linear form, $\lambda_1$ does not enter the generative model and thus the top middle and top right panels are flat lines.}
	\label{fig:simulation-cont-control_pattern}
\end{figure}

\begin{figure}
        \centering
        \includegraphics[width = \textwidth]{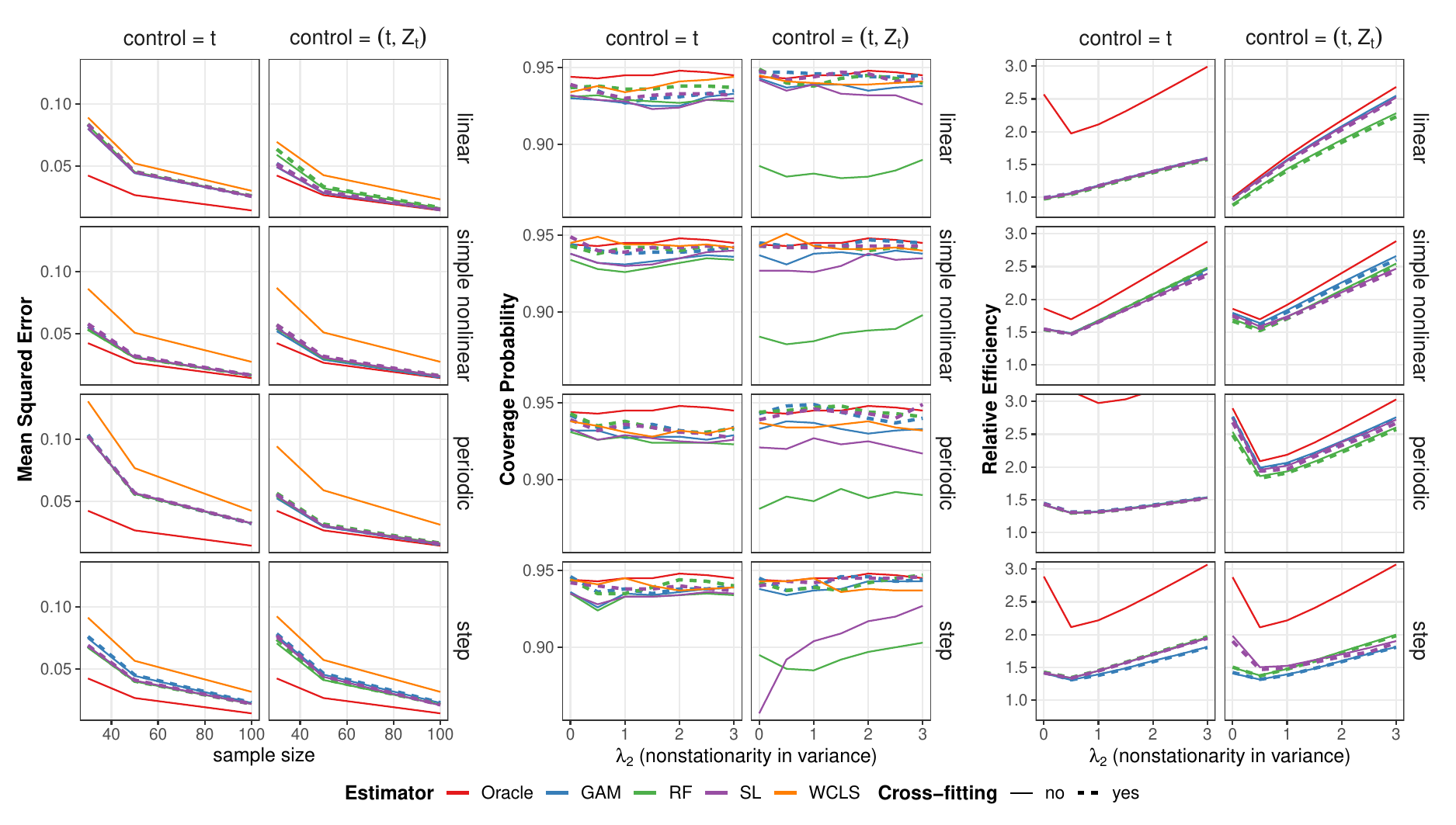}
	\caption{Simulation results for continuous outcome where we varied $\lambda_2$, the nonstationarity over time in the variance of the outcomes. The other tuning parameters were fixed at $\lambda_1 = 2$, $\lambda_3 = 1$, $\rho = 0.5$. In the left panel, $\lambda_2$ is fixed at 1 and the sample size $n$ varies from 30 to 100. In the middle and the right panels, the sample size $n$ is fixed at 100 and $\lambda_2$ varies from 0 to 3.}
	\label{fig:simulation-cont-error_var_pattern}
\end{figure}

\subsection{Time-Varying Causal Excursion Effect}
\label{subsec:simulation-time-varying-CEE}

The simulations in \cref{subsec:simulation-continuous} target a marginal CEE that is constant in $t$. To assess the proposed estimators when the CEE varies with the decision point, we modify the continuous-outcome generative model in a single line: the CEE becomes $\gamma_t(\beta) = \beta_0 + \beta_1 (t/T)$, with $\beta_0 = 0.5$ and $\beta_1 = -0.3$, so the true CEE decays from approximately $0.47$ at $t=1$ to $0.20$ at $t=T=10$. The covariate $Z_t \sim \mathrm{Unif}[-2,2]$ is used to set $\mu_t(H_t,0)$ to the four control patterns (linear, simple nonlinear, periodic, step), and all other parameters and the error structure are as in \cref{subsec:simulation-continuous}.

We aim to estimate $\beta = (\beta_0, \beta_1)^T$ using moderator $S_t = t/T$. We include the same set of estimators as in \cref{subsec:simulation-continuous}: WCLS, the proposed estimator $\hat\beta$ with GAM, RF, and SL nuisance fits, the cross-fitted $\check\beta$ for each of these three ML methods, and the oracle. To keep compute manageable, the SL library here uses the smooth learners (\texttt{SL.mean, SL.glm, SL.gam, SL.earth}); we still label it ``SL'' in the figure. The control variables for fitting $\hat\mu_t$ are $(t, Z_t)$.

\Cref{fig:simulation-cont-time-varying-CEE} reports MSE, confidence interval coverage, and relative efficiency against WCLS, separately for $\hat\beta_0$ and $\hat\beta_1$. The findings mirror those of \cref{subsec:simulation-continuous}. MSE decreases for every estimator as $n$ grows. Coverage is close to the nominal $95\%$ level except for the non-cross-fitted RF, which falls to approximately $0.90$; cross-fitting restores nominal coverage, consistent with the Donsker discussion in \cref{subsec:simulation-continuous}. Relative efficiency against WCLS grows with $\lambda_1$ in the three nonlinear control patterns, reaching approximately $3.5$ in simple-nonlinear, $4$ in periodic, and $5$ in step (for the oracle) at $\lambda_1 = 3$ and $n = 100$. The cross-fitted GAM, RF, and SL track the oracle closely in the smooth control patterns, with a larger gap in the step pattern where smooth learners cannot fit well the discontinuous $\mu_t$. These results indicate that the efficiency gains documented in \cref{subsec:simulation-continuous} extend to a multi-dimensional $\beta$.

\begin{figure}
        \centering
        \includegraphics[width = \textwidth]{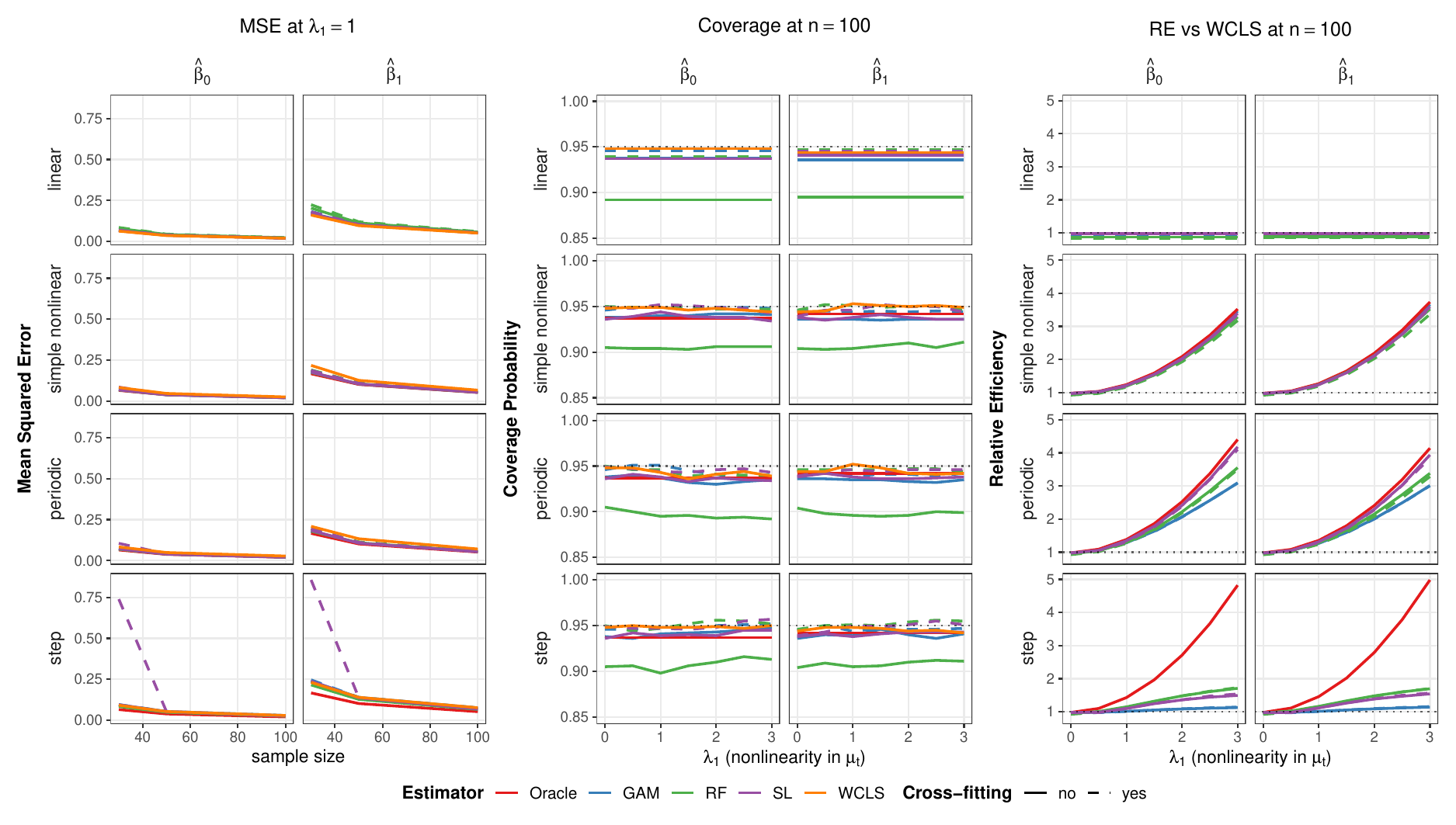}
	\caption{Simulation results for continuous outcome with time-moderated CEE $\gamma_t(\beta) = \beta_0 + \beta_1(t/T)$, $\beta_0 = 0.5$, $\beta_1 = -0.3$. Each plot is split into two columns for $\hat\beta_0$ and $\hat\beta_1$, and four rows for the four $\mu_t$ patterns. In the left plot, $\lambda_1$ is fixed at 1 and the sample size $n$ varies from 30 to 100. In the middle and right plots, $n$ is fixed at 100 and $\lambda_1$ varies from 0 to 3.}
	\label{fig:simulation-cont-time-varying-CEE}
\end{figure}

\subsection{Performance under Many Decision Points}
\label{subsec:simulation-large-T}

The asymptotic theory in this paper is developed under fixed $T$ and large $n$, which is the correct asymptotic regime for the Drink Less MRT analyzed in \cref{sec:application} with $T=30$ and $n=349$. Other MRTs may have $T$ comparable to or larger than $n$ \citep{klasnja2015,battalio2021sense2stop}. To assess the proposed estimators in this regime, we revert to the marginal CEE ($S_t = \emptyset$, $\gamma_t(\beta) = \beta_0 = 0.5$) using the simple-nonlinear control pattern at $\lambda_1 = 1$, fix $n \in \{30, 100\}$, and vary $T \in \{10, 30, 50, 100, 200\}$, with all other settings as in \cref{subsec:simulation-continuous}. We use the same eight estimators as in \cref{subsec:simulation-time-varying-CEE}.

\Cref{fig:simulation-cont-large-T-regime} reports MSE, coverage, and relative efficiency against WCLS as functions of $T$. The MSE of $\hat\beta_0$ shrinks substantially as $T$ grows, indicating that the estimators gain efficiency from the additional decision points despite the within-subject correlation. Confidence interval coverage stays approximately at nominal across the entire grid, except for the non-cross-fitted RF, which dips to $0.85$--$0.90$; cross-fitting restores nominal coverage, mirroring the Donsker-related finding in \cref{subsec:simulation-continuous}. The relative efficiency of the cross-fitted proposed estimators against WCLS remains stable in $T$, ranging roughly between $1.1$ and $1.25$ across the grid, and the data-adaptive methods track the oracle closely. Although the asymptotic theory does not directly cover the regime where $T$ grows with $n$ fixed, these empirical results suggest that the proposed estimator continues to provide consistent estimation, valid inference, and meaningful efficiency gain when $T$ is large relative to $n$.

\begin{figure}
        \centering
        \includegraphics[width = \textwidth]{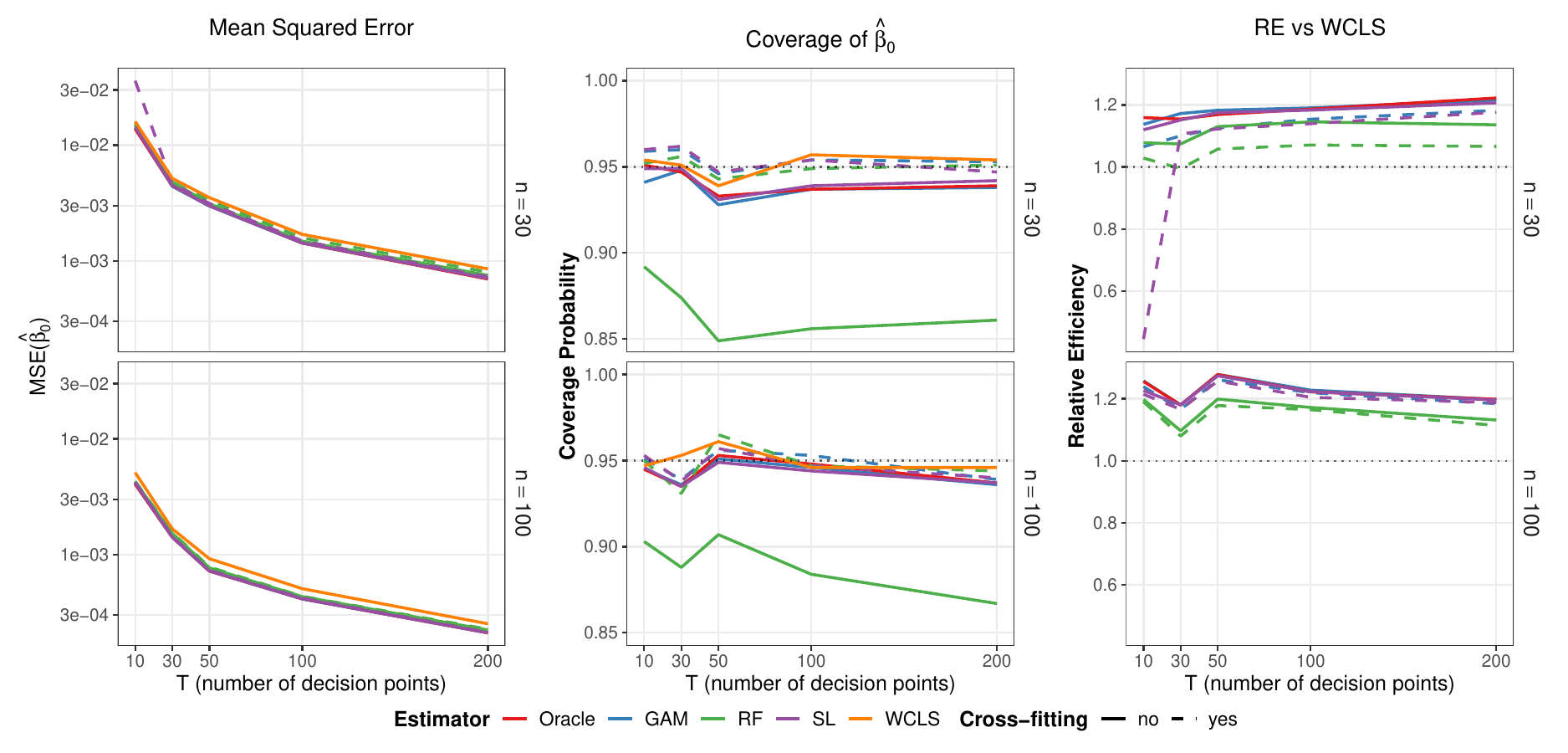}
	\caption{Simulation results for continuous outcome with marginal CEE $\gamma_t(\beta) = \beta_0 = 0.5$ as the number of decision points $T$ varies. The control pattern is simple nonlinear at $\lambda_1 = 1$, with $\rho = 0.5$. The top row of facets is for $n = 30$ and the bottom row for $n = 100$.}
	\label{fig:simulation-cont-large-T-regime}
\end{figure}

\section{Application: Drink Less MRT}
\label{sec:application}

Drink Less is a smartphone app aimed at reducing harmful alcohol consumption \citep{garnett2019development, garnett2021refining}. An MRT was conducted to assess the effect of push notifications on user engagement \citep{bell2020notifications}. 349 participants were randomized daily at 8pm for 30 days to either receive an engagement prompt (0.6 probability) or nothing (0.4 probability). The prompts encouraged monitoring of drinking habits via the app. Researchers observed a substantial increase in app engagement within the first hour after the notification and a smaller effect up to 24 hours after the notification \citep{bell2023notifications}.

We illustrate the efficiency gain of our methods using the Drink Less data set through analyzing the CEE of the push notification. We considered the continuous proximal outcome to be the seconds of app interaction within one hour of the decision point, i.e., between 8pm and 9pm, the binary proximal outcome to be whether the user interacts with the app within one hour of the decision point, and the count proximal outcome to be the number of times user interacts with the app within one hour of the decision point. We considered a marginal analysis ($S_t = \emptyset$) and a moderation analysis ($S_t = t$, the decision point index). For each analysis, we applied the proposed two-stage estimator using super learner with cross-fitting (SL.CF), generalized additive model with cross-fitting (GAM.CF) and without cross-fitting (GAM), and random forest with cross-fitting (RF.CF). For count outcome, SL.CF is replaced by stacking with cross-fitting (STACK.CF). Super learner (or stacking) without cross-fitting and random forest without cross-fitting were not included due to their lower-than-nominal coverage probability (\cref{sec:simulation}). For comparison, we also applied the WCLS for continuous proximal outcome and the EMEE for binary and count proximal outcomes.

\Cref{fig:analysis-result} shows the estimated CEE parameters along with the 95\% confidence intervals and the estimated relative efficiency against the original WCLS/EMEE method. When the moderator is the empty set, $\beta_0$ represents the marginal effect. $\beta_0$ and $\beta_1$ represent the intercept and the slope of the moderator in the CEE model. The estimated relative efficiency is the ratio between the estimated variances of two estimators, and a value greater than 1 indicates that the proposed two-stage estimator is more efficient than the original WCLS/EMEE method. The point estimates of all estimators are close to each other. For continuous proximal outcome (left column of \cref{fig:analysis-result}), GAM and RF.CF are always more efficient than WCLS with relative efficiency ranging between 1.14 and 1.40, which roughly translates to savings in sample size between 12\% ($= 1 - 1/1.14$) and 29\% ($= 1 - 1/1.40$) if one were to power a study using the proposed estimator instead of the WCLS. SL.CF and GAM.CF have comparable efficiency to WCLS. For binary and count outcomes (middle and right columns of \cref{fig:analysis-result}), all the proposed estimators have similar efficiency and are more efficient than EMEE in all settings. The efficiency gain for these two outcomes is also more pronounced, with relative efficiency up to 2.15 (RF.CF in the right middle panel), which roughly translates to 53\% savings in sample size. This analysis demonstrates the efficiency gain achieved by the proposed method in real world applications.

\begin{figure}
    \centering
    \includegraphics[width = \textwidth]{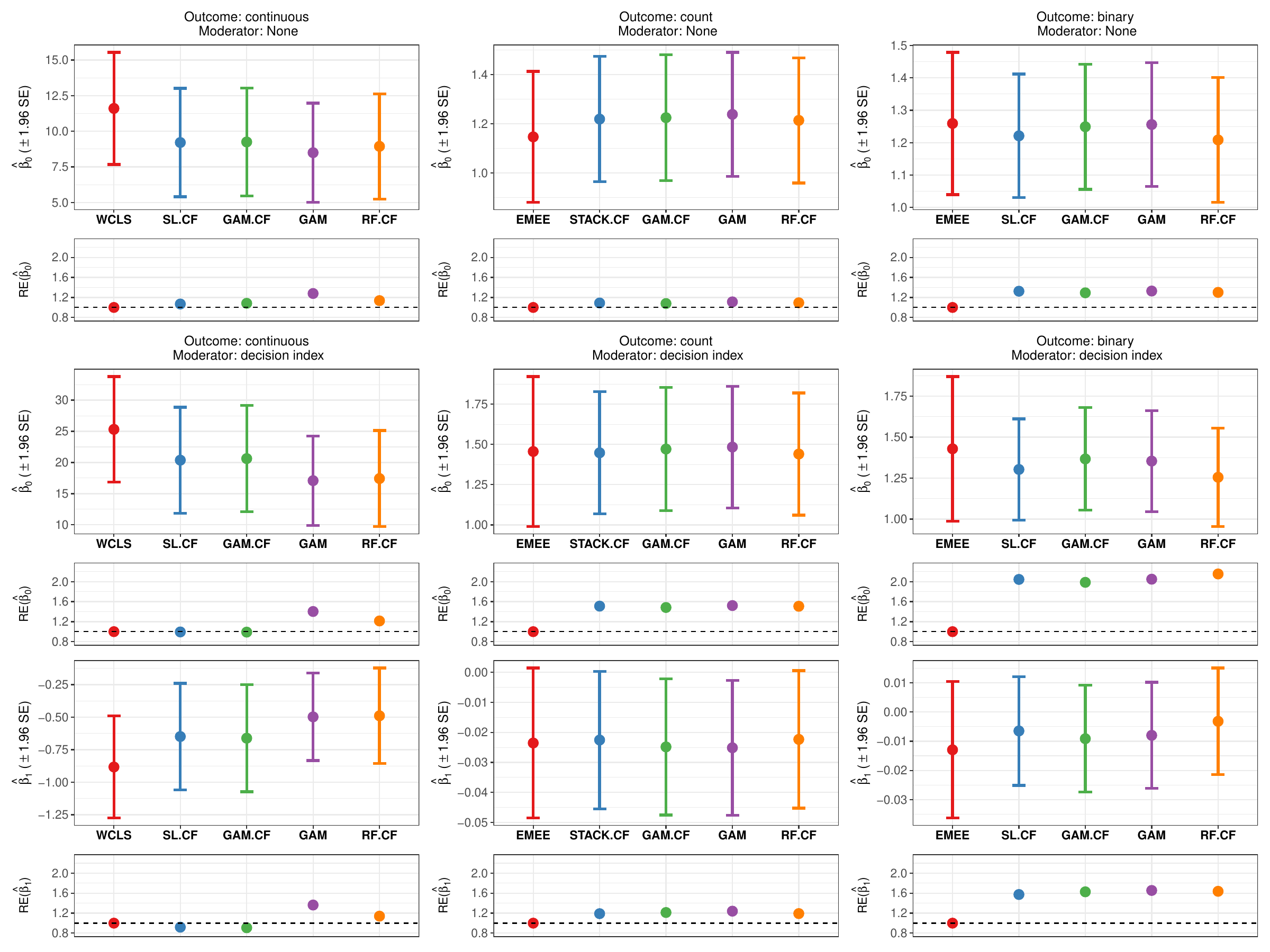}
    \caption{Analysis of Drink Less MRT for continuous, count, and binary proximal outcomes comparing the original WCLS/EMEE estimator and the proposed two-stage estimators with two moderators. In the top panel, moderator is the empty set. In the bottom panel, moderator is the decision point index.}
    \label{fig:analysis-result}
\end{figure}

\section{Discussion}
\label{sec:discussion}

In this paper, we studied the semiparametric efficient estimation of CEE. We derived a semiparametric efficiency bound under working assumptions and proposed two-stage estimators that achieve the efficiency bound. Our unified framework handles CEE with identity or log link and applies to continuous, binary, or count longitudinal outcomes. We proved the asymptotic normality of the proposed estimators (with and without cross-fitting) and established a general asymptotic theory for a class of globally robust Z-estimators. Substantial efficiency gain over widely-used estimators for CEE in the literature was demonstrated through simulation studies and real world applications. The proposed estimator can be applied to primary and secondary analyses of MRTs and improve power or reduce the required sample size.

The efficiency gains of our proposed estimator arise from two sources: (i) nonparametric estimation of the optimal nuisance control model and (ii) optimal weighting across decision points. An important direction for future work is to investigate when additional efficiency gain occurs from each of the sources. When the DR-WCLS implementation \citep{shi2023meta} becomes publicly available in code or software, it would be valuable to compare our approach with theirs to assess how much of the additional efficiency stems from our inverse-variance weights versus the optimal outcome regression nuisance functions. Based on \cref{fig:simulation-cont-error_var_pattern}, we conjecture that efficiency improvement of our estimators over DR-WCLS will be more prominent when there is substantial nonstationarity in outcome variance over time.

When using random forest with default parameters to estimate the nuisance model, we noticed some numerical instability when sample size is small. Tuning the parameters for random forest improved its numerical stability and brought its performance on par with that of other estimators.
For other nonparametric methods we used, parameter tuning seems to make little difference. A practical recommendation is to always incorporate tuning-parameter selection when using machine learning algorithms to fit the nuisance parameters.

Computational scalability of the proposed estimator is governed by the cost of fitting $\hat\mu_t$ in Step 1 and $\hat d_t$ in Step 3 of \cref{algo:estimator-ncf} (and the analogous Steps 2(i) and 2(iii) in \cref{A-algo:estimator-cf} for the cross-fitted variant $\check\beta$). To make the computation scalable with large $T$, both stages are implemented by pooling across decision points where $H_t$ is summarized through a fixed set of features (e.g., recent lags, contextual variables, and $t$ itself). Cross-fitting refits each model on $(K-1)/K$ of the data $K$ times and is embarrassingly parallelizable across folds. On the Drink Less data ($n = 349$, $T = 30$, marginal CEE for the continuous outcome), total wall-clock time per analysis (with $K = 5$ folds for cross-fitting) ranged from 0.4 seconds (GAM, no cross-fitting) to 4.9 seconds (super learner with cross-fitting and a smooth-only library) to 61 seconds (super learner with cross-fitting and a library that includes XGBoost and a neural network) on a 2024 M4 MacBook Pro using each package's default threading settings, with cross-fitting folds and super-learner library learners run sequentially in R. The dominant cost is the XGBoost learner in the Super Learner library; replacing the library with the smooth-only subset (\texttt{SL.mean}, \texttt{SL.glm}, \texttt{SL.gam}, \texttt{SL.earth}) reduces computation time by roughly an order of magnitude while maintaining similar efficiency in our simulations. The proposed estimator is thus practical at MRT scale; for trials with substantially larger $T$, the same long-format pooling continues to apply, and one can additionally summarize $H_t$ through a fixed-dimensional vector of sufficient statistics to keep the input dimension constant in $t$.

The semiparametric efficiency bound is derived under specific working assumptions about serial dependence. Future work may extend the theoretical analysis in two directions. First, it would be useful to develop a set of more interpretable and easily verifiable assumptions. Second, one could consider the more general problem of semiparametric efficient estimation without such assumptions, potentially by explicitly modeling correlation across decision points---an approach that becomes tractable when the sample size substantially exceeds the total number of decision points.

The asymptotic regime in the paper allows the sample size $n$ to diverge with the total number of decision points $T$ fixed. This is appropriate for many MRTs including the Drink Less MRT where $n$ is greater than $T$. For MRTs where $T$ is much greater than $n$ (e.g., \citealt{battalio2021sense2stop}), however, it may be fruitful to consider the asymptotic regime with $T\to\infty$.

Another future direction is to consider general link functions and survival outcomes. Extending our approach to alternative link functions, such as the logit link, or to survival outcomes introduces several nontrivial challenges. For example, with the logit link, there are no doubly-robust estimators under the usual joint model of correctly specified propensity score or outcome regression \citep{tchetgen2010doubly}, which makes our globally robust estimation framework inapplicable. Alternative estimators for causal excursion odds ratios might be considered and extended to improve efficiency \citep{yu2025robust}. For survival proximal outcomes, appropriately defining the causal effect in the MRT setting and handling censoring is also nontrivial and thus beyond the scope of this paper.

Code to reproduce the simulation and data analysis can be accessed at \url{https://github.com/tqian/paper_CEE-efficient-ML}.

\begin{acks}[Acknowledgments]
The authors would like to thank Dr. Edward Kennedy and Dr. Bin Nan for helpful comments. The authors also thank Dr. Elizabeth Williamson, Dr. Henry WW Potts, Dr. Claire Garnett, Dr. Olga Perski, and the participants of the Drink Less MRT for their important contributions to the Drink Less MRT.
\end{acks}

\begin{funding}
The second author was supported by a PhD studentship funded by the MRC Network of Hubs for Trials Methodology Research (MR/L004933/2-R18).
\end{funding}

\bibliographystyle{imsart-nameyear}
\bibliography{mhealth-ref}

\clearpage
\begin{appendix}
\numberwithin{equation}{section}
\numberwithin{table}{section}
\numberwithin{figure}{section}
\numberwithin{algocf}{section}
\setlist[thmlist]{label=(\roman*), ref=\thethm\,(\roman*)}
\setlist[asulist]{label=(\roman*), ref=\theasu\,(\roman*)}
\thinmuskip=1mu
\medmuskip=2mu
\thickmuskip=3mu
\renewcommand\arraystretch{1.2}

\pdfbookmark[1]{Appendix}{appendix-start}
\begin{center}
{\Large\bfseries Appendix}
\end{center}
\vspace{0.5\baselineskip}

\noindent\Cref{A-sec:proof-efficiency-bound} establishes the semiparametric efficiency bound and gives sufficient conditions under which the working assumptions hold. \Cref{A-sec:comparison-existing-estimators} compares the proposed estimators with existing ones: it proves the asymptotic-variance comparison with the unadjusted estimator, presents a counter-example under carryover effects, and verifies the special-case equivalences underlying \cref{tab:estimator-comparison}. \Cref{A-sec:general-lemmas} presents technical proofs for the asymptotic theory for globally robust Z-estimators. \Cref{A-sec:continuous-can} presents technical proofs for the asymptotic theory of the CEE estimators, including the cross-fitted estimator $\check\beta$. \Cref{A-sec:dr-observational} establishes the double robustness of the proposed estimator in observational settings. \Cref{A-sec:implementation} discusses estimator implementation details. \Cref{A-sec:small-sample-correction} presents details for the bias-corrected variance estimator. \Cref{A-sec:appendix-additional-sims} reports simulation studies for binary and count outcomes.

\section{Semiparametric Efficiency Bound}
\label{A-sec:proof-efficiency-bound}

\subsection{Proof of \texorpdfstring{\cref{thm:efficiency-bound}}{Theorem 1}}
\label{A-subsec:proof-thm-efficiency-bound}

\begin{lem}[CEE model as an unconditional moment restriction]
	\label{A-lem:moment-restriction}
	The CEE model \cref{eq:linear-model-cee} is equivalent to the unconditional moment restriction
	\begin{align}
		\EE\bigg\{ \frac{A_t - p_t}{p_t(1-p_t)} I_t U_t(\beta^\star) \nu_t(S_t) \bigg\} = 0 \quad \text{for all functions $\nu_t$ and for } t \in [T]. \label{A-eq:unconditional-moment-restriction}
	\end{align}
\end{lem}

\begin{proof}[Proof of \cref{A-lem:moment-restriction}]

We first show that the CEE model \cref{eq:linear-model-cee} is equivalent to the conditional moment restriction
\begin{align}
	\EE\bigg\{ \frac{A_t - p_t}{p_t(1-p_t)} U_t(\beta^\star) ~\bigg|~ S_t, I_t = 1 \bigg\} = 0 \quad \text{for } t \in [T], \label{A-eq:conditional-moment-restriction}
\end{align}
and then show that the conditional moment restriction \cref{A-eq:conditional-moment-restriction} is equivalent to the unconditional moment restriction \cref{A-eq:unconditional-moment-restriction}.

\textbf{Step 1:} Proving the equivalence between \cref{eq:linear-model-cee} and \cref{A-eq:conditional-moment-restriction}. Using the law of iterated expectation, the left hand side of \cref{A-eq:conditional-moment-restriction} is
\begin{align*}
	& ~~~\EE\bigg\{ \frac{A_t - p_t}{p_t(1-p_t)} U_t(\beta^\star) ~\bigg|~ S_t, I_t = 1 \bigg\} \\
	& = \EE\bigg[ \EE \bigg\{ \frac{A_t - p_t}{p_t(1-p_t)} U_t(\beta^\star) \bigg| H_t, A_t = 1 \bigg\} p_t ~\bigg|~ S_t, I_t = 1 \bigg] \\
	& ~~~ + \EE\bigg[ \EE \bigg\{ \frac{A_t - p_t}{p_t(1-p_t)} U_t(\beta^\star) \bigg| H_t, A_t = 0 \bigg\} (1-p_t) ~\bigg|~ S_t, I_t = 1 \bigg] \\
	& = \EE [ \EE\{ U_t(\beta^\star) \mid H_t, A_t = 1 \} \mid S_t, I_t = 1] - \EE [ \EE\{ U_t(\beta^\star) \mid H_t, A_t = 0 \} \mid S_t, I_t = 1].
\end{align*}
So the conditional moment restriction \cref{A-eq:conditional-moment-restriction} is equivalent to
\begin{align}
	\EE [ \EE\{ U_t(\beta^\star) \mid H_t, A_t = 1 \} \mid S_t, I_t = 1] - \EE [ \EE\{ U_t(\beta^\star) \mid H_t, A_t = 0 \} \mid S_t, I_t = 1] = 0 \quad \text{for } t \in [T], \label{A-eq:conditional-moment-restriction2}
\end{align}
If $g$ is identity, then $\EE\{ U_t(\beta^\star) \mid H_t, A_t = 1 \} = \EE(Y_{t+1} \mid H_t, A_t = 1) - \gamma_t(S_t; \beta^\star)$ and $\EE\{ U_t(\beta^\star) \mid H_t, A_t = 0 \} = \EE(Y_{t+1} \mid H_t, A_t = 0)$. If $g$ is log, then $\EE\{ U_t(\beta^\star) \mid H_t, A_t = 1 \} = \EE(Y_{t+1} \mid H_t, A_t = 1) \exp\{- \gamma_t(S_t; \beta^\star) \}$ and $\EE\{ U_t(\beta^\star) \mid H_t, A_t = 0 \} = \EE(Y_{t+1} \mid H_t, A_t = 0)$. In both cases, \cref{A-eq:conditional-moment-restriction2} is equivalent to the CEE model \cref{eq:linear-model-cee}. So we proved the equivalence between \cref{eq:linear-model-cee} and \cref{A-eq:conditional-moment-restriction}.

\textbf{Step 2:} Proving the equivalence between \cref{A-eq:conditional-moment-restriction} and \cref{A-eq:unconditional-moment-restriction}. That \cref{A-eq:conditional-moment-restriction} implies \cref{A-eq:unconditional-moment-restriction} follows immediately by the law of iterated expectation. To show that \cref{A-eq:unconditional-moment-restriction} implies \cref{A-eq:conditional-moment-restriction}, we first show this assuming that $S_t$ is discrete and takes value in $\{z_1, z_2, \ldots, z_m\}$. For any $j \in [m]$, letting $\nu_t(S_t) = \indic(S_t = z_j)$ in \cref{A-eq:unconditional-moment-restriction} implies that
\begin{align*}
	0 & = \EE\bigg\{ \frac{A_t - p_t}{p_t(1-p_t)} I_t U_t(\beta^\star) \indic(S_t = z_j) \bigg\} \\
	& = \EE\bigg\{ \frac{A_t - p_t}{p_t(1-p_t)} U_t(\beta^\star) \bigg| S_t = z_j, I_t = 1 \bigg\} P(S_t = z_j, I_t = 1),
\end{align*}
which implies that
\begin{align}
	\EE\bigg\{ \frac{A_t - p_t}{p_t(1-p_t)} U_t(\beta^\star) \bigg| S_t = z_j, I_t = 1 \bigg\} = 0. \label{A-eq:lem:moment-restriction:proofuse1}
\end{align}
\cref{A-eq:lem:moment-restriction:proofuse1} for all $j \in [m]$ combined with the multinomial approximation approach by \citet{chamberlain1987asymptotic} implies \cref{A-eq:conditional-moment-restriction}.

This completes the proof.

\end{proof}

\begin{lem}[Projection on score functions for the treatment assignment probability]
	\label{A-lem:projection-analytic-form}
	Consider the estimating function for $\beta$:
	\begin{align}
		\xi(\beta) = \sum_{t=1}^T \xi_t(\beta) = \sum_{t=1}^T \frac{A_t - p_t}{p_t(1-p_t)} I_t U_t(\beta) \partial_\beta \gamma(S_t; \beta). \label{A-eq:primitive-ee-for-beta}
	\end{align}
	If (WA-1) in \cref{thm:efficiency-bound} holds, then
	\begin{align}
		\xi(\beta) - \sum_{u=1}^T \Big[\EE\{\xi(\beta) \mid H_u, A_u\} - \EE\{\xi(\beta) \mid H_u\}\Big] = \sum_{t=1}^T \phi_t(\beta), \label{A-eq:decomposition-xi-minus-projection}
	\end{align}
	with $\phi_t(\beta)$ defined in \cref{eq:phi-definition}.
\end{lem}

\begin{proof}[Proof of \cref{A-lem:projection-analytic-form}]
	We first compute $\EE\{\xi_t(\beta) \mid H_u, A_u\} - \EE\{\xi_t(\beta) \mid H_u\}$ for each $(t,u)$-pair, with
	\begin{align*}
		\xi_t(\beta) = \frac{A_t - p_t}{p_t(1-p_t)} I_t U_t(\beta) \partial_\beta \gamma(S_t; \beta).
	\end{align*}

	For $u > t$, we have $\EE\{\xi_t(\beta) \mid H_u, A_u\} = \EE\{\xi_t(\beta) \mid H_u\} = \xi_t(\beta)$, so $\EE\{\xi_t(\beta) \mid H_u, A_u\} - \EE\{\xi_t(\beta) \mid H_u\} = 0$.

	For $u = t$, we have
	\begin{align}
		& ~~~ \EE\{\xi_t(\beta) \mid H_u, A_u\} \nonumber \\
		& = \frac{A_t - p_t}{p_t(1-p_t)} I_t \EE\{U_t(\beta)\mid H_t, A_t\} \partial_\beta \gamma(S_t; \beta) \nonumber \\
		& = I_t \bigg[\frac{A_t}{p_t} \EE\{U_t(\beta)\mid H_t, A_t = 1\} - \frac{1 - A_t}{1 - p_t} \EE\{U_t(\beta)\mid H_t, A_t = 0\} \bigg] \partial_\beta \gamma(S_t; \beta), \label{A-eq:lem:projection-analytic-form:proofuse1}
	\end{align}
	and
	\begin{align}
		& ~~~ \EE\{\xi_t(\beta) \mid H_u\} \nonumber \\
		& = \EE[\EE\{\xi_t(\beta) \mid H_t, A_t = 1\} p_t + \EE\{\xi_t(\beta) \mid H_t, A_t = 0\} (1 - p_t) \mid H_t] \nonumber \\
		& = I_t [\EE\{U_t(\beta)\mid H_t, A_t = 1\} - \EE\{U_t(\beta)\mid H_t, A_t = 0\} ] \partial_\beta \gamma(S_t; \beta). \label{A-eq:lem:projection-analytic-form:proofuse2}
	\end{align}
	\cref{A-eq:lem:projection-analytic-form:proofuse1} and \cref{A-eq:lem:projection-analytic-form:proofuse2} imply that when $u = t$,
	\begin{align*}
		& ~~~ \EE\{\xi_t(\beta) \mid H_u, A_u\} - \EE\{\xi_t(\beta) \mid H_u\} \\
		& = \frac{A_t - p_t}{p_t(1 - p_t)} I_t \Big[(1 - p_t) \EE\{U_t(\beta) \mid H_t, A_t = 1\} + p_t \EE\{U_t(\beta) \mid H_t, A_t = 0\} \Big] \partial_\beta \gamma(S_t; \beta).
	\end{align*}
	
	For $u < t$, under (WA-1) we have $\EE\{\xi_t(\beta) \mid H_u, A_u\} = \EE\{\xi_t(\beta) \mid H_u\}$, so $\EE\{\xi_t(\beta) \mid H_u, A_u\} - \EE\{\xi_t(\beta) \mid H_u\} = 0$.

	Putting the three cases together, we have
	\begin{align}
		& ~~~ \sum_{u=1}^T \Big[\EE\{\xi(\beta) \mid H_u, A_u\} - \EE\{\xi(\beta) \mid H_u\}\Big] \nonumber \\
		& = \sum_{t=1}^T \Big[\EE\{\xi_t(\beta) \mid H_t, A_t\} - \EE\{\xi_t(\beta) \mid H_t\}\Big] \nonumber \\
		& = \sum_{t=1}^T \frac{A_t - p_t}{p_t(1 - p_t)} I_t \Big[(1 - p_t) \EE\{U_t(\beta) \mid H_t, A_t = 1\} + p_t \EE\{U_t(\beta) \mid H_t, A_t = 0\} \Big] \partial_\beta \gamma(S_t; \beta).
	\end{align}
	\cref{A-eq:decomposition-xi-minus-projection} follows immediately. This completes the proof.
	
\end{proof}

With these lemmas, we are ready to prove \cref{thm:efficiency-bound}.

\begin{proof}[Proof of \cref{thm:efficiency-bound}]

\cref{A-lem:moment-restriction} implies that a score function for $\beta$ is $\xi(\beta)$ defined in \eqref{A-eq:primitive-ee-for-beta}. Because the efficient score must be orthogonal to the score functions for the treatment assignment probability, we replace $\xi(\beta)$ by itself minus its projection on the score functions for the treatment assignment probability. The projection is given by $\sum_{u=1}^T [\EE\{\xi(\beta) \mid H_u, A_u\} - \EE\{\xi(\beta) \mid H_u\}]$ \citep{robins1999testing}. \cref{A-lem:projection-analytic-form} implies that
\begin{align*}
	\xi(\beta) - \sum_{u=1}^T \Big[\EE\{\xi(\beta) \mid H_u, A_u\} - \EE\{\xi(\beta) \mid H_u\}\Big] = \sum_{t=1}^T \phi_t(\beta).
\end{align*}
Finally, due to (WA-2), Theorem 1 in \citet{chamberlain1992comment} implies that we can appropriately weight $\phi_t(\beta)$ to obtain the efficiency bound $\EE \{\partial_\beta \psi(\beta^\star)\} [\EE\{ \psi(\beta^\star)^{\otimes 2} \}]^{-1} \EE \{\partial_\beta \psi(\beta^\star)\}^T$.

\end{proof}

\subsection{Sufficient Conditions for the Working Assumptions}
\label{A-subsec:sufficient-conditions-WA}

This subsection provides the proof of the claim in \cref{rmk:sufficient-conditions-WA} that conditions C1--C3 are sufficient for (WA-1) and (WA-2), and elaborates on when each condition is plausible in MRT settings.

Let $\mu_a(H_t) := \EE(Y_{t+1} \mid H_t, A_t = a)$ for $a \in \{0,1\}$, $\epsilon_t := Y_{t+1} - \mu_{A_t}(H_t)$, and $\nu_a(H_t;\beta) := \EE\{U_t(\beta) \mid H_t, A_t = a\}$. For the identity link, $\nu_1 = \mu_1 - \gamma_t$ and $\nu_0 = \mu_0$. For the log link, $\nu_1 = e^{-\gamma_t}\mu_1$ and $\nu_0 = \mu_0$.

For convenience we restate C1--C3. 
\begin{itemize}[leftmargin=2.4em]
	\item[(C1)] For each $t$, $X_t(\ba_{t-1}) = X_t$ does not depend on $\ba_{t-1}$, and $Y_t(\ba_{t-1}) = Y_t(a_{t-1})$ depends on $\ba_{t-1}$ only through $a_{t-1}$.
	\item[(C2)] $p_t(H_t)$ depends on $H_t$ only through $\bX_t = (X_1, \ldots, X_t)$.
	\item[(C3)] $S_t$ is a function of $\bX_t = (X_1, \ldots, X_t)$ alone, and for almost every $H_t$ with $I_t = 1$, $g\{\mu_1(H_t)\} - g\{\mu_0(H_t)\} = \gamma_t(S_t; \beta^\star)$.
\end{itemize}

\begin{lem}[Sufficient conditions for (WA-1) and (WA-2)]
\label{A-lem:sufficient-WA}
Suppose \cref{asu:causal-assumptions} holds. If C1, C2, and C3 hold, then both (WA-1) and (WA-2) hold, for both identity and log links.
\end{lem}

\begin{proof}[Proof of \cref{A-lem:sufficient-WA}]

We prove (WA-1) and (WA-2) in turn. (WA-1) is a distributional statement about the observed-data process and makes no reference to $g$ or $\phi_t$, so the first part of the proof works for both links. The second part treats the two links separately.

\smallskip
\noindent\textbf{Part 1: C1, C2, and C3 imply (WA-1).}

Fix $1 \le u < t \le T$. By \cref{asu:consistency} and C1,
\begin{align}
X_v = X_v(\bA_{v-1}) = X_v(\ba_{v-1}) \quad \text{for all } \ba_{v-1} \text{ and all } v, \label{A-eq:lem:sufficient-WA:proofuse1}
\end{align}
that is, the covariate process is functionally exogenous to past treatments. By C1,
\begin{align}
Y_{t+1} = Y_{t+1}(\bA_t) = Y_{t+1}(A_t), \label{A-eq:lem:sufficient-WA:proofuse2}
\end{align}
that is, $Y_{t+1}$ depends on past treatments only through the current treatment $A_t$. By C2, $p_t(H_t)$ is a function of $\bX_t$, and so $A_t \mid H_t \sim \text{Bernoulli}\{p_t(\bX_t)\}$. By C3, $S_t$ is a function of $\bX_t$. By the manuscript's notation, $I_t$ is a component of $X_t$, hence a function of $\bX_t$. Combining,
\begin{align}
S_t, ~ I_t, ~ p_t(H_t) ~ \text{are functions of } \bX_t, \quad \text{and} \quad A_t \mid H_t ~\sim~ \text{Bernoulli}\{p_t(\bX_t)\}. \label{A-eq:lem:sufficient-WA:proofuse3}
\end{align}

By \eqref{A-eq:lem:sufficient-WA:proofuse1}, the conditional law of $\bX_t = (\bX_u, X_{u+1}, \ldots, X_t)$ given $H_u$ does not depend on $A_u$, since $\bX_u \subseteq H_u$ is fixed by conditioning and $X_{u+1}, \ldots, X_t$ are exogenous. Hence by \eqref{A-eq:lem:sufficient-WA:proofuse3}, the conditional laws of $S_t$, $I_t$, $p_t(H_t)$, and $A_t$ given $H_u$ do not depend on $A_u$. By \eqref{A-eq:lem:sufficient-WA:proofuse2}, $Y_{t+1}$ is a function of $A_t$ and exogenous potential-outcome randomness, so its conditional law given $H_u$ also does not depend on $A_u$. Therefore the joint conditional law of $\{S_t, I_t, A_t, p_t(H_t), Y_{t+1}\}$ given $H_u$ does not depend on $A_u$, which is (WA-1).

\smallskip
\noindent\textbf{Part 2: C3 implies (WA-2).}

The proof of (WA-2) uses only the moment-restriction part of C3, namely $g\{\mu_1(H_t)\} - g\{\mu_0(H_t)\} = \gamma_t(S_t;\beta^\star)$ for almost every $H_t$ with $I_t = 1$.

We first show that $\EE\{\phi_t(\beta^\star) \mid H_t, A_t\} = 0$, and then deduce (WA-2) by the law of iterated expectation. By the definition of $\nu_a$, $\EE\{U_t(\beta^\star) \mid H_t, A_t\} = \nu_{A_t}(H_t; \beta^\star)$. The factors $A_t - p_t$, $p_t$, $I_t$, $\partial_\beta \gamma_t(S_t;\beta^\star)$, $\nu_1(H_t;\beta^\star)$, and $\nu_0(H_t;\beta^\star)$ are all $\sigma(H_t, A_t)$-measurable, so by \cref{eq:phi-definition},
\begin{align}
\EE\{\phi_t(\beta^\star) \mid H_t, A_t\} &= \frac{A_t - p_t}{p_t(1-p_t)} \, I_t \, \big[\nu_{A_t}(H_t;\beta^\star) - (1-p_t)\, \nu_1(H_t;\beta^\star) - p_t\, \nu_0(H_t;\beta^\star)\big] \nonumber \\
&\quad \times \partial_\beta \gamma_t(S_t;\beta^\star). \label{A-eq:lem:sufficient-WA:proofuse4}
\end{align}
For both $A_t = 0$ and $A_t = 1$, a direct calculation (verifying the two cases separately) gives
\begin{align}
\frac{A_t - p_t}{p_t(1-p_t)} \, \big[\nu_{A_t} - (1-p_t)\, \nu_1 - p_t\, \nu_0\big] = \nu_1 - \nu_0. \label{A-eq:lem:sufficient-WA:proofuse5}
\end{align}
Substituting \eqref{A-eq:lem:sufficient-WA:proofuse5} into \eqref{A-eq:lem:sufficient-WA:proofuse4},
\begin{align}
\EE\{\phi_t(\beta^\star) \mid H_t, A_t\} = I_t \, \big[\nu_1(H_t;\beta^\star) - \nu_0(H_t;\beta^\star)\big] \, \partial_\beta \gamma_t(S_t;\beta^\star). \label{A-eq:lem:sufficient-WA:proofuse6}
\end{align}
For the identity link, $\nu_1 - \nu_0 = (\mu_1 - \gamma_t) - \mu_0 = \mu_1 - \mu_0 - \gamma_t$, which equals $0$ by C3. For the log link, $\nu_1 - \nu_0 = e^{-\gamma_t}\mu_1 - \mu_0 = e^{-\gamma_t}(\mu_1 - e^{\gamma_t}\mu_0)$, which equals $0$ by C3. Hence in both cases,
\begin{align}
\EE\{\phi_t(\beta^\star) \mid H_t, A_t\} = 0. \label{A-eq:lem:sufficient-WA:proofuse7}
\end{align}

For $u < t$, $\phi_u(\beta^\star)$ is a function of $(H_u, A_u, Y_{u+1})$, all of which are contained in $H_t$, so $\phi_u(\beta^\star)$ is $\sigma(H_t, A_t)$-measurable. By the law of iterated expectation,
\begin{align}
\EE\{\phi_t(\beta^\star)\,\phi_u(\beta^\star)^\top\}
& = \EE\big[\EE\{\phi_t(\beta^\star) \mid H_t, A_t\}\,\phi_u(\beta^\star)^\top\big] \label{A-eq:lem:sufficient-WA:proofuse8} \\
& = \EE\{0 \cdot \phi_u(\beta^\star)^\top\} = 0, \label{A-eq:lem:sufficient-WA:proofuse9}
\end{align}
where \eqref{A-eq:lem:sufficient-WA:proofuse8} follows by pulling the $\sigma(H_t, A_t)$-measurable factor $\phi_u(\beta^\star)^\top$ inside the conditional expectation, and \eqref{A-eq:lem:sufficient-WA:proofuse9} follows from \eqref{A-eq:lem:sufficient-WA:proofuse7}. This is (WA-2), which completes the proof.
\end{proof}

\subsubsection*{When MRTs satisfy C1--C3}

We elaborate on the design, measurement, and modeling features that govern whether each of C1--C3 is plausible in MRT settings.

\paragraph{C1: no carryover effects.}
\begin{itemize}[leftmargin=1.5em]
\item \emph{Plausible when} the intervention is brief and acts in the moment (e.g., a single prompt, reminder, or contextual cue) so its effect dissipates before subsequent decisions; covariates $X_t$ are dominated by exogenous context (e.g., time of day, weekday, weather) that is not affected by past treatments.
\item \emph{Violated when} the intervention is expected to produce persistent state changes such as engagement, fatigue, habituation, learning, or behavior change that compound over time; covariates summarize the cumulative effect of past interventions (e.g., recent app-use intensity, exposure-induced engagement state, accumulated dose).
\end{itemize}

\paragraph{C2: randomization probability not depending on past treatment.}
\begin{itemize}[leftmargin=1.5em]
\item \emph{Plausible when} the randomization probability is a function of covariates only, not of past treatments. This includes the common cases of a fixed randomization probability (e.g., $p_t = 0.4$ at all available decisions) and probabilities that depend on context covariates such as time of day or weekday.
\item \emph{Violated when} the design uses adaptive randomization in which $p_t$ is a function of past treatments $\bA_{t-1}$, e.g., reinforcement-learning policies or just-in-time-adaptive designs that increase or decrease treatment intensity or frequency based on past treatment history.
\end{itemize}

\paragraph{C3: moderator specified through covariate process captures all moderation.}
C3 has two parts: (i) $S_t$ is a function of $\bX_t$ alone, and (ii) the link-specific contrast $g\{\mu_1\} - g\{\mu_0\}$ collapses to $\gamma_t(S_t;\beta^\star)$.
\begin{itemize}[leftmargin=1.5em]
\item \emph{Plausible when} (i) the moderator of scientific interest is a function of context covariates not affected by past treatment (e.g., time-of-day, weekday, weather, baseline characteristics); and (ii) the chosen $S_t$ is rich enough to capture between- and within-individual heterogeneity in the link-specific contrast, so that the contrast is approximately constant across the components of $H_t$ outside $S_t$.
\item \emph{Violated when} $S_t$ is a post-treatment quantity such as a previous treatment indicator $A_{t-1}$, or a covariate whose value depends on the trajectory of past treatments (violating part (i); an example is $S_t = A_{t-1}$, in which case $S_t$ is not even independent of itself given $H_{t-1}$); or when $S_t$ is chosen for parsimony or interpretability and omits history components that genuinely modify the effect (violating part (ii)).
\end{itemize}

\section{Comparison with Existing Estimators}
\label{A-sec:comparison-existing-estimators}

\subsection{Proof of \texorpdfstring{\cref{thm:efficiency-comparison}}{Theorem efficiency-comparison}}
\label{A-subsec:efficiency-comparison-proof}

This subsection provides the proof of \cref{thm:efficiency-comparison} (asymptotic variance comparison between the proposed estimator $\hat\beta$ and the unadjusted estimator $\hat\beta_\text{unadj}$ for marginal CEE under the theorem's assumptions: $p_t = p \in (0,1)$ constant, $I_t \equiv 1$, no carryover effect). The proof proceeds via three lemmas: \cref{A-lem:meat-matters} shows the bread $A$ is augmentation-invariant, so only the meat $B$ matters; \cref{A-lem:cross-t-vanishing} shows the meat difference $B(0) - B(\mu)$ decomposes into a sum of per-$t$ contributions with no cross-time correction; and \cref{A-lem:per-t-loewner} shows each per-$t$ contribution is nonnegative under the OLS fit. The theorem follows by combining the three lemmas. \Cref{A-subsec:counter-example-carryover} gives a counter-example demonstrating that the no-carryover assumption cannot be dropped.

\paragraph{Setup and notation.} Under the theorem's assumptions ($p_t = p$ constant, $I_t \equiv 1$, marginal CEE so $\gamma_t(S_t;\beta) = \beta \in \RR$ and $\partial_\beta\gamma_t = 1$), and taking the weight $d_t \equiv 1$ as in the theorem statement, the per-decision-point estimating function (identity link, from \cref{eq:phi-definition}) decomposes as
\begin{equation*}
\phi_t(\beta, \mu_t) = \underbrace{\frac{A_t - p}{p(1-p)}}_{W_t}\, \underbrace{\Bigl[\, Y_{t+1} - (A_t + p - 1)\beta - \overbrace{\bigl\{(1-p)\,\mu_t(H_t, 1) + p\, \mu_t(H_t, 0)\bigr\}}^{b_t(H_t;\,\mu_t)} \,\Bigr]}_{R_t(\beta,\, \mu_t)}.
\end{equation*}
Formally,
\begin{align*}
\phi_t(\beta, \mu_t) &= W_t \cdot R_t(\beta, \mu_t), \\
W_t &:= \frac{A_t - p}{p(1-p)}, \\
R_t(\beta, \mu_t) &= Y_{t+1} - (A_t + p - 1)\beta - (1-p)\mu_t(H_t,1) - p\, \mu_t(H_t,0), \\
b_t(H_t; \mu_t) &:= (1-p)\,\mu_t(H_t, 1) + p\, \mu_t(H_t, 0).
\end{align*}
For the log link, replace $R_t$ by $\tilde R_t(\beta, \mu_t) := e^{-A_t \beta}\, Y_{t+1} - \tilde b_t(H_t; \mu_t)$ with $\tilde b_t(H_t; \mu_t) := (1-p)\, e^{-\beta}\, \mu_t(H_t, 1) + p\, \mu_t(H_t, 0)$. Let $\mu_t^\star(H_t, A_t) := \EE(Y_{t+1} \mid H_t, A_t)$ be the true conditional mean, $b_t^\star := b_t(H_t; \mu_t^\star)$, $\tilde b_t^\star := \tilde b_t(H_t; \mu_t^\star)$, and $\rho := e^{-\beta^\star}$. Define
\begin{align*}
\xi(\mu) := \sum_{t=1}^T \phi_t(\beta^\star, \mu_t), \qquad B(\mu) := \PP\{\xi(\mu)^2\}, \qquad A := \PP\Bigl\{\textstyle\sum_t \partial_\beta \phi_t(\beta^\star, \mu_t)\Bigr\}.
\end{align*}
Under \cref{thm:normality-ncf}, the proposed estimator (using $\hat\mu_t$ with $L_2$ limit $\mu_t'$) has asymptotic variance $V(\hat\beta(\mu')) = B(\mu')/A^2$, and the unadjusted estimator has $V(\hat\beta_\text{unadj}) = V(\hat\beta(\mu = 0)) = B(0)/A^2$ (\cref{A-lem:meat-matters} below shows $A$ does not depend on $\mu$). The conditional moments $\EE(W_t \mid H_t) = 0$ and $\EE(W_t^2 \mid H_t) = 1/[p(1-p)]$ are used throughout.

\begin{lem}[Only the meat matters]
\label{A-lem:meat-matters}
For both the identity and log links, $\PP\{\sum_t \phi_t(\beta, \mu_t)\}$ and $A$ do not depend on $\mu$. Hence the population zero $\beta^\star$ and the bread $A$ are $\mu$-invariant, and $V(\hat\beta(\mu)) \leq V(\hat\beta(\tilde\mu))$ iff $B(\mu) \leq B(\tilde\mu)$.
\end{lem}

\begin{proof}
\emph{Identity link.} $\phi_t(\beta, \mu_t) - \phi_t(\beta, 0) = -W_t\, b_t(H_t; \mu_t)$, with $b_t$ $H_t$-measurable. Since $\EE(W_t \mid H_t) = 0$, $\PP\{W_t b_t\} = \EE\{b_t\, \EE(W_t \mid H_t)\} = 0$, giving the first claim. For $A$, $\partial_\beta b_t(H_t; \mu_t) = 0$, so $\partial_\beta \phi_t$ does not depend on $\mu$.

\emph{Log link.} Same first claim with $\tilde b_t$ replacing $b_t$. For $A$, $\partial_\beta \tilde b_t = -(1-p)\, e^{-\beta}\, \mu_t(H_t, 1)$ is $H_t$-measurable, so $\partial_\beta\{W_t \tilde b_t\}$ is the product of $W_t$ with an $H_t$-measurable quantity, and its expectation is zero by $\EE(W_t \mid H_t) = 0$.
\end{proof}

\begin{lem}[Cross-time terms in $B(0) - B(\mu)$ vanish]
\label{A-lem:cross-t-vanishing}
For any $\mu$,
\begin{align}
B(0) - B(\mu) = \frac{1}{p(1-p)} \sum_{t=1}^T \PP\bigl\{2\, b_t(H_t; \mu_t)\, b_t^\star(H_t) - b_t(H_t; \mu_t)^2\bigr\}, \label{A-eq:meat-decomp-simplified}
\end{align}
and analogously with $b_t, b_t^\star$ replaced by $\tilde b_t, \tilde b_t^\star$ for the log link.
\end{lem}

\begin{proof}
We give the identity-link argument; the log link follows by the same steps with $b_t$ replaced by $\tilde b_t$.

Since $\xi(\mu) = \xi(0) - \sum_t W_t\, b_t(H_t; \mu_t)$, expanding the square gives
\begin{align}
B(0) - B(\mu)
&= \PP\{\xi(0)^2\} - \PP\Bigl\{\bigl(\xi(0) - \textstyle\sum_t W_t b_t\bigr)^2\Bigr\} \notag \\
&= 2\, \PP\Bigl\{\xi(0)\, \textstyle\sum_t W_t b_t\Bigr\} - \PP\Bigl\{\bigl(\textstyle\sum_t W_t b_t\bigr)^2\Bigr\}. \label{A-eq:meat-decomp-simplified-proofuse1}
\end{align}

\emph{Variance term.} Expand
\begin{align*}
\PP\Bigl\{\bigl(\textstyle\sum_t W_t b_t\bigr)^2\Bigr\} = \sum_{t=1}^T \PP\{W_t^2 b_t^2\} + 2\sum_{s<t} \PP\{W_s\, b_s\, W_t\, b_t\}.
\end{align*}
For $s < t$, $W_s b_s b_t$ is $H_t$-measurable, so by $\EE(W_t \mid H_t) = 0$ the off-diagonal terms vanish. Each diagonal term equals $\PP\{b_t^2\, \EE(W_t^2 \mid H_t)\} = \frac{1}{p(1-p)}\, \PP\{b_t^2\}$.

\emph{Cross term.} Expand
\begin{align*}
\PP\Bigl\{\xi(0)\, \textstyle\sum_t W_t b_t\Bigr\} = \sum_{s, t} \PP\{\phi_s(\beta^\star, 0)\, W_t\, b_t\}.
\end{align*}

\smallskip\noindent\emph{Case $s = t$.} Substituting $\phi_t(\beta^\star, 0) = W_t \bigl[Y_{t+1} - (A_t + p - 1)\beta^\star\bigr]$, the diagonal is $\PP\bigl\{W_t^2\, b_t\, [Y_{t+1} - (A_t + p - 1)\beta^\star]\bigr\}$. Conditioning on $H_t$ and summing over $A_t \in \{0,1\}$ with $P(A_t = 1 \mid H_t) = p$ and $\EE(Y_{t+1} \mid H_t, A_t = a) = \mu_t^\star(H_t, a)$:
\begin{align*}
& \EE\bigl\{(A_t - p)^2\, [Y_{t+1} - (A_t + p - 1)\beta^\star] \,\big|\, H_t\bigr\} \\
&\quad= p\,(1-p)^2\bigl[\mu_t^\star(H_t, 1) - p\beta^\star\bigr] + (1-p)\, p^2 \bigl[\mu_t^\star(H_t, 0) + (1-p)\beta^\star\bigr] \\
&\quad= \bigl[p(1-p)^2 \mu_t^\star(H_t, 1) + p^2(1-p)\mu_t^\star(H_t, 0)\bigr] + \beta^\star \bigl[-p^2(1-p)^2 + p^2(1-p)^2\bigr] \\
&\quad= p(1-p)\bigl[(1-p)\mu_t^\star(H_t, 1) + p\, \mu_t^\star(H_t, 0)\bigr] \\
&\quad= p(1-p)\, b_t^\star(H_t),
\end{align*}
where the second line is the $A_t = 1$ and $A_t = 0$ contributions, and the $\beta^\star$ coefficient cancels in the third line. Hence $\EE\bigl\{W_t^2\, [Y_{t+1} - (A_t + p - 1)\beta^\star] \,\big|\, H_t\bigr\} = b_t^\star / [p(1-p)]$, and the diagonal equals $\frac{1}{p(1-p)}\, \PP\{b_t\, b_t^\star\}$.

\smallskip\noindent\emph{Case $s < t$.} $\phi_s(\beta^\star, 0)\, b_t$ is $H_t$-measurable; $\EE(W_t \mid H_t) = 0$ gives zero.

\smallskip\noindent\emph{Case $s > t$.} Under no carryover and constant $p$, the joint distribution of $(X_{t+1}, A_{t+1}, \ldots,\allowbreak X_s, A_s, Y_{s+1})$ given $(H_t, A_t)$ does not depend on $A_t$: no carryover ensures $X_u = X_u(\bar a_{u-1})$ does not depend on $\bar a_{u-1}$ for $u > t$ (so does not depend on $a_t$), and $Y_{u+1}(\bar a_u) = Y_{u+1}(a_u)$ for $u > t$ (so does not depend on $a_t$); constant $p$ ensures $p_u$ does not depend on past treatments. Since $\phi_s(\beta^\star, 0)$ is a function of $(H_s, A_s, Y_{s+1})$, $\EE\{\phi_s(\beta^\star, 0) \mid H_t, A_t\}$ does not depend on $A_t$, i.e., is $H_t$-measurable. The factor $W_t b_t$ is $(H_t, A_t)$-measurable, so applying the tower property first conditioning on $(H_t, A_t)$ and then on $H_t$,
\begin{align*}
\PP\{\phi_s(\beta^\star, 0)\, W_t\, b_t\}
&= \EE\bigl\{W_t\, b_t \cdot \EE\{\phi_s(\beta^\star, 0) \mid H_t, A_t\}\bigr\} \\
&= \EE\bigl\{b_t \cdot \EE\{\phi_s(\beta^\star, 0) \mid H_t, A_t\} \cdot \EE(W_t \mid H_t)\bigr\} = 0,
\end{align*}
since $\EE(W_t \mid H_t) = 0$.

\smallskip
Summing the three cases gives $\PP\{\xi(0)\, \sum_t W_t b_t\} = \frac{1}{p(1-p)} \sum_t \PP\{b_t\, b_t^\star\}$. Substituting into \eqref{A-eq:meat-decomp-simplified-proofuse1},
\begin{align*}
B(0) - B(\mu)
&= \frac{2}{p(1-p)} \sum_t \PP\{b_t\, b_t^\star\} - \frac{1}{p(1-p)} \sum_t \PP\{b_t^2\} \\
&= \frac{1}{p(1-p)} \sum_t \PP\{2\, b_t\, b_t^\star - b_t^2\}.
\end{align*}
\end{proof}

\begin{lem}[Per-$t$ bound under OLS fit]
\label{A-lem:per-t-loewner}
Let $\hat\mu_t$ be the OLS fit of $Y_{t+1}$ on the linear space
\begin{align*}
\mathcal{M}_t = \bigl\{ c_1 + c_2 A_t + g_1(H_t) + A_t\, g_2(H_t) : c_1, c_2 \in \RR,\; g_1, g_2 \in \mathcal{G}_t \bigr\},
\end{align*}
with $\mathcal{G}_t$ a linear space of $H_t$-measurable functions, and let $\mu_t'$ denote its $L_2$ limit. Set $b_t' := b_t(H_t; \mu_t')$ and $\tilde b_t' := \tilde b_t(H_t; \mu_t')$. Then for each $t$,
\begin{align*}
\PP\{2\, b_t'\, b_t^\star - (b_t')^2\} \geq 0,
\end{align*}
and analogously with $b_t', b_t^\star$ replaced by $\tilde b_t', \tilde b_t^\star$ for the log link.
\end{lem}

\begin{proof}
\emph{Identity link.} Define $\mathcal{B}_t := \{(1-p)\,\mu(H_t, 1) + p\, \mu(H_t, 0) : \mu \in \mathcal{M}_t\}$, a linear space (closed under linear combinations since $\mathcal{M}_t$ is) containing $b_t'$. Write $\Delta_a := (\mu_t^\star - \mu_t')(H_t, a)$ for $a \in \{0,1\}$, so that $b_t^\star - b_t' = (1-p)\Delta_1 + p\Delta_0$.

\emph{Reduction.} We claim it suffices to show that $b_t'$ is the $L_2$ projection of $b_t^\star$ onto $\mathcal{B}_t$, i.e.,
\begin{align}
\PP\{(b_t^\star - b_t')\, v\} = 0 \quad \text{for every } v \in \mathcal{B}_t. \label{A-eq:bt-projection}
\end{align}
Indeed, since $b_t' \in \mathcal{B}_t$, taking $v = b_t'$ in \eqref{A-eq:bt-projection} gives $\PP\{b_t'\, b_t^\star\} = \PP\{(b_t')^2\}$, hence
\begin{align*}
\PP\{2\, b_t'\, b_t^\star - (b_t')^2\} = 2\,\PP\{(b_t')^2\} - \PP\{(b_t')^2\} = \PP\{(b_t')^2\} \geq 0.
\end{align*}

\emph{Proof of \eqref{A-eq:bt-projection}.} The OLS fit $\mu_t'$ satisfies the first-order condition
\begin{align}
\PP\{(\mu_t^\star - \mu_t')\, h\} = 0 \quad \text{for every } h \in \mathcal{M}_t. \label{A-eq:ols-foc}
\end{align}
We will construct, for each $v \in \mathcal{B}_t$, an $h \in \mathcal{M}_t$ such that $\PP\{(\mu_t^\star - \mu_t')\, h\} = \PP\{(b_t^\star - b_t')\, v\}$. Combined with \eqref{A-eq:ols-foc}, this gives \eqref{A-eq:bt-projection}. Below is the construction.

Fix $v \in \mathcal{B}_t$. By definition of $\mathcal{B}_t$, $v = (1-p)\mu(H_t, 1) + p\,\mu(H_t, 0)$ for some $\mu \in \mathcal{M}_t$, which we expand as $\mu = c_1 + c_2 A_t + g_1(H_t) + A_t\, g_2(H_t)$. Then
\begin{align*}
v
= (1-p)\bigl[c_1 + c_2 + g_1 + g_2\bigr] + p\bigl[c_1 + g_1\bigr]
= \bigl[c_1 + (1-p)c_2\bigr] + \bigl[g_1 + (1-p)g_2\bigr](H_t),
\end{align*}
so $v$ depends on $H_t$ only and has the form $v = c + g(H_t)$ with $c := c_1 + (1-p)c_2 \in \RR$ and $g := g_1 + (1-p)g_2 \in \mathcal{G}_t$ (using that $\mathcal{G}_t$ is closed under linear combinations).

Define
\begin{align*}
h(H_t, A_t) := \tfrac{p}{1-p}\, v\, (1 - A_t) + \tfrac{1-p}{p}\, v\, A_t.
\end{align*}
We verify $h \in \mathcal{M}_t$ by substituting $v = c + g$ and rearranging:
\begin{align*}
h
&= \tfrac{p}{1-p}(c + g) + A_t\,\Bigl[\tfrac{1-p}{p} - \tfrac{p}{1-p}\Bigr](c + g) \\
&= \underbrace{\tfrac{p}{1-p}c}_{\in \RR} + \underbrace{\Bigl[\tfrac{1-p}{p} - \tfrac{p}{1-p}\Bigr]c}_{\in \RR}\, A_t \\
&\quad + \underbrace{\tfrac{p}{1-p}g(H_t)}_{\in \mathcal{G}_t} + A_t\, \underbrace{\Bigl[\tfrac{1-p}{p} - \tfrac{p}{1-p}\Bigr]g(H_t)}_{\in \mathcal{G}_t},
\end{align*}
which is of the form $c_1' + c_2' A_t + g_1'(H_t) + A_t\, g_2'(H_t)$ with $c_1', c_2' \in \RR$ and $g_1', g_2' \in \mathcal{G}_t$ (using $\mathcal{G}_t$'s closure under scalar multiplication), hence $h \in \mathcal{M}_t$.

Now compute $\PP\{(\mu_t^\star - \mu_t')\, h\}$ by conditioning on $H_t$, using $P(A_t = 1 \mid H_t) = p$:
\begin{align*}
\PP\{(\mu_t^\star - \mu_t')\, h\}
&= \EE\Bigl\{p\, \Delta_1\, h(H_t, 1) + (1-p)\, \Delta_0\, h(H_t, 0)\Bigr\} \\
&= \EE\Bigl\{p\, \Delta_1 \cdot \tfrac{1-p}{p}\, v + (1-p)\, \Delta_0 \cdot \tfrac{p}{1-p}\, v\Bigr\} \\
&= \EE\bigl\{v\, [(1-p)\Delta_1 + p\,\Delta_0]\bigr\} \\
&= \PP\{(b_t^\star - b_t')\, v\}.
\end{align*}
Combined with \eqref{A-eq:ols-foc}, $\PP\{(b_t^\star - b_t')\, v\} = 0$, establishing \eqref{A-eq:bt-projection}.

\emph{Log link.} The argument is the same with $\tilde{\mathcal{B}}_t := \{(1-p)\rho\, \mu(H_t, 1) + p\, \mu(H_t, 0) : \mu \in \mathcal{M}_t\}$ in place of $\mathcal{B}_t$ (linear, containing $\tilde b_t'$). For $\tilde v \in \tilde{\mathcal{B}}_t$, the lift is
\begin{align*}
h(H_t, A_t) := \tfrac{p}{1-p}\, \tilde v\, (1 - A_t) + \tfrac{1-p}{p}\, \rho\, \tilde v\, A_t,
\end{align*}
which lies in $\mathcal{M}_t$ since $\rho$ is a scalar constant (analogous expansion to the identity-link case). The same conditioning calculation gives $\PP\{(\mu_t^\star - \mu_t')\, h\} = \PP\{\tilde v\, (\tilde b_t^\star - \tilde b_t')\}$ (using $\tilde b_t^\star - \tilde b_t' = (1-p)\rho\Delta_1 + p\Delta_0$), so the OLS FOC \eqref{A-eq:ols-foc} gives $\PP\{(\tilde b_t^\star - \tilde b_t')\, \tilde v\} = 0$ for every $\tilde v \in \tilde{\mathcal{B}}_t$. The reduction argument applied to $\tilde b_t'$ yields $\PP\{2\,\tilde b_t'\, \tilde b_t^\star - (\tilde b_t')^2\} \geq 0$.
\end{proof}

\paragraph{Proof of \cref{thm:efficiency-comparison}.}

\begin{proof}[Proof of \cref{thm:efficiency-comparison}]
By \cref{A-lem:meat-matters}, $V(\hat\beta) \leq V(\hat\beta_\text{unadj})$ iff $B(\mu') \leq B(0)$. By \cref{A-lem:cross-t-vanishing} at $\mu = \mu'$ combined with \cref{A-lem:per-t-loewner},
\begin{align*}
B(0) - B(\mu') = \frac{1}{p(1-p)} \sum_{t=1}^T \PP\{2\, b_t'\, b_t^\star - (b_t')^2\} \geq 0,
\end{align*}
for the identity link. The log-link case is verified by the analogous identities with $b_t, b_t^\star, b_t'$ replaced by $\tilde b_t, \tilde b_t^\star, \tilde b_t'$ throughout.
\end{proof}

\subsection{Counter-Example in the Presence of Carryover Effects}
\label{A-subsec:counter-example-carryover}

This subsection presents a $T = 2$ MRT with a genuine time-varying covariate and a carryover effect, in which the proposed estimator $\hat\beta$ (with the optimal weights and the true conditional means), DR-WCLS, and WCLS each have strictly larger asymptotic variance than the unadjusted estimator $\hat\beta_\text{unadj}$. This shows that the no-carryover assumption in \cref{thm:efficiency-comparison} cannot be dropped (cf.\ \cref{rmk:no-harm-fails-without-c1}).

\paragraph{Construction.} Take $T = 2$, randomization probability $p = 0.5$, and availability $I_t \equiv 1$. Generate exogenous time-varying covariates $X_1, X_2 \stackrel{\text{i.i.d.}}{\sim} \mathrm{Normal}(0, \sigma_X^2)$, treatments $A_1, A_2 \stackrel{\text{i.i.d.}}{\sim} \mathrm{Bernoulli}(0.5)$, and errors $\varepsilon_1, \varepsilon_2 \stackrel{\text{i.i.d.}}{\sim} \mathrm{Normal}(0, \sigma^2)$, all mutually independent. Set $H_1 = X_1$, $H_2 = (X_1, A_1, Y_2, X_2)$, and
\begin{align}
Y_2 &= \alpha X_1 + A_1 + \varepsilon_1, \nonumber \\
Y_3 &= \alpha X_2 + \theta\, (A_1 - 0.5)(A_2 - 0.5) + A_2 + \varepsilon_2, \label{A-eq:Y3-counter}
\end{align}
for real parameters $\alpha$ and $\theta$. Marginalizing over $X_t$ and $A_1$ gives $\EE(Y_2 \mid A_1 = a) = a$ and $\EE(Y_3 \mid A_2 = a) = a$, so the marginal CEE at each decision point equals $1$ and the working model $\gamma_t(\emptyset; \beta) = \beta$ has true value $\beta^\star = 1$. The covariate $X_t$ is genuinely informative for $Y_{t+1}$ whenever $\alpha \neq 0$; the no-carryover assumption fails when $\theta \neq 0$, as $A_1$ has a delayed effect on $Y_3$ through the interaction in \eqref{A-eq:Y3-counter}.

The true conditional means are
\begin{align*}
\EE(Y_2 \mid H_1, A_1 = a) = \alpha X_1 + a \quad \text{and} \quad \EE(Y_3 \mid H_2, A_2 = a) = \alpha X_2 + \theta(A_1 - 0.5)(a - 0.5) + a.
\end{align*}
For each $t$, the within-$t$ average of the conditional means is $0.5\, \EE(Y_{t+1} \mid H_t, A_t = 1) + 0.5\, \EE(Y_{t+1} \mid H_t, A_t = 0) = \alpha X_t + 0.5$; this serves as the ideal augmentation for DR-WCLS and the proposed $\hat\beta$. Crucially, the carryover term $\theta(A_1 - 0.5)(a - 0.5)$ in the second conditional mean averages to zero over $a \in \{0,1\}$ and is therefore \emph{not} removed by the augmentation, even when the true $\mu_t^\star$ is used.

\paragraph{Estimators.} We compare four estimators of $\beta^\star = 1$:
\begin{itemize}
\item \textbf{Unadjusted} $\hat\beta_\text{unadj}$: solves $\sum_{t=1}^2 4(A_t - 0.5)\{Y_{t+1} - (A_t - 0.5)\beta\} = 0$.
\item \textbf{WCLS} with intercept and $X_t$ as a working control: includes $\alpha_0 + \alpha_X X_t$ in the model. Since the within-$t$ average $\alpha X_t + 0.5$ is linear in $X_t$, this working control is correctly specified, and WCLS is asymptotically equivalent to DR-WCLS in this construction.
\item \textbf{DR-WCLS} with the true conditional means substituted into the augmentation: solves $\sum_{t=1}^2 4(A_t - 0.5)\{Y_{t+1} - (A_t - 0.5)\beta - \alpha X_t - 0.5\} = 0$.
\item \textbf{Proposed} $\hat\beta$ with the true conditional means and the optimal weights $d_t^\star \propto 1/\Var\{4(A_t - 0.5)[Y_{t+1} - (A_t - 0.5)\beta^\star - \alpha X_t - 0.5]\}$.
\end{itemize}

\paragraph{Asymptotic variances.} Let $U_t$ and $\tilde U_t$ denote the per-decision-point unadjusted and augmented scores at $\beta^\star = 1$, respectively. Substituting the DGM and the true $\mu_t^\star$ and using $(A_t - 0.5)^2 = 0.25$ for $A_t \in \{0,1\}$:
\begin{align*}
U_1 &= 4(A_1 - 0.5)(\alpha X_1 + 0.5 + \varepsilon_1), &
\tilde U_1 &= 4(A_1 - 0.5)\varepsilon_1, \\
U_2 &= \theta(A_1 - 0.5) + 4(A_2 - 0.5)(\alpha X_2 + 0.5 + \varepsilon_2), &
\tilde U_2 &= \theta(A_1 - 0.5) + 4(A_2 - 0.5)\varepsilon_2.
\end{align*}
The carryover term $\theta(A_1 - 0.5)$ persists in $\tilde U_2$ because it averages to zero over $A_2$ but not over $A_1$. Using the mutual independence of $X_t, A_t, \varepsilon_t$ across $t$,
\begin{align*}
\Var(U_t) &= 4\alpha^2\sigma_X^2 + 4\sigma^2 + 1 + 0.25\theta^2\, \mathbf{1}\{t = 2\}, & \Var(\tilde U_t) &= 4\sigma^2 + 0.25\theta^2\, \mathbf{1}\{t = 2\}, \\
\Cov(U_1, U_2) &= 0.5\theta, & \Cov(\tilde U_1, \tilde U_2) &= 0.
\end{align*}
With the bread $A = -\sum_t \PP\{4(A_t - 0.5)^2\} = -2$ for the unadjusted estimator, WCLS, and DR-WCLS (all of which use $d_t \equiv 1$), the asymptotic variances are
\begin{align*}
V(\hat\beta_\text{unadj}) &= \tfrac{1}{4}\bigl[\Var(U_1) + \Var(U_2) + 2\Cov(U_1, U_2)\bigr] \\
&= 2\alpha^2\sigma_X^2 + 2\sigma^2 + 0.5 + 0.0625\theta^2 + 0.25\theta, \\
V(\hat\beta_\text{DR-WCLS}) = V(\hat\beta_\text{WCLS}) &= \tfrac{1}{4}\bigl[\Var(\tilde U_1) + \Var(\tilde U_2)\bigr] = 2\sigma^2 + 0.0625\theta^2, \\
V(\hat\beta) &= \bigl[1/\Var(\tilde U_1) + 1/\Var(\tilde U_2)\bigr]^{-1} \\
&= \bigl[1/(4\sigma^2) + 1/(4\sigma^2 + 0.25\theta^2)\bigr]^{-1}.
\end{align*}
Evaluating at $\alpha = 0.2$, $\sigma_X = \sigma = 1$, $\theta = -4$ yields $V(\hat\beta_\text{unadj}) = 2.580$, $V(\hat\beta_\text{DR-WCLS}) = V(\hat\beta_\text{WCLS}) = 3.000$, and $V(\hat\beta) = 8/3 \approx 2.667$, the theoretical values reported in \cref{A-tab:counterexample-sim}. The negative cross-time covariance $\Cov(U_1, U_2) = 0.5\theta < 0$ helps the unadjusted estimator by reducing $\Var(U_1 + U_2)$; augmentation eliminates this covariance ($\Cov(\tilde U_1, \tilde U_2) = 0$) but pays the price of the persisting $\theta(A_1 - 0.5)$ term in $\tilde U_2$. For moderately negative $\theta$, the cross-time cancellation dominates and the unadjusted estimator becomes strictly more efficient than all three covariate-adjusted alternatives, even though all three correctly adjust for $X_t$.

\paragraph{Simulation.} We take $\alpha = 0.2$, $\sigma_X = 1$, $\sigma = 1$, $\theta = -4$, sample size $n = 2{,}000$, and $50{,}000$ Monte Carlo replications. \Cref{A-tab:counterexample-sim} reports the empirical $n \cdot \widehat\Var(\hat\beta)$ alongside the theoretical asymptotic variances. All four estimators are essentially unbiased (absolute empirical bias below $10^{-3}$), and the empirical variances agree with the theoretical values to within one Monte Carlo standard error. All three covariate-adjusted estimators have strictly larger variance than the unadjusted estimator: the proposed estimator is $3.0\%$ less efficient ($2.66$ vs.\ $2.58$), while WCLS and DR-WCLS are about $16\%$ less efficient.

\begin{table}[h]
\centering
\caption{Empirical and theoretical asymptotic variances for the counter-example at $\alpha = 0.2$, $\sigma_X = 1$, $\sigma = 1$, $\theta = -4$, $n = 2{,}000$, and $50{,}000$ Monte Carlo replications. ``MCSE'' is the Monte Carlo standard error of $n \cdot \widehat\Var$. The unadjusted estimator is the most efficient.}
\label{A-tab:counterexample-sim}
\begin{tabular}{lcccc}
\hline
Estimator & $n \cdot \widehat\Var(\hat\beta)$ & MCSE & Theoretical $V$ & RE vs.\ unadj.\ \\
\hline
Unadjusted & $2.582$ & $0.016$ & $2.580$ & $1.000$ \\
WCLS (control $= X$) & $2.996$ & $0.019$ & $3.000$ & $0.862$ \\
DR-WCLS (true $\mu^\star$) & $2.995$ & $0.019$ & $3.000$ & $0.862$ \\
Proposed (true $\mu^\star$, optimal $d_t^\star$) & $2.660$ & $0.017$ & $2.667$ & $0.971$ \\
\hline
\end{tabular}
\end{table}

The simulation confirms that all three covariate-adjusted estimators are strictly less efficient than the unadjusted estimator, even though they correctly use the genuine covariate $X_t$ in their working models or augmentations.

\subsection{Equivalences Underlying \texorpdfstring{\cref{tab:estimator-comparison}}{Table 1}}
\label{A-subsec:table-equivalences}

This subsection verifies that DR-WCLS \citep{shi2023meta}, DR-EMEE \citep{liu2024incorporating}, WCLS \citep{boruvka2018}, EMEE \citep{qian2021estimating}, and the unadjusted estimator are special cases of the proposed framework
\begin{align*}
\PP_n \sum_{t=1}^T d_t(S_t;\mu_t,\beta)\,\phi_t(\beta,\mu_t) = 0, \qquad \phi_t(\beta,\mu_t) = \frac{A_t-p_t}{p_t(1-p_t)}\, I_t\, R_t\, \partial_\beta\gamma_t(S_t;\beta),
\end{align*}
with the $(d_t, R_t)$ pairs listed in \cref{tab:estimator-comparison}. All five verifications rest on the algebraic identity below.

\begin{lem}[Stabilized weight identity]\label{A-lem:stab-identity}
For $A_t \in \{0,1\}$, the stabilized weight $W_t = \tilde p_t(A_t \mid S_t)/p_t(A_t \mid H_t)$ satisfies
\begin{align}
& W_t(A_t - \tilde p_t) = \tilde p_t(1-\tilde p_t) \cdot \frac{A_t - p_t}{p_t(1-p_t)}, \label{A-eq:stab-i} \\
& W_t(A_t - \tilde p_t)(A_t + p_t - 1) = \tilde p_t(1-\tilde p_t), \label{A-eq:stab-ii}
\end{align}
where $\tilde p_t = \tilde p_t(1 \mid S_t)$ and $p_t = p_t(1 \mid H_t)$.
\end{lem}

\begin{proof}
For $A_t = 1$, $W_t = \tilde p_t/p_t$, so $W_t(A_t - \tilde p_t) = \tilde p_t(1-\tilde p_t)/p_t$, which equals $\tilde p_t(1-\tilde p_t)\cdot(1-p_t)/[p_t(1-p_t)]$. For $A_t = 0$, $W_t = (1-\tilde p_t)/(1-p_t)$, so $W_t(A_t - \tilde p_t) = -\tilde p_t(1-\tilde p_t)/(1-p_t) = \tilde p_t(1-\tilde p_t)\cdot(-p_t)/[p_t(1-p_t)]$. This proves \eqref{A-eq:stab-i}. For \eqref{A-eq:stab-ii}, expand $(A_t-p_t)(A_t+p_t-1) = A_t^2 - A_t + p_t(1-p_t) = p_t(1-p_t)$ since $A_t^2 = A_t$ for $A_t \in \{0,1\}$, then combine with \eqref{A-eq:stab-i}.
\end{proof}

We now verify the equivalences in the order they appear in \cref{tab:estimator-comparison}. Throughout, we let $\gamma_t = \gamma_t(S_t;\beta)$ for brevity, $\mu_a := \mu_t(H_t,a)$ for $a \in \{0,1\}$, $\partial_\beta \gamma_t = f_t(S_t)$ in the linear-model case $\gamma_t = f_t(S_t)^\top\beta$, and adopt the convention that the source-paper notation matches ours according to: $g_t(H_t) \leftrightarrow c_t(H_t)$ (parametric control variable), $g_t(H_t,a) \leftrightarrow \mu_t(H_t,a)$ (outcome regression), and $f_t(S_t) \leftrightarrow \partial_\beta\gamma_t$.

\paragraph{Proposed estimator ($\hat\beta$, $\check\beta$).}
The non-cross-fitted estimator $\hat\beta$ (\cref{algo:estimator-ncf} of the main paper) and the cross-fitted estimator $\check\beta$ (\cref{A-algo:estimator-cf} of this appendix) both solve $\PP_n \sum_t \hat d_t(S_t)\,\phi_t(\beta,\hat\mu_t) = 0$ by construction. Plugging $\mu_t = \mu_t^\star$ and $d_t = d_t^\star$ into \cref{eq:phi-definition} yields the rows labeled ``Proposed'' in \cref{tab:estimator-comparison}: $R_t = Y_{t+1} - (A_t + p_t - 1)\gamma_t - (1-p_t)\mu_1 - p_t\mu_0$ for the identity link and $R_t = e^{-A_t\gamma_t}Y_{t+1} - (1-p_t)e^{-\gamma_t}\mu_1 - p_t\mu_0$ for the log link.

\paragraph{DR-WCLS \citep{shi2023meta}.}
Equations~(5)--(6) of \citet{shi2023meta} define the DR-WCLS estimator as the solution to $\PP_n[\sum_t \psi_t] = 0$ with
\begin{align}
\psi_t = \tilde\sigma_t^2(S_t)\Big(\frac{W_t(A_t - \tilde p_t)\{Y_{t+1} - g_t(H_t,A_t)\}}{\tilde\sigma_t^2(S_t)} + \beta(t;H_t) - f_t(S_t)^\top\beta\Big) f_t(S_t), \label{A-eq:dr-wcls-original}
\end{align}
where $g_t(H_t,a) = \EE[Y_{t+1} \mid H_t, A_t = a]$, $\beta(t;H_t) = g_t(H_t,1) - g_t(H_t,0) = \mu_1 - \mu_0$, $\tilde\sigma_t^2(S_t) = \tilde p_t(1-\tilde p_t)$, and $\gamma_t = f_t(S_t)^\top\beta$. \citet{shi2023meta} consider the always-available case ($I_t \equiv 1$); the same derivation goes through with the indicator inserted, recovering the row labeled ``DR-WCLS''.

We claim that $\psi_t = d_t\,\phi_t$ with $d_t = \tilde p_t(1-\tilde p_t)$ and $R_t = Y_{t+1} - (A_t+p_t-1)\gamma_t - (1-p_t)\mu_1 - p_t\mu_0$. By \eqref{A-eq:stab-i}, $d_t\,\phi_t = W_t(A_t-\tilde p_t)\,R_t\,f_t$. Distribute $\tilde\sigma_t^2$ in \eqref{A-eq:dr-wcls-original} and use $\beta(t;H_t) = \mu_1 - \mu_0$ to write $\psi_t = W_t(A_t-\tilde p_t)\{Y_{t+1} - g_t(H_t,A_t)\}f_t + \tilde p_t(1-\tilde p_t)(\mu_1 - \mu_0 - \gamma_t)f_t$. Then
\begin{align*}
d_t\phi_t - \psi_t & = W_t(A_t-\tilde p_t)f_t \big[\{- (A_t+p_t-1)\gamma_t - (1-p_t)\mu_1 - p_t\mu_0\} + g_t(H_t,A_t)\big] \\
& \quad - \tilde p_t(1-\tilde p_t)(\mu_1-\mu_0-\gamma_t) f_t.
\end{align*}
Using $g_t(H_t,A_t) = A_t\mu_1 + (1-A_t)\mu_0$, the term in square brackets simplifies to $-(A_t+p_t-1)\gamma_t + (A_t+p_t-1)(\mu_1-\mu_0) = (A_t+p_t-1)(\mu_1-\mu_0-\gamma_t)$. Therefore
\begin{align*}
d_t\phi_t - \psi_t = (\mu_1-\mu_0-\gamma_t)\,f_t\,\big\{W_t(A_t-\tilde p_t)(A_t+p_t-1) - \tilde p_t(1-\tilde p_t)\big\} = 0,
\end{align*}
where the last equality is \eqref{A-eq:stab-ii}.

\paragraph{DR-EMEE \citep{liu2024incorporating}.}
Equation~(12) of \citet{liu2024incorporating} defines the DR-EMEE-NonP estimator as the solution to $\PP_n[\sum_t m_D] = 0$ with
\begin{align}
m_D = I_t W_t\big\{U_t(\beta) - \mu_{A_t}(H_t)e^{-A_t S_t^\top\beta}\big\}\{A_t-\tilde p_t\}S_t + I_t \tilde p_t(1-\tilde p_t)\{\mu_1 e^{-S_t^\top\beta} - \mu_0\}S_t, \label{A-eq:dr-emee-original}
\end{align}
where $U_t(\beta) = e^{-A_t S_t^\top\beta}Y_{t+1}$ and $\mu_a(H_t) = \EE[Y_{t+1}\mid H_t, A_t=a, I_t=1]$. We claim that $m_D = d_t\,\phi_t$ with $d_t = \tilde p_t(1-\tilde p_t)$ and $R_t = e^{-A_t\gamma_t}Y_{t+1} - (1-p_t)e^{-\gamma_t}\mu_1 - p_t\mu_0$, the row labeled ``DR-EMEE''.

The first term of \eqref{A-eq:dr-emee-original} simplifies to $I_t W_t(A_t-\tilde p_t)e^{-A_t\gamma_t}(Y_{t+1}-\mu_{A_t})S_t$ using $U_t(\beta) - \mu_{A_t}e^{-A_t\gamma_t} = e^{-A_t\gamma_t}(Y_{t+1}-\mu_{A_t})$. By \eqref{A-eq:stab-i}, $d_t\phi_t = I_t W_t(A_t-\tilde p_t)\{e^{-A_t\gamma_t}Y_{t+1} - (1-p_t)e^{-\gamma_t}\mu_1 - p_t\mu_0\}S_t$. Therefore
\begin{align*}
d_t\phi_t - m_D & = I_t W_t(A_t-\tilde p_t)S_t\big\{- (1-p_t)e^{-\gamma_t}\mu_1 - p_t\mu_0 + e^{-A_t\gamma_t}\mu_{A_t}\big\} \\
& \quad - I_t \tilde p_t(1-\tilde p_t)(e^{-\gamma_t}\mu_1 - \mu_0)S_t.
\end{align*}
Using $e^{-A_t\gamma_t}\mu_{A_t} = A_t e^{-\gamma_t}\mu_1 + (1-A_t)\mu_0$, the brace simplifies to $(A_t+p_t-1)(e^{-\gamma_t}\mu_1 - \mu_0)$, giving
\begin{align*}
d_t\phi_t - m_D = I_t S_t(e^{-\gamma_t}\mu_1 - \mu_0)\big\{W_t(A_t-\tilde p_t)(A_t+p_t-1) - \tilde p_t(1-\tilde p_t)\big\} = 0,
\end{align*}
where the last equality is \eqref{A-eq:stab-ii}.

\paragraph{WCLS \citep{boruvka2018}.}
Setting $k = 1$ in equation~(6) of \citet{boruvka2018} (the immediate-effect setting) and including the availability indicator as in their Section~4, the $\beta$-component of the WCLS estimating equation is
\begin{align}
\sum_{t=1}^T I_t W_t (A_t - \tilde p_t)\big\{Y_{t+1} - g_t(H_t)^\top\alpha - (A_t - \tilde p_t)\gamma_t\big\} f_t(S_t) = 0. \label{A-eq:wcls-beta}
\end{align}
Applying \eqref{A-eq:stab-i} to \eqref{A-eq:wcls-beta} gives $\sum_t I_t \tilde p_t(1-\tilde p_t)\frac{A_t-p_t}{p_t(1-p_t)}\,R_t\,f_t = 0$ with $R_t = Y_{t+1} - g_t(H_t)^\top\alpha - (A_t-\tilde p_t)\gamma_t$, which equals $\sum_t d_t\,\phi_t$ with $d_t = \tilde p_t(1-\tilde p_t)$ and the $R_t$ matching the row labeled ``WCLS''.

\paragraph{EMEE \citep{qian2021estimating}.}
Setting $\Delta = 1$ in equation~(10) of \citet{qian2021estimating} (the immediate-effect setting), the weight $J_t$ reduces to $W_t$, and the $\beta$-component becomes
\begin{align*}
\sum_{t=1}^T I_t W_t (A_t-\tilde p_t)\, e^{-A_t S_t^\top\beta}\big[Y_{t+1} - \exp\{g(H_t)^\top\alpha + A_t S_t^\top\beta\}\big] S_t = 0.
\end{align*}
Since $e^{-A_t\gamma_t}\exp\{g_t^\top\alpha + A_t\gamma_t\} = e^{g_t^\top\alpha}$, the bracketed term simplifies to $e^{-A_t\gamma_t}Y_{t+1} - e^{g_t^\top\alpha}$. Applying \eqref{A-eq:stab-i} then yields $\sum_t I_t \tilde p_t(1-\tilde p_t)\frac{A_t-p_t}{p_t(1-p_t)}\,R_t\,S_t = 0$ with $R_t = e^{-A_t\gamma_t}Y_{t+1} - e^{g_t(H_t)^\top\alpha}$, which equals $\sum_t d_t\,\phi_t$ with $d_t = \tilde p_t(1-\tilde p_t)$ and the $R_t$ matching the row labeled ``EMEE''.

\paragraph{Unadjusted estimator $\hat\beta_\text{unadj}$.}
By definition (see \cref{subsec:efficiency-comparison}), $\hat\beta_\text{unadj}$ solves \[ \PP_n \sum_t I_t \frac{A_t-p_t}{p_t(1-p_t)} U_t(\beta)\,\partial_\beta\gamma_t = 0 \] with $U_t(\beta) = Y_{t+1} - A_t\gamma_t$ for the identity link and $U_t(\beta) = e^{-A_t\gamma_t}Y_{t+1}$ for the log link. This matches $\sum_t d_t\,\phi_t$ with $d_t = 1$ and $R_t = U_t(\beta)$, namely the rows labeled ``Unadjusted'' in \cref{tab:estimator-comparison}.

\section{Technical Details for the Asymptotic Theory for Globally Robust Z-Estimators}
\label{A-sec:general-lemmas}

\subsection{Notation and General Setup}
\label{A-subsec:setup}

Suppose the data observed is i.i.d. $O_i, 1 \leq i \leq n$. Let $\cO$ denote the support of $O_i$. Let $P$ denote the distribution of $O_i$. We use upper-case letters to denote random variables and lower-case letters to denote realizations. For any function $f(o)$, define $\PP(f) = \int f(o) dP(o)$, and $\PP_n (f) = \frac{1}{n} \sum_{i=1}^n f(O_i)$. Let $\| \cdot \|$ denote the $L_2$ norm on a functional space: $\| f \| = \{\int |f(o)|^2 dP(o)\}^{1/2}$. If the functional space is matrix-valued, let $\| \cdot \|$ denote $\| f \| = \{\int |f(o)|_F^2 dP(o)\}^{1/2}$, where $|\cdot|_F$ is the Frobenius norm. Let $| \cdot |$ denote the absolute value or the Euclidean norm (depending on whether the argument is a scalar, vector, or matrix). For a matrix $A$, let $A^+$ denote the Moore-Penrose pseudoinverse of $A$.

The parameter of interest is $\beta$, a $p$-dimensional vector with parameter space $\Theta \subset \mathbb{R}^p$. The estimating function for $\beta$ is $m(\beta, \eta)$, where $\eta$ is a nuisance function (we omit the notational dependence of $m(\cdot)$ on the observations $o$). We use $\cT$ to denote the set of all possible $\eta$.

\subsection{Assumptions}
\label{A-subsec:asu-lemmas}

\begin{asu}[Globally robust estimating function]
	\label{A-asu:unique-zero}
	There exists $\beta^\star \in \Theta$, such that for each $\eta \in \cT$, $\beta^\star$ is the unique solution to $\PP\{m(\beta,\eta)\} = 0$.
\end{asu}

\begin{asu}[Regularity conditions]
	\label{A-asu:reg-general}
	~
	\begin{asulist}
		\item \label{A-asu:reg-compact-param-space} Suppose the parameter space $\Theta$ of $\beta$ is compact.
		\item \label{A-asu:reg-bounded-obs} Suppose the support of $O$ is bounded.
		\item \label{A-asu:reg-cont-PPee} Suppose $\PP\{m(\beta, \eta')\}$ is a continuous function in $\beta$.
		\item \label{A-asu:reg-bounded-and-cont-differentiable-ee} Suppose $m(\beta,\eta)$ is continuously differentiable in $\beta$, and the class $\{m(\beta,\eta): \beta \in \Theta, \eta \in \cT\}$ is uniformly bounded and bounded by an integrable function.
		\item \label{A-asu:reg-dominated-ee-deriv} Suppose $\partial_\beta m(\beta,\eta) := \frac{\partial m(\beta,\eta)}{\partial\beta^T}$ is bounded by an integrable function.
		\item \label{A-asu:reg-dominated-ee-meat} Suppose $m(\beta,\eta) m(\beta,\eta)^T$ is bounded by an integrable function.
		\item \label{A-asu:reg-invertible-ee-deriv} Suppose $\PP\{\partial_\beta m(\beta^\star,\eta')\}$ is invertible.
	\end{asulist}
\end{asu}

The following are additional assumptions needed for establishing the asymptotics for the non-cross-fitted estimator $\hat\beta$.

\begin{asu}[Convergence of nuisance parameter estimator]
	\label{A-asu:nuisance-conv-general}
	Suppose there exists $\eta' \in \cT$ such that the following hold.
	\begin{asulist}
		\item \label{A-asu:nuisance-conv-PPee-sup} $\sup_{\beta \in \Theta} | \PP m(\beta, \hat\eta) - \PP m(\beta, \eta') | = o_P(1)$;
		\item \label{A-asu:nuisance-conv-ee-l2} $\| m(\beta^\star, \hat\eta) - m(\beta^\star, \eta') \|^2  := \int |m(\beta^\star, \hat\eta) - m(\beta^\star, \eta')|^2 dP = o_P(1)$;
		\item \label{A-asu:nuisance-conv-PPee-deriv} $| \PP \{\partial_\beta m(\beta^\star, \hat\eta)\} - \PP \{\partial_\beta m(\beta^\star, \eta')\} | = o_P(1)$;
		\item \label{A-asu:nuisance-conv-PPee-meat} $| \PP \{m(\beta^\star, \hat\eta) m(\beta^\star, \hat\eta)^T\} - \PP \{m(\beta^\star, \eta') m(\beta^\star, \eta')^T \} | = o_P(1)$.
	\end{asulist}
\end{asu}

\begin{asu}[Donsker condition]
	\label{A-asu:donsker}
	Suppose $\cM := \{m(\beta, \eta): \beta \in \Theta,\allowbreak \eta \in \cT\}$ and $\{\partial_\beta m(\beta,\eta): \beta \in \Theta,\allowbreak \eta \in \cT\}$ are $P$-Donsker classes.
\end{asu}

The following are additional assumptions needed for establishing the asymptotics for the cross-fitted estimator $\check\beta$.

\begin{asu}[Convergence of nuisance parameter estimator (cross-fitting)]
	\label{A-asu:nuisance-conv-general-cf}
	Suppose there exists $\eta' \in \cT$ such that the following hold.
	\begin{asulist}
		\item \label{A-asu:nuisance-conv-PPee-sup-cf} For each $k\in[K]$, $\sup_{\beta \in \Theta} | \PP m(\beta, \hat\eta_k) - \PP m(\beta, \eta') | = o_P(1)$;
		\item \label{A-asu:nuisance-conv-ee-l2-cf} For each $k\in[K]$, $ \| m(\beta^\star, \hat\eta_k)  - m(\beta^\star, \eta') \|^2  = o_P(1)$.
		\item \label{A-asu:nuisance-conv-PPee-deriv-cf} For each $k\in[K]$, $| \PP \{\partial_\beta m(\beta^\star, \hat\eta_k)\} - \PP \{\partial_\beta m(\beta^\star, \eta')\} | = o_P(1)$;
		\item \label{A-asu:nuisance-conv-PPee-meat-cf} For each $k\in[K]$, $| \PP \{m(\beta^\star, \hat\eta_k) m(\beta^\star, \hat\eta_k)^T\} - \PP \{m(\beta^\star, \eta') m(\beta^\star, \eta')^T \} | = o_P(1)$.
	\end{asulist}
\end{asu}

\begin{asu}[Additional regularity conditions (cross-fitting)]
	\label{A-asu:reg-bounded-ee-deriv-and-meat}
	Suppose $\partial_\beta m(\beta,\eta) := \frac{\partial m(\beta,\eta)}{\partial\beta^T}$ and $m(\beta,\eta) m(\beta,\eta)^T$ are uniformly bounded.
\end{asu}

\subsection{Lemmas}

\begin{lem}[Well-separated zero]
	\label{A-lem:well-separated-zero}
	Suppose \cref{A-asu:unique-zero,A-asu:reg-compact-param-space,A-asu:reg-cont-PPee} hold, then the unique zero $\beta^\star$ of the function $\PP\{m(\beta, \eta')\}$ is well-separated. That is, for any $\epsilon > 0$, there exists $\delta > 0$ such that $|\beta - \beta^\star| > \epsilon$ implies $|\PP\{m(\beta, \eta')\}| > \delta$.
\end{lem}

\begin{proof}[Proof of \cref{A-lem:well-separated-zero}]
	For any $\epsilon > 0$, consider $\delta := \inf_{\beta \in \Theta: d(\beta, \beta^\star) \geq \epsilon}|\PP\{m(\beta, \eta')\}|$. Because the parameter space $\Theta$ is compact (\cref{A-asu:reg-compact-param-space}), $\{\beta: d(\beta, \beta^\star) \geq \epsilon\} \cap \Theta$ is also compact. This combined with the fact that $\PP\{m(\beta, \eta')\}$ is a continuous function in $\beta$ (\cref{A-asu:reg-cont-PPee}) implies that the infimum is attained, i.e., $\inf_{\beta \in \Theta: d(\beta, \beta^\star) \geq \epsilon}|\PP\{m(\beta, \eta')\}| = \min_{\beta \in \Theta: }|\PP\{m(\beta, \eta')\}|$. Because $\beta^\star$ is the unique zero of $\PP\{m(\beta, \eta')\}$ (\cref{A-asu:unique-zero}), $\min_{\beta \in \Theta: d(\beta, \beta^\star) \geq \epsilon}|\PP\{m(\beta, \eta')\}| > 0$. Therefore, we proved the lemma by constructing a particular $\delta > 0$.
\end{proof}

\begin{lem}
	\label{A-lem:sum-of-squares-function}
	Let $|\cdot|$ denote the Euclidean norm. For any $p$-dimensional vector-valued functions $\{a_t(x)\}_{t=1}^T$, $\{b_t(x)\}_{t=1}^T$, and $\{c_t(x)\}_{t=1}^T$, we have
	\begin{align*}
		\int\bigg| \sum_{t=1}^T \{a_t(x) + b_t(x) + c_t(x)\} \bigg|^2 dP(x) \leq 4 T^2  & \bigg\{\max_{1 \leq t \leq T} \int |a_t(x)|^2 dP(x) \\
& + \max_{1 \leq t \leq T} \int |b_t(x)|^2 dP(x) \\
        & + \max_{1 \leq t \leq T} \int |c_t(x)|^2 dP(x) \bigg\}.
	\end{align*}
\end{lem}

\begin{proof}[Proof of \cref{A-lem:sum-of-squares-function}]
	Using the fact that $2z_1^T z_2 \leq |z_1|^2 + |z_2|^2$ and $|z_1+z_2|^2 \leq 2|z_1|^2 + 2|z_2|^2$, we have
	\begin{align*}
		& ~~~ \int\bigg| \sum_{t=1}^T \{a_t(x) + b_t(x) + c_t(x)\} \bigg|^2 dP(x) \\
		& = \int \sum_{1 \leq t, s \leq T}\{a_t(x) + b_t(x) + c_t(x)\}^T \{a_s(x) + b_s(x) + c_s(x)\} dP(x) \\
		& \leq \int \sum_{1 \leq t, s \leq T} \frac{1}{2}\left\{|a_t(x) + b_t(x) + c_t(x)|^2 + |a_s(x) + b_s(x) + c_s(x)|^2\right\} dP(x) \\
		& \leq \int \sum_{1 \leq t, s \leq T} \frac{1}{2}\left\{(2 |a_t(x) + b_t(x)|^2 + 2 |c_t(x)|^2) + (2 |a_s(x) + b_s(x)|^2 + 2 |c_s(x)|^2)\right\} dP(x) \\
            & \leq \int \sum_{1 \leq t, s \leq T} \frac{1}{2}\big\{(4 |a_t(x)|^2 + 4 |b_t(x)|^2 + 4 |c_t(x)|^2) \\
& \qquad + (4 |a_s(x)|^2 + 4 |b_s(x)|^2 + 4 |c_s(x)|^2)\big\} dP(x) \\
		& = \sum_{1 \leq t, s \leq T} 2 \bigg\{\int |a_t(x)|^2 dP(x) + \int |b_t(x)|^2 dP(x) + \int |c_t(x)|^2 dP(x) \\
            & ~~~~~~ + \int |a_s(x)|^2 dP(x) + \int |b_s(x)|^2 dP(x) + \int |c_s(x)|^2 dP(x)\bigg\} \\
		& \leq \sum_{1 \leq t, s \leq T} 4  \bigg\{ \max_{1 \leq t \leq T} \int |a_t(x)|^2 dP(x) + \max_{1 \leq t \leq T} \int |b_t(x)|^2 dP(x) + \max_{1 \leq t \leq T} \int |c_t(x)|^2 dP(x)\bigg\}
        \end{align*}
        \begin{align*}
        & = 4 T^2 \bigg\{\max_{1 \leq t \leq T} \int |a_t(x)|^2 dP(x) + \max_{1 \leq t \leq T} \int |b_t(x)|^2 dP(x) + \max_{1 \leq t \leq T} \int |c_t(x)|^2 dP(x) \bigg\}.
	\end{align*}
	This completes the proof.
\end{proof}

\begin{lem}[Lemma 19.24 of \citet{van2000asymptotic}]
	\label{A-lem:vdv19.24}
	Suppose that $\mathcal{F}$ is a $P$-Donsker class of measurable functions and $\hat{f}_n$ is a sequence of random functions that take their values in $\mathcal{F}$ such that $\int\{ \hat{f}_n(x) - f_0(x) \}^2 dP(x)$ converges in probability to 0 for some $f_0 \in L_2(P)$. Then $\mathbb{G}_n(\hat{f}_n - f_0) \pto 0$ and hence $\mathbb{G} \hat{f}_n \rightsquigarrow \mathbb{G}_P(f_0)$.
\end{lem}

\subsection{Asymptotic Normality for Non-cross-fitted Estimator}
\label{A-subsec:normality-general}

\begin{lem}[Consistency]
	\label{A-lem:general-consistency}
	Suppose \cref{A-asu:unique-zero,A-asu:nuisance-conv-PPee-sup,A-asu:donsker,A-asu:reg-compact-param-space,A-asu:reg-bounded-obs,A-asu:reg-cont-PPee} hold, then $\hat\beta \pto \beta^\star$ as $n\to\infty$.
\end{lem}

\begin{proof}[Proof of \cref{A-lem:general-consistency}]
	Consider an arbitrary $\epsilon > 0$. We need to prove $\lim_{n\to\infty} P(| \hat\beta - \beta^\star | > \epsilon) = 0$. Because \cref{A-asu:unique-zero,A-asu:reg-compact-param-space,A-asu:reg-cont-PPee} hold, by \cref{A-lem:well-separated-zero} there exists $\delta > 0$ such that
	\begin{align*}
		P(| \hat\beta - \beta^\star | > \epsilon) \leq P[ |\PP\{m(\hat\beta, \eta')\}| > \delta].
	\end{align*}
	Therefore, it suffices to prove that $|\PP\{m(\hat\beta, \eta')\}|$ converges in probability to 0.

	We have
	\begin{align}
		|\PP\{m(\hat\beta, \eta')\}|
		& = | \PP\{m(\hat\beta, \eta')\} - \PP_n\{m(\hat\beta, \hat\eta)\} | \nonumber \\
		& \leq | \PP\{m(\hat\beta, \eta')\} - \PP\{m(\hat\beta, \hat\eta)\} | + | \PP\{m(\hat\beta, \hat\eta)\} - \PP_n\{m(\hat\beta, \hat\eta)\} | \label{A-eq:lem:general-consistency:proofuse1}
	\end{align}
	Next we show that both terms in \cref{A-eq:lem:general-consistency:proofuse1} are $o_P(1)$.

	For the first term in \cref{A-eq:lem:general-consistency:proofuse1}, by \cref{A-asu:nuisance-conv-PPee-sup} we have
	\begin{align}
		| \PP\{m(\hat\beta, \eta')\} - \PP\{m(\hat\beta, \hat\eta)\} | \leq \sup_{\beta \in \Theta} | \PP\{m(\beta, \eta')\} - \PP\{m(\beta, \hat\eta)\} | = o_P(1), \label{A-eq:lem:general-consistency:proofuse2}
	\end{align}
	
	For the second term in \cref{A-eq:lem:general-consistency:proofuse1}, because $\cM$ is a $P$-Donsker class (\cref{A-asu:donsker}) and thus a $P$-Glivenko-Cantelli class, $\sup_{\beta \in \Theta, \eta \in \cT}| (\PP_n - \PP) m (\beta, \eta)| = o_P(1)$. Therefore, $| \PP\{m(\hat\beta, \hat\eta)\} - \PP_n\{m(\hat\beta, \hat\eta)\} | = | (\PP_n - \PP) m (\hat\beta, \hat\eta)| = o_P(1)$.

	Thus, we showed that both terms in \cref{A-eq:lem:general-consistency:proofuse1} are $o_P(1)$. This completes the proof.	
\end{proof}

\begin{lem}[Convergence of the derivative]
	\label{A-lem:conv-ee-deriv}
	Suppose \cref{A-asu:nuisance-conv-PPee-deriv,A-asu:reg-dominated-ee-deriv,A-asu:donsker} hold. If a sequence of random variables $\tilde\beta_n$ satisfies $\tilde\beta_n \pto \beta^\star$, then $\PP_n \{ \partial_\beta m(\tilde\beta_n, \hat\eta) \} \pto \PP \left\{ \partial_\beta m(\beta^\star, \eta') \right\}$.
\end{lem}

\begin{proof}[Proof of \cref{A-lem:conv-ee-deriv}]
	We have
	\begin{align}
		& ~~~~ \PP_n \{ \partial_\beta m(\tilde\beta_n, \hat\eta) \} - \PP \big\{ \partial_\beta m(\beta^\star, \eta') \big\} \nonumber \\
		& =  \PP_n \{ \partial_\beta m(\tilde\beta_n, \hat\eta) \} - \PP \big\{ \partial_\beta m(\tilde\beta_n, \hat\eta) \big\} + \PP \big\{ \partial_\beta m(\tilde\beta_n, \hat\eta) \big\} -  \PP \big\{ \partial_\beta m(\beta^\star, \eta') \big\} \nonumber \\
		& \leq \sup_{\beta \in \Theta, \eta \in \cT} | (\PP_n - \PP) \partial_\beta m(\beta, \eta) | \nonumber \\
		& ~~~~ + \bigg[\PP \big\{ \partial_\beta m(\tilde\beta_n, \hat\eta) \big\} -  \PP \big\{ \partial_\beta m(\beta^\star, \hat\eta) \big\}\bigg] + \bigg[\PP \big\{ \partial_\beta m(\beta^\star, \hat\eta) \big\} -  \PP \big\{ \partial_\beta m(\beta^\star, \eta') \big\}\bigg].
		\label{A-eq:lem:conv-ee-deriv:proofuse1}
	\end{align}
	We now control each of the three terms in \cref{A-eq:lem:conv-ee-deriv:proofuse1}.

	For the first term in \cref{A-eq:lem:conv-ee-deriv:proofuse1}, \cref{A-asu:donsker} implies that $\{\partial_\beta m(\beta,\eta): \beta \in \Theta, \eta \in \cT\}$ is a $P$-Glivenko-Cantelli class and thus
	\begin{align}
		\sup_{\beta \in \Theta, \eta \in \cT} | (\PP_n - \PP) \partial_\beta m(\beta, \eta) | = o_P(1). \label{A-eq:lem:conv-ee-deriv:proofuse2}
	\end{align}	

	For the second term in \cref{A-eq:lem:conv-ee-deriv:proofuse1}, it follows from the fact that $\tilde\beta_n \pto \beta^\star$ (lemma assumption), the dominatedness of $\partial_\beta m(\beta,\eta)$ (\cref{A-asu:reg-dominated-ee-deriv}), and the dominated convergence theorem for convergence in probability (see, e.g., \citet[][Result (viii) of Chapter 3.2]{chung2001course}) that
	\begin{align}
		\PP \big\{ \partial_\beta m(\tilde\beta_n, \hat\eta) \big\} - \PP \big\{ \partial_\beta m(\beta^\star, \hat\eta) \big\} = o_P(1). \label{A-eq:lem:conv-ee-deriv:proofuse3}
	\end{align}

	For the third term in \cref{A-eq:lem:conv-ee-deriv:proofuse1}, \cref{A-asu:nuisance-conv-PPee-deriv} states that
	\begin{align}
		\PP \big\{ \partial_\beta m(\beta^\star, \hat\eta) \big\} - \PP \big\{ \partial_\beta m(\beta^\star, \eta') \big\} = o_P(1). \label{A-eq:lem:conv-ee-deriv:proofuse4}
	\end{align}

	Plugging \cref{A-eq:lem:conv-ee-deriv:proofuse4,A-eq:lem:conv-ee-deriv:proofuse3,A-eq:lem:conv-ee-deriv:proofuse2} into \cref{A-eq:lem:conv-ee-deriv:proofuse1} yields
	\begin{align*}
		\PP_n \big\{ \partial_\beta m(\tilde\beta_n, \hat\eta) \big\} - \PP \big\{ \partial_\beta m(\beta^\star, \eta') \big\} = o_P(1).
	\end{align*}
	The proof is thus completed.
\end{proof}

\begin{lem}[Convergence of the ``meat'' term]
	\label{A-lem:conv-ee-meat}
	Suppose \cref{A-asu:nuisance-conv-PPee-meat,A-asu:donsker,A-asu:reg-bounded-and-cont-differentiable-ee,A-asu:reg-dominated-ee-meat} hold. If a sequence of random variables $\tilde\beta_n$ satisfies $\tilde\beta_n \pto \beta^\star$, then $\PP_n \{ m(\tilde\beta_n, \hat\eta) m(\tilde\beta_n, \hat\eta)^T \} \pto \PP \{ m(\beta^\star, \eta') m(\beta^\star, \eta')^T \}$.
\end{lem}

\begin{proof}[Proof of \cref{A-lem:conv-ee-meat}]
	The proof is identical to the proof of \cref{A-lem:conv-ee-deriv}, except that we need to establish the following:
	\begin{itemize}
	 	\item $m(\beta, \eta) m(\beta, \eta)^T$ is continuous: this is \cref{A-asu:reg-bounded-and-cont-differentiable-ee};
	 	\item $m(\beta, \eta) m(\beta, \eta)^T$ is bounded by an integrable function: this is \cref{A-asu:reg-dominated-ee-meat};
	 	\item $\PP \{m(\beta^\star, \hat\eta) m(\beta^\star, \hat\eta)^T\} - \PP \{m(\beta^\star, \eta') m(\beta^\star, \eta')^T\} = o_P(1)$: this is \cref{A-asu:nuisance-conv-PPee-meat};
	 	\item $m(\beta, \eta) m(\beta, \eta)^T$ takes value in a $P$-Glivenko-Cantelli class.
	\end{itemize}
	We establish the last bullet point now. Because $m(\beta,\eta)$ is uniformly bounded (\cref{A-asu:reg-bounded-and-cont-differentiable-ee}), $\cM$ is a $P$-Donsker class (\cref{A-asu:donsker}), and the product $fg$ is a Lipschitz transformation, these imply that $\cM\times\cM$ is a $P$-Donsker class (see, e.g., \citet[][Example 19.20]{van2000asymptotic}). Thus, $\{m(\beta, \eta) m(\beta, \eta)^T: \beta\in\Theta, \eta \in \cT\} \subset \cM\times\cM$ is a $P$-Donsker class and thus a $P$-Glivenko-Cantelli class. This completes the proof.
\end{proof}

\begin{thm}[Asymptotic Normality]
	\label{A-thm:general-normality}
	Suppose \cref{A-asu:unique-zero,A-asu:nuisance-conv-general,A-asu:donsker,A-asu:reg-general} hold. Then
	\begin{align*}
		\sqrt{n}(\hat\beta - \beta^\star) \dto N ( 0, V) \quad \text{as } n\to\infty,
	\end{align*}
	where the asymptotic variance is
	\begin{align}
		V = \PP\{\partial_\beta m(\beta^\star, \eta')\}^{-1} ~ \PP\{ m(\beta^\star, \eta') m(\beta^\star, \eta')^T \} ~ \PP\{\partial_\beta m(\beta^\star, \eta')\}^{-1, T}. \label{A-eq:def-asymp-var}
	\end{align}	
	In addition,
	\begin{align*}
		\PP_n\{\partial_\beta m(\hat\beta, \hat\eta)\}^{-1} ~ \PP_n\{ m(\hat\beta, \hat\eta) m(\hat\beta, \hat\eta)^T \} ~ \PP_n\{\partial_\beta m(\hat\beta, \hat\eta)\}^{-1, T}
	\end{align*}
	is a consistent estimator for $V$.	
\end{thm}

\begin{proof}[Proof of \cref{A-thm:general-normality}]
	Because $m(\beta,\eta)$ is continuously differentiable in $\beta$ (\cref{A-asu:reg-bounded-and-cont-differentiable-ee}), the Lagrange mean value theorem implies that
	\begin{align}
		0 = \PP_n\{ m(\hat\beta, \hat\eta)\} & = \PP_n\{ m(\beta^\star, \hat\eta) \} + \left[ \frac{\partial}{\partial \beta^T} \PP_n\{m(\tilde\beta, \hat\eta)\} \right](\hat\beta - \beta^\star) \nonumber \\
		& = \PP_n\{ m(\beta^\star, \hat\eta) \} + \PP_n\{\partial_\beta m(\tilde\beta, \hat\eta)\} (\hat\beta - \beta^\star), \label{A-eq:thm:general-normality:proofuse1}
	\end{align}
	where $\tilde\beta$ is between $\hat\beta$ and $\beta^\star$.
	Because \cref{A-lem:general-consistency} implies that $\hat\beta \pto \beta^\star$ and thus $\tilde\beta \pto \beta^\star$, \cref{A-lem:conv-ee-deriv} implies that
	\begin{align}
		\PP_n\{ \partial_\beta m(\tilde\beta, \hat\eta)\} \pto \PP\{\partial_\beta m(\beta^\star, \eta')\}. \label{A-eq:thm:general-normality:proofuse2}
	\end{align}
	Because $\PP\{\partial_\beta m(\beta^\star,\eta')\}$ is invertible (\cref{A-asu:reg-invertible-ee-deriv}), \cref{A-eq:thm:general-normality:proofuse2} implies that $\PP_n\{ \partial_\beta m(\tilde\beta, \hat\eta)\}$ is invertible with probability approaching 1. Therefore, \cref{A-eq:thm:general-normality:proofuse1,A-eq:thm:general-normality:proofuse2} imply that with probability approaching 1 we have
	\begin{align}
		\sqrt{n} (\hat\beta - \beta^\star) = - [ \PP\{\partial_\beta m(\beta^\star, \eta')\} ]^{-1} [ \sqrt{n} \PP_n\{ m(\beta^\star, \hat\eta)\} ]. \label{A-eq:thm:general-normality:proofuse3}
	\end{align}
	
	Consider the second term in \cref{A-eq:thm:general-normality:proofuse3}. We have
	\begin{align}
		\PP_n \{ m(\beta^\star, \hat\eta) \} & = (\PP_n - \PP)\{ m(\beta^\star, \hat\eta) \} + \PP\{ m(\beta^\star, \hat\eta) \} \nonumber \\
		& = (\PP_n - \PP)\{ m(\beta^\star, \hat\eta) \}, \label{A-eq:thm:general-normality:proofuse4}
	\end{align}
	where the last equality follows from the fact that $\PP\{ m(\beta^\star, \hat\eta) \} = 0$ (\cref{A-asu:unique-zero}). Because of \cref{A-asu:nuisance-conv-ee-l2}, we can invoke \cref{A-lem:vdv19.24} to get
	\begin{align*}
		\mathbb{G}_n\{m(\beta^\star, \hat\eta) - m(\beta^\star, \eta')\} \pto 0,
	\end{align*}
	or equivalently
	\begin{align}
		(\PP_n - \PP)\{ m(\beta^\star, \hat\eta) \} = (\PP_n - \PP)\{ m(\beta^\star, \eta') \} + o_P(n^{-1/2}). \label{A-eq:thm:general-normality:proofuse5}
	\end{align}
	Plugging \cref{A-eq:thm:general-normality:proofuse5} into \cref{A-eq:thm:general-normality:proofuse4} and we get
	\begin{align}
		\sqrt{n} \PP_n \{ m(\beta^\star, \hat\eta) \} & = \sqrt{n}(\PP_n - \PP)\{ m(\beta^\star, \eta') \} + o_P(1) \nonumber \\
		& \dto N \Big(0, ~\PP\{ m(\beta^\star, \eta') m(\beta^\star, \eta')^T \} \Big), \label{A-eq:thm:general-normality:proofuse6}
	\end{align}
	where the last step follows from the Lindeberg-Feller Central Limit Theorem and the fact that $\PP\{ m(\beta^\star, \eta')\} = 0$ (\cref{A-asu:unique-zero}).

	Therefore, plugging \cref{A-eq:thm:general-normality:proofuse6} into \cref{A-eq:thm:general-normality:proofuse3}, it follows from Slutsky's Theorem and the continuous mapping theorem that $\sqrt{n} (\hat\beta - \beta^\star) \dto N(0, V)$ with $V$ defined in \cref{A-eq:def-asymp-var}. Consistency of the variance estimator follows immediately from \cref{A-lem:conv-ee-deriv,A-lem:conv-ee-meat} and the continuous mapping theorem. This completes the proof.
\end{proof}

\subsection{Asymptotic Normality for Cross-fitted Estimator}
\label{A-subsec:normality-general-cf}

\begin{lem}[Consistency (cross-fitting)]
	\label{A-lem:general-consistency-cf}
	Suppose \cref{A-asu:unique-zero,A-asu:nuisance-conv-PPee-sup-cf,A-asu:reg-compact-param-space,A-asu:reg-cont-PPee,A-asu:reg-bounded-and-cont-differentiable-ee} hold, then $\check\beta \pto \beta^\star$ as $n\to\infty$.
\end{lem}

\begin{proof}[Proof of \cref{A-lem:general-consistency-cf}]
	Consider an arbitrary $\epsilon > 0$. We need to prove $\lim_{n\to\infty} P(| \check\beta - \beta^\star | > \epsilon) = 0$. Because \cref{A-asu:unique-zero,A-asu:reg-compact-param-space,A-asu:reg-cont-PPee} hold, by \cref{A-lem:well-separated-zero} there exists $\delta > 0$ such that
	\begin{align*}
		P(| \check\beta - \beta^\star | > \epsilon) \leq P[ |\PP\{m(\check\beta, \eta')\}| > \delta].
	\end{align*}
	Therefore, it suffices to prove that $|\PP\{m(\check\beta, \eta')\}|$ converges in probability to 0.

	We have
	\begin{align}
		|\PP\{m(\check\beta, \eta')\}|
		& = | \PP\{m(\check\beta, \eta')\} - \frac{1}{K} \sum_{k=1}^K \PP_{n,k}\{m(\check\beta, \hat\eta_k)\} | \nonumber \\
		& \leq \frac{1}{K} \sum_{k=1}^K | \PP\{m(\check\beta, \eta')\} - \PP_{n,k}\{m(\check\beta, \hat\eta_k)\} |.
		\label{A-eq:lem:general-consistency-cf:proofuse1}
	\end{align}
	Because $K$ is fixed, to show $|\PP\{m(\check\beta, \eta')\}| = o_P(1)$ it suffices to show that for each $k \in [K]$, $| \PP\{m(\check\beta, \eta')\} - \PP_{n,k}\{m(\check\beta, \hat\eta_k)\} | = o_P(1)$. 

	We have
	\begin{align}
		& | \PP\{m(\check\beta, \eta')\} - \PP_{n,k}\{m(\check\beta, \hat\eta_k)\} | \nonumber \\
		\leq & | \PP\{m(\check\beta, \eta')\} - \PP\{m(\check\beta, \hat\eta_k)\} | + | \PP\{m(\check\beta, \hat\eta_k)\} - \PP_{n,k}\{m(\check\beta, \hat\eta_k)\} | . \label{A-eq:lem:general-consistency-cf:proofuse2}
	\end{align}
	Next we show that both terms in \cref{A-eq:lem:general-consistency-cf:proofuse2} are $o_P(1)$.

	For the first term in \cref{A-eq:lem:general-consistency-cf:proofuse2}, it being $o_P(1)$ follows immediately from \cref{A-asu:nuisance-conv-PPee-sup-cf}.

	For the second term in \cref{A-eq:lem:general-consistency-cf:proofuse2}, we have
	\begin{align}
		| \PP\{m(\check\beta, \hat\eta_k)\} - \PP_{n,k}\{m(\check\beta, \hat\eta_k)\} | \leq \sup_{\beta \in \Theta} | \PP\{m(\beta, \hat\eta_k)\} - \PP_{n,k}\{m(\beta, \hat\eta_k)\} |.
	\end{align}
	We first show that for fixed $\beta \in \Theta$, $| \PP\{m(\beta, \hat\eta_k)\} - \PP_{n,k}\{m(\beta, \hat\eta_k)\} | = o_P(1)$, using a simple analysis of the variance. Because $\hat\eta_k$ is fixed when conditioning on $(O_i)_{i \in B_k^c}$, we have
	\begin{align}
		\EE [ \PP\{m(\beta, \hat\eta_k)\} - \PP_{n,k}\{m(\beta, \hat\eta_k)\} \mid (O_i)_{i \in B_k^c} ] = 0. \label{A-eq:lem:general-consistency-cf:proofuse3}
	\end{align}
	We also have
	\begin{align}
		& \quad \EE [ | \PP\{m(\beta, \hat\eta_k)\} - \PP_{n,k}\{m(\beta, \hat\eta_k)\} |^2 \mid (O_i)_{i \in B_k^c} ] \nonumber \\
		& = |\PP\{m(\beta, \hat\eta_k)\}|^2 - 2 |\PP\{m(\beta, \hat\eta_k)\}|^2 + \EE [ |\PP_{n,k}\{m(\beta, \hat\eta_k)\} |^2 \mid (O_i)_{i \in B_k^c} ] \nonumber \\
		& = - |\PP\{m(\beta, \hat\eta_k)\}|^2 + \EE \bigg\{ \frac{1}{N^2} \sum_{j_1,j_2 \in B_k} m(O_{j_1}; \beta, \hat\eta_k) m(O_{j_2}; \beta, \hat\eta_k) \Big\vert (O_i)_{i \in B_k^c} \bigg\} \nonumber \\
		& = - \frac{1}{N} |\PP\{m(\beta, \hat\eta_k)\}|^2 + \frac{1}{N} \PP\{ |m(\beta, \hat\eta_k)|^2 \} \nonumber \\
		& \leq \frac{1}{N} \PP\{ |m(\beta, \hat\eta_k)|^2 \} \nonumber \\
		& = O_P(n^{-1}), \label{A-eq:lem:general-consistency-cf:proofuse4}
	\end{align}
	where the last equality follows from \cref{A-asu:reg-bounded-and-cont-differentiable-ee}. \cref{A-eq:lem:general-consistency-cf:proofuse3,A-eq:lem:general-consistency-cf:proofuse4} and Chebyshev's inequality imply that conditional on $(O_i)_{i \in B_k^c}$, $| \PP\{m(\beta, \hat\eta_k)\} - \PP_{n,k}\{m(\beta, \hat\eta_k)\} | = O_P(n^{-1/2})$. Lemma 6.1 in \citet{chernozhukov2018double} then implies that unconditionally,
	\begin{align}
		| \PP\{m(\beta, \hat\eta_k)\} - \PP_{n,k}\{m(\beta, \hat\eta_k)\} | = O_P(n^{-1/2}). \label{A-eq:lem:general-consistency-cf:proofuse5}
	\end{align}
	By \cref{A-asu:reg-compact-param-space,A-asu:reg-bounded-and-cont-differentiable-ee} and Example 19.8 in \citet{van2000asymptotic}, for a fixed $\hat{\eta}_k$, $\{m(\beta, \hat{\eta}_k) : \beta \in \Theta\}$ is $P$-Glivenko-Cantelli. Because $\hat\eta_k$ is independent of the observations being averaged over in $\PP_{n,k}$, \cref{A-eq:lem:general-consistency-cf:proofuse5} and the property of $P$-Glivenko-Cantelli class imply that
	\begin{align*}
		\sup_{\beta \in \Theta} | \PP\{m(\beta, \hat\eta_k)\} - \PP_{n,k}\{m(\beta, \hat\eta_k)\} | = O_P(n^{-1/2}),
	\end{align*}
	which further implies that
	\begin{align*}
		| \PP\{m(\check\beta, \hat\eta_k)\} - \PP_{n,k}\{m(\check\beta, \hat\eta_k)\} | = O_P(n^{-1/2}).
	\end{align*}

	Therefore, we showed that both terms in \cref{A-eq:lem:general-consistency-cf:proofuse2} are $o_P(1)$. This completes the proof.
\end{proof}

\begin{lem}[Convergence of the derivative (cross-fitting)]
	\label{A-lem:conv-ee-deriv-cf}
	Suppose \cref{A-asu:reg-bounded-ee-deriv-and-meat,A-asu:reg-dominated-ee-deriv,A-asu:nuisance-conv-PPee-deriv-cf,A-asu:reg-compact-param-space,A-asu:reg-bounded-and-cont-differentiable-ee} hold. If a sequence of random variables $\tilde\beta_n$ satisfies $\tilde\beta_n \pto \beta^\star$, then
	\begin{align*}
		\frac{1}{K}\sum_{k=1}^K \PP_{n,k} \{ \partial_\beta m(\tilde\beta_n, \hat\eta_k) \} \pto \PP \{ \partial_\beta m(\beta^\star, \eta') \}.
	\end{align*}
\end{lem}

\begin{proof}[Proof of \cref{A-lem:conv-ee-deriv-cf}]
	Because $K$ is finite, it suffices to show that
	\begin{align*}
		\PP_{n,k} \{ \partial_\beta m(\tilde\beta_n, \hat\eta_k) \} - \PP \{ \partial_\beta m(\beta^\star, \eta') \} = o_P(1).
	\end{align*}
	We have
	\begin{align}
		& \quad |\PP_{n,k} \{ \partial_\beta m(\tilde\beta_n, \hat\eta_k) \} - \PP \{ \partial_\beta m(\beta^\star, \eta') \}| \nonumber \\
		& \leq |\PP_{n,k} \{ \partial_\beta m(\tilde\beta_n, \hat\eta_k) \} - \PP \{\partial_\beta m(\tilde\beta, \hat\eta_k) \} |
		+ | \PP \{\partial_\beta m(\tilde\beta, \hat\eta_k) \} - \PP \{ \partial_\beta m(\beta^\star, \hat\eta_k) \} | \nonumber \\
		& \quad + | \PP \{ \partial_\beta m(\beta^\star, \hat\eta_k) \} - \PP \{ \partial_\beta m(\beta^\star, \eta') \} |. \label{A-eq:lem:conv-ee-deriv-cf:proofuse1}
	\end{align}
	Next we show that all three terms in \cref{A-eq:lem:conv-ee-deriv-cf:proofuse1} are $o_P(1)$.
	
	For the first term in \cref{A-eq:lem:conv-ee-deriv-cf:proofuse1}, using the same moment bounding technique as in \cref{A-eq:lem:general-consistency-cf:proofuse3,A-eq:lem:general-consistency-cf:proofuse4} in the proof of \cref{A-lem:general-consistency-cf}, we have that for fixed $\beta \in \Theta$,
	\begin{align*}
		\EE [ |\PP_{n,k} \{ \partial_\beta m(\beta, \hat\eta_k) \} - \PP \{\partial_\beta m(\beta, \hat\eta_k) \} |^2 \mid (O_i)_{i \in B_k^c} ] \leq \frac{1}{N} \PP \{|\partial_\beta m(\beta, \hat\eta_k)|^2\} = O_P(n^{-1}),
	\end{align*}
	where the last equality follows from \cref{A-asu:reg-bounded-ee-deriv-and-meat}. Therefore, Chebyshev's inequality and Lemma 6.1 in \citet{chernozhukov2018double} imply that
	\begin{align}
		\PP_{n,k} \{ \partial_\beta m(\beta, \hat\eta_k) \} - \PP \{\partial_\beta m(\beta, \hat\eta_k) \} = O_P(n^{-1/2}). \label{A-eq:lem:conv-ee-deriv-cf:proofuse2}
	\end{align}
	\cref{A-asu:reg-bounded-and-cont-differentiable-ee} implies that $\partial_\beta m(\beta, \eta)$ is continuous. By \cref{A-asu:reg-compact-param-space,A-asu:reg-dominated-ee-deriv} and Example 19.8 in \citet{van2000asymptotic}, for a fixed $\hat{\eta}_k$, $\{\partial_\beta m(\beta, \hat{\eta}_k) : \beta \in \Theta\}$ is $P$-Glivenko-Cantelli. Because $\hat\eta_k$ is independent of the observations being averaged over in $\PP_{n,k}$, \cref{A-eq:lem:conv-ee-deriv-cf:proofuse2} and the property of $P$-Glivenko-Cantelli class imply that
	\begin{align*}
		\sup_{\beta \in \Theta} | \PP_{n,k} \{ \partial_\beta m(\beta, \hat\eta_k) \} - \PP \{\partial_\beta m(\beta, \hat\eta_k) \} | = O_P(n^{-1/2}),
	\end{align*}
	which further implies that
	\begin{align*}
		| \PP_{n,k} \{ \partial_\beta m(\tilde\beta, \hat\eta_k) \} - \PP \{\partial_\beta m(\tilde\beta, \hat\eta_k) \} | = O_P(n^{-1/2}).
	\end{align*}
	This shows that the first term in \cref{A-eq:lem:conv-ee-deriv-cf:proofuse1} is $o_P(1)$.

	For the second term in \cref{A-eq:lem:conv-ee-deriv-cf:proofuse1}, it follows from the fact that $\tilde\beta_n \pto \beta^\star$ (lemma assumption), the dominatedness of $\partial_\beta m(\beta,\eta)$ (\cref{A-asu:reg-dominated-ee-deriv}), and the dominated convergence theorem for convergence in probability (see, e.g., \citet[][Result (viii) of Chapter 3.2]{chung2001course}) that
	\begin{align*}
		\PP \{\partial_\beta m(\tilde\beta, \hat\eta_k) \} - \PP \{ \partial_\beta m(\beta^\star, \hat\eta_k) \} = o_P(1).
	\end{align*}

	For the third term \cref{A-eq:lem:conv-ee-deriv-cf:proofuse1}, it being $o_P(1)$ is stated in \cref{A-asu:nuisance-conv-PPee-deriv-cf}.

	Therefore, we showed that all three terms in \cref{A-eq:lem:conv-ee-deriv-cf:proofuse1} are $o_P(1)$. This completes the proof.
\end{proof}

\begin{lem}[Convergence of the ``meat'' term (cross-fitting)]
	\label{A-lem:conv-ee-meat-cf}
	Suppose \cref{A-asu:reg-bounded-ee-deriv-and-meat,A-asu:reg-dominated-ee-meat,A-asu:nuisance-conv-PPee-meat-cf,A-asu:reg-compact-param-space,A-asu:reg-bounded-and-cont-differentiable-ee} hold. If a sequence of random variables $\tilde\beta_n$ satisfies $\tilde\beta_n \pto \beta^\star$, then
	\begin{align*}
		\frac{1}{K} \sum_{k=1}^K \PP_{n,k} \{ m(\tilde\beta_n, \hat\eta_k) m(\tilde\beta_n, \hat\eta_k)^T \} \pto \PP \{ m(\beta^\star, \eta') m(\beta^\star, \eta')^T \}
	\end{align*}
\end{lem}

\begin{proof}[Proof of \cref{A-lem:conv-ee-meat-cf}]
	The proof is almost identical to the proof of \cref{A-lem:conv-ee-deriv-cf} and thus omitted, except that we need to verify that for a fixed $\hat{\eta}_k$, $\{m(\beta, \hat{\eta}_k) m(\beta, \hat{\eta}_k)^T : \beta \in \Theta\}$ is $P$-Glivenko-Cantelli. \cref{A-asu:reg-bounded-and-cont-differentiable-ee} implies that $m(\beta, \hat{\eta}_k) m(\beta, \hat{\eta}_k)^T $ is continuous. By \cref{A-asu:reg-compact-param-space}, \cref{A-asu:reg-dominated-ee-meat}, and Example 19.8 in \citet{van2000asymptotic}, for a fixed $\hat{\eta}_k$, $\{m(\beta, \hat{\eta}_k) m(\beta, \hat{\eta}_k)^T : \beta \in \Theta\}$ is $P$-Glivenko-Cantelli.
\end{proof}

\begin{thm}[Asymptotic Normality (cross-fitting)]
	\label{A-thm:general-normality-cf-formal}
	Suppose \cref{A-asu:unique-zero,A-asu:nuisance-conv-general-cf,A-asu:reg-general,A-asu:reg-bounded-ee-deriv-and-meat} hold. Then
	\begin{align}
		\sqrt{n}(\check\beta - \beta^\star) \dto N ( 0, V) \quad \text{as } n\to\infty, \label{A-eq:asymp-normality-cf}
	\end{align}
	where the asymptotic variance is
	\begin{align}
		V = \PP\{\partial_\beta m(\beta^\star, \eta')\}^{-1} ~ \PP\{ m(\beta^\star, \eta') m(\beta^\star, \eta')^T \} ~ \PP\{\partial_\beta m(\beta^\star, \eta')\}^{-1, T}. \label{A-eq:def-asymp-var-cf}
	\end{align}	
	In addition,
	\begin{align*}
		\PP_n\{\partial_\beta m(\check\beta, \hat\eta)\}^{-1} ~ \PP_n\{ m(\check\beta, \hat\eta) m(\check\beta, \hat\eta)^T \} ~ \PP_n\{\partial_\beta m(\check\beta, \hat\eta)\}^{-1, T}
	\end{align*}
	is a consistent estimator for $V$.	
\end{thm}

\begin{proof}[Proof of \cref{A-thm:general-normality-cf-formal}]
	Because $m(\beta,\eta)$ is continuously differentiable in $\beta$ (\cref{A-asu:reg-bounded-and-cont-differentiable-ee}), the Lagrange mean value theorem implies that
	\begin{align}
		0 & = \frac{1}{K}\sum_{k=1}^K \PP_{n,k}\{ m(\check\beta, \hat\eta_k)\} \nonumber \\
		& = \frac{1}{K}\sum_{k=1}^K \PP_{n,k}\{ m(\beta^\star, \hat\eta_k)\} + \left[ \frac{\partial}{\partial \beta^T} \frac{1}{K}\sum_{k=1}^K \PP_{n,k}\{ m(\tilde\beta, \hat\eta_k)\} \right](\check\beta - \beta^\star), \label{A-eq:thm:general-normality-cf:proofuse1}
	\end{align}
	where $\tilde\beta$ is between $\check\beta$ and $\beta^\star$. 
	Because \cref{A-lem:general-consistency-cf} implies that $\check\beta \pto \beta^\star$ and thus $\tilde\beta \pto \beta^\star$, \cref{A-lem:conv-ee-deriv-cf} implies that
	\begin{align}
		\frac{1}{K}\sum_{k=1}^K \PP_{n,k}\{ \partial_\beta m(\tilde\beta, \hat\eta_k)\} \pto \PP\{\partial_\beta m(\beta^\star, \eta')\}. \label{A-eq:thm:general-normality-cf:proofuse2}
	\end{align}
	Because $\PP\{\partial_\beta m(\beta^\star,\eta')\}$ is invertible (\cref{A-asu:reg-invertible-ee-deriv}), \cref{A-eq:thm:general-normality-cf:proofuse2} implies that $\frac{1}{K}\sum_{k=1}^K \PP_{n,k}\{ \partial_\beta m(\tilde\beta, \hat\eta_k)\}$ is invertible with probability approaching 1. Therefore, \cref{A-eq:thm:general-normality-cf:proofuse1,A-eq:thm:general-normality-cf:proofuse2} imply that with probability approaching 1 we have
	\begin{align}
		\sqrt{n} (\check\beta - \beta^\star) = - [ \PP\{\partial_\beta m(\beta^\star, \eta')\} ]^{-1} \bigg[ \sqrt{n} \frac{1}{K}\sum_{k=1}^K \PP_{n,k}\{ m(\beta^\star, \hat\eta_k)\} \bigg]. \label{A-eq:thm:general-normality-cf:proofuse3}
	\end{align}

	Next we study the asymptotic distribution of the second term in \cref{A-eq:thm:general-normality-cf:proofuse3}. We will first show that for each $k \in [K]$,
	\begin{align}
		\PP_{n,k}\{ m(\beta^\star, \hat\eta_k)\} - \PP_{n,k} \{m(\beta^\star, \eta')\} = o_P(n^{-1/2}). \label{A-eq:thm:general-normality-cf:proofuse4}
	\end{align}
	Because of \cref{A-asu:unique-zero}, we have
	\begin{align}
		\EE [ \PP_{n,k}\{ m(\beta^\star, \hat\eta_k)\} \mid (O_i)_{i \in B_k^c}] = \EE [ \PP_{n,k}\{ m(\beta^\star, \eta')\} \mid (O_i)_{i \in B_k^c} ] = 0. \label{A-eq:thm:general-normality-cf:proofuse5}
	\end{align}
	We also have
	\begin{align}
		& \quad \EE [ \| \PP_{n,k}\{ m(\beta^\star, \hat\eta_k)\} - \PP_{n,k}\{ m(\beta^\star, \eta')\} \|^2 \mid (O_i)_{i \in B_k^c} ] \nonumber \\
		& = \EE \bigg[ \frac{1}{N^2} \sum_{j_1, j_2\in B_k} \{ m(O_{j_1}; \beta^\star, \hat\eta_k) - m(O_{j_1}; \beta^\star, \eta') \} \nonumber \\
& \qquad \times \{ m(O_{j_2}; \beta^\star, \hat\eta_k) - m(O_{j_2}; \beta^\star, \eta') \} \Big\vert (O_i)_{i \in B_k^c} \bigg] \nonumber \\
		& = \frac{1}{N} \EE \{ | m(\beta^\star, \hat\eta_k) - m(\beta^\star, \eta') |^2 \mid (O_i)_{i \in B_k^c} \} \label{A-eq:thm:general-normality-cf:proofuse6} \\
		& = o_P(n^{-1}), \label{A-eq:thm:general-normality-cf:proofuse7}
	\end{align}
	where \cref{A-eq:thm:general-normality-cf:proofuse6} follows from the fact that the cross-product term has expectation zero if $j_1 \neq j_2$ due to \cref{A-asu:unique-zero}, and \cref{A-eq:thm:general-normality-cf:proofuse7} follows from \cref{A-asu:nuisance-conv-ee-l2-cf}. \cref{A-eq:thm:general-normality-cf:proofuse5,A-eq:thm:general-normality-cf:proofuse7} and Chebyshev's inequality imply that \cref{A-eq:thm:general-normality-cf:proofuse4} holds when conditional on $(O_i)_{i \in B_k^c}$, and Lemma 6.1 in \citet{chernozhukov2018double} implies that \cref{A-eq:thm:general-normality-cf:proofuse4} holds unconditionally. Thus we established \cref{A-eq:thm:general-normality-cf:proofuse4}.

	\cref{A-eq:thm:general-normality-cf:proofuse4} implies that
	\begin{align}
		\sqrt{n} \frac{1}{K}\sum_{k=1}^K \PP_{n,k}\{ m(\beta^\star, \hat\eta_k)\} & = \sqrt{n}\PP_n \{m(\beta^\star, \eta')\} + o_P(1) \nonumber \\
		& \dto N(0, \PP \{m(\beta^\star, \eta') m(\beta^\star, \eta')^T\}), \label{A-eq:thm:general-normality-cf:proofuse8}
	\end{align}
	where the convergence in distribution follows from the Lindeberg-Feller Central Limit Theorem. Plugging \cref{A-eq:thm:general-normality-cf:proofuse8} into \cref{A-eq:thm:general-normality-cf:proofuse3}, the desired asymptotic normality result \cref{A-eq:asymp-normality-cf,A-eq:def-asymp-var-cf} follows from Slutsky's theorem.

	Consistency of the variance estimator follows immediately from \cref{A-lem:conv-ee-deriv-cf,A-lem:conv-ee-meat-cf} and the continuous mapping theorem. This completes the proof.
\end{proof}

\section{Asymptotic Normality of \texorpdfstring{$\hat\beta$}{beta-hat} and \texorpdfstring{$\check\beta$}{beta-check} for CEE}
\label{A-sec:continuous-can}

The proofs for the identity link and the log link are almost identical. Therefore, to be concise, we only present the proof for the identity link.

The true parameter value $\beta^\star$ satisfies \cref{eq:identify-cee}, which implies
\begin{align}
	\PP \{ \PP ( Y_{t+1} \mid H_t, A_t = 1 ) - \PP (Y_{t+1} \mid H_t, A_t = 0 ) \mid S_t, I_t = 1 \} = \gamma_t(S_t;\beta^\star) . \label{A-eq:beta0-continuous}
\end{align}

Consider the following estimating function for $\beta$:
\begin{align}
	\mc(\beta, \eta) = \sum_{t=1}^T d_t(S_t; \mu_t, \betainit) \phic_t(\beta, \mu_t), \label{A-eq:ee-continuous}
\end{align}
where 
\begin{align*}
    \phic_t(\beta, \mu_t) & = \frac{A_t - p_t}{p_t (1 - p_t)} I_t ( Y_{t+1} - (A_t + p_t - 1) \gamma_t(S_t; \beta) - (1 - p_t)\mu_t(H_t, 1) \\
& \quad - p_t \mu_t(H_t, 0)) \partial_{\beta} \gamma_t(S_t;\beta),
\end{align*}
and $\mu_t(H_t, 0)$ and $\mu_t(H_t, 1)$ are nuisance functions that take value in $\RR$. 

\subsection{The Estimator $\check\beta$ With Cross-Fitting and Its Asymptotic Normality}
\label{A-subsec:crossfit-estimator}

This appendix presents the cross-fitted estimator $\check\beta$ and its asymptotic normality; the non-cross-fitted estimator $\hat\beta$ in \cref{algo:estimator-ncf} and its asymptotic normality in \cref{thm:normality-ncf} are the main-text counterparts. The cross-fitted estimator $\check\beta$ presented here removes the Donsker condition required by $\hat\beta$ and is preferred when the nuisance estimators are non-Donsker (e.g., random forests, neural networks, super learner with such learners in its library).

\begin{algorithm}
	\caption{Estimator $\check\beta$ (using cross-fitting)}
	\label{A-algo:estimator-cf}
	\textbf{Step 1:} Take a $K$-fold equally-sized random partition $(B_k)_{k=1}^K$ of observation indices $[n] = \{1,\ldots,n\}$. Define $B_k^c = [n] \setminus B_k$ for $k \in [K]$.

	\textbf{Step 2:} For each $k \in [K]$, using only observations from $B_k^c$:
	\begin{itemize}[rightmargin=1.5em]
		\item[(i)] Obtain an estimator $\hat\mu_{kt}(H_t, A_t)$ of $\EE(Y_{t+1} \mid H_t, A_t, I_t = 1)$ for $t \in [T]$.
		\item[(ii)] Obtain $\betainithat_k$ by solving $\PP_n^{B_k^c}\{\sum_{t=1}^T \phi_t(\beta, \hat\mu_{kt})\} = 0$ for $\beta$, where $\PP_n^{B_k^c}$ denotes empirical average over observations in $B_k^c$.
		\item[(iii)] Obtain estimators of $\EE\{\partial_\beta \phi_t(\betainithat_k, \hat\mu_{kt}) \mid S_t\}$ and $\EE\{\phi_t(\betainithat_k, \hat\mu_{kt})^{\otimes 2} \mid S_t\}$ for $t \in [T]$, and plug them into \cref{eq:dt-definition-working} to obtain $\hat{d}_{kt}(S_t; \hat\mu_{kt}, \betainithat_k)$.
	\end{itemize}

	\textbf{Step 3:} Obtain $\check\beta$ by solving \[ K^{-1} \sum_{k=1}^K \PP_{n,k} \sum_{t=1}^T \hat{d}_{kt}(S_t; \hat\mu_{kt}, \betainithat_k) \phi_t(\beta, \hat\mu_{kt}) = 0 \] for $\beta$. Here $\PP_{n,k}$ denotes empirical average over observations from $B_k$.
\end{algorithm}

\noindent\textit{Implementation note.} As with \cref{algo:estimator-ncf}, the estimators in Steps~2(i) and~2(iii) may be obtained by fitting separately for each $t$ or by pooling across $t$ (e.g., a single regression with $t$ entered as a covariate); the choice does not affect the asymptotic theory. \Cref{A-subsec:practical-implementation} describes the pooled implementation used in our simulations and data application.

\begin{thm}[Asymptotic normality of $\check\beta$]
	\label{A-thm:normality-cf}
	Suppose \cref{eq:linear-model-cee} holds. In addition, suppose the following hold.
	\begin{enumerate}
		\item[(a)] For each $t \in [T]$, $\hat\mu_{kt}$ converges to some limit in $L_2$ for all $k \in [K]$. In other words, there exists a function $\mu_t'(h_t, a_t)$ such that $\|\hat\mu_{kt} - \mu_t'\| = o_P(1)$ for all $k \in [K]$.

		\item[(b)] For each $t \in [T]$, $\hat{d}_{kt}(s_t; \mu_t', \beta^\star)$ converges to some limit in $L_2$ for all $k \in [K]$. In other words, there exists a function $d_t'(s_t)$ such that $\|\hat{d}_{kt} - d_t'\| = o_P(1)$ for all $k \in [K]$.

		\item[(c)] Given the limit functions $\{\mu_t'(h_t, a_t): t \in [T]\}$ and $\{d_t'(s_t): t \in [T]\}$, \[ \PP \sum_{t=1}^T d_t'(S_t) \phi_t(\beta, \mu_t') = 0 \] as a function of $\beta$ has a unique zero.

		\item[(d)] Suppose the parameter space $\Theta$ of $\beta$ is compact, the support of $O$ is bounded, $\hat\mu_{kt}$ and $\mu_t'$ are bounded for each $t \in [T]$ and $k \in [K]$, and $\PP\{\sum_{t=1}^T d_t'(S_t) \partial_\beta \phi_t(\beta^\star, \mu_t') \}$ is invertible.
	\end{enumerate}
	Conditions (a)--(d) are informal restatements of the formal assumptions in \cref{A-subsec:asu-continuous}.
	Then we have $\sqrt{n}(\check\beta - \beta^\star) \dto N ( 0, V)$ as $n\to\infty$, where
	\begin{align}
		V = {} & \PP\Big\{\sum_{t=1}^T d_t'(S_t) \partial_\beta \phi_t(\beta^\star, \mu_t') \Big\}^{-1} \PP \Big[ \Big\{ \sum_{t=1}^T d_t'(S_t) \phi_t(\beta^\star, \mu_t') \Big\}^{\otimes 2} \Big] \nonumber \\
& \times \PP\Big\{\sum_{t=1}^T d_t'(S_t) \partial_\beta \phi_t(\beta^\star, \mu_t') \Big\}^{-1, T}, \label{A-eq:avar}
	\end{align}
	and $V$ can be consistently estimated by
  \begin{small}
	\begin{align*}
		& \bigg[\frac{1}{K}\sum_{k=1}^K \PP_{n,k}\Big\{\sum_{t=1}^T \hat{d}_{kt}(S_t; \hat\mu_{kt}, \betainithat_k) \partial_\beta \phi_t(\check\beta, \hat\mu_{kt}) \Big\} \bigg]^{-1} \\
& \times \bigg(\frac{1}{K}\sum_{k=1}^K \PP_{n,k} \Big[ \Big\{ \sum_{t=1}^T \hat{d}_{kt}(S_t; \hat\mu_{kt}, \betainithat_k) \phi_t(\check\beta, \hat\mu_{kt}) \Big\}^{\otimes 2} \Big] \bigg) \\
		& \times \bigg[\frac{1}{K}\sum_{k=1}^K \PP_{n,k}\Big\{\sum_{t=1}^T \hat{d}_{kt}(S_t; \hat\mu_{kt}, \betainithat_k) \partial_\beta \phi_t(\check\beta, \hat\mu_{kt}) \Big\} \bigg]^{-1, T}.
	\end{align*}
  \end{small}
\end{thm}

The proof of \cref{A-thm:normality-cf} is given in \cref{A-subsec:normality-continuous} below; specifically, \cref{A-thm:continuous-normality-cf} restates the same result using the formal assumptions in \cref{A-subsec:asu-continuous} and proves it as a consequence of the general theory in \cref{thm:general-normality-cf}. The robustness implications of \cref{A-thm:normality-cf} parallel those for $\hat\beta$ stated in \cref{thm:normality-ncf} and the surrounding remarks in \cref{subsec:proposed-estimator}; the cross-fitting in \cref{A-algo:estimator-cf} removes condition (e) (the Donsker condition) of \cref{thm:normality-ncf}.

\subsection{Assumptions}
\label{A-subsec:asu-continuous}

The assumptions used for establishing the consistency and asymptotic normality of $\hat\beta$ and $\check\beta$ for the continuous outcome case are stated in \cref{thm:normality-ncf} of the main paper and \cref{A-thm:normality-cf} of this appendix, respectively, and for reference we restate them below.

\begin{asu}[Unique zero]
	\label{A-asu:continuous-unique-zero}
	Given the limit functions $\{\mu_t'(h_t, a_t): t \in [T]\}$ and $\{d_t'(s_t): t \in [T]\}$, $\PP\{\mc(\beta, \eta')\} = 0$ as a function of $\beta$ has a unique zero.
\end{asu}

\begin{asu}[Convergence of nuisance parameter estimator]
	\label{A-asu:continuous-nuisance-converge}
	Suppose for each $t$, $\hat\mu_t$ converges in $L_2$ to some limit. In other words, there exists a function $\mu'_t(h_t, a_t)$ such that
	\begin{align*}
		\| \hat\mu_t - \mu_t' \|^2 & = \int |\hat\mu_t - \mu_t' |^2 dP = o_P(1).
	\end{align*}
	When considering the cross-fitted estimator $\check\beta$, this assumption is revised by replacing $\hat\mu_t$ with $\hat\mu_{kt}$.
\end{asu}

\begin{asu}[Convergence of d term]
	\label{A-asu:continuous-d-converge}
	Suppose for each $t$, $\hat d_t(s_t; \mu'_t, \beta^\star)$ converges in $L_2$ to some limit. In other words, there exists a function $d'_t(s_t)$ such that
	\begin{align*}
		\| \hat d_t - d_t' \|^2 & = \int |\hat d_t(s_t;\mu_t',\beta^\star) - d'_t(s_t) |_F^2 dP = o_P(1).
	\end{align*}
	When considering the cross-fitted estimator $\check\beta$, this assumption is revised by replacing $\hat d_t$ with $\hat d_{kt}$.
\end{asu}

\begin{asu}[Regularity conditions]
	\label{A-asu:continuous-regularity}
	~
	\begin{asulist}
		\item \label{A-asu:continuous-compact-param-space} Suppose the parameter space $\Theta$ of $\beta$ is compact.
		\item \label{A-asu:continuous-bounded-obs} Suppose the support of $O$ is bounded.
		\item \label{A-asu:continuous-bounded-gamma} Suppose for each $t$, $\hat \mu_t$ is bounded over $o \in \cO$.
            \item \label{A-asu:continuous-bounded-gamma-prime}
        Suppose for each $t$, $\mu_t'$ is bounded over $o \in \cO$.
		\item \label{A-asu:continuous-invertible-deriv} Suppose $\PP\{\partial_\beta \mc(\beta^\star, \eta')\}$ is invertible.
            \item \label{A-asu:continuous-continuous-d-hat}
        Suppose $\hat d_t(s_t; \mu_t, \betainit)$ is continuous in $\mu_t$ (with respect to the $L_2(P)$ norm) and in $\betainit$.
            \item \label{A-asu:continuous-uniform-integrable}
        Suppose for each $t$, $\hat d_t(s_t; \hat \mu_t, \betainithat)$ is uniformly bounded in $L_2$.
            \item \label{A-asu:continuous-hat-tilde-d-bounded}
        Suppose for each $t$, $\hat d_t(s_t; \mu_t', \beta^\star)$ is bounded over $o \in \cO$.
	\end{asulist}
\end{asu}

\begin{asu}[Donsker condition]
	\label{A-asu:continuous-donsker-nuisance}
	Suppose for each $t$, the estimators $\hat\mu_t$ and $\hat d_t$ each take values in a $P$-Donsker class.
\end{asu}

\subsection{Lemmas}
\label{A-subsec:lemma-continuous}

\begin{lem}
	\label{A-lem:continuous-beta0-is-zero}
	$\PP\{\mc(\beta^\star, \eta)\} = 0$ for any $\eta \in \cT$.
\end{lem}

\begin{proof}[Proof of \cref{A-lem:continuous-beta0-is-zero}]
	For any $\beta \in \Theta$, $\eta \in \cT$, we have
	\begin{align}
		&\PP \{\mc(\beta, \eta) \} \nonumber \\
		& = \sum_t \PP \bigg[ \PP \Big\{ d_t(S_t; \mu_t, \betainit) I_t \frac{A_t - p_t}{p_t (1 - p_t)} ( Y_{t+1} - (A_t + p_t - 1) \gamma_t(S_t; \beta) - (1 - p_t)\mu_t(H_t, 1)\nonumber \\
  & ~~~ - p_t \mu_t(H_t, 0)) \partial_\beta \gamma_t(S_t; \beta) \mid H_t \Big\} \bigg] \nonumber \\
		& = \sum_t \PP \bigg[ I_t \PP \Big\{ d_t(S_t; \mu_t, \betainit) \frac{A_t - p_t}{p_t (1 - p_t)} ( Y_{t+1} - (A_t + p_t - 1) \gamma_t(S_t; \beta) - (1 - p_t)\mu_t(H_t, 1)  \nonumber \\
  & ~~~ - p_t \mu_t(H_t, 0)) \partial_\beta \gamma_t(S_t; \beta) \mid H_t, A_t = 1 \Big\} p_t(H_t) \bigg] \nonumber \\
		& ~~~+ \sum_t \PP \bigg[ I_t \PP \Big\{ d_t(S_t; \mu_t, \betainit) \frac{A_t - p_t}{p_t (1 - p_t)} ( Y_{t+1} - (A_t + p_t - 1) \gamma_t(S_t; \beta) - (1 - p_t)\mu_t(H_t, 1) \nonumber \\
  & ~~~~~~ - p_t \mu_t(H_t, 0)) \partial_\beta \gamma_t(S_t; \beta) \mid H_t, A_t = 0 \Big\} \Big\{1 - p_t(H_t)\Big\} \bigg] \nonumber \\
		& = \sum_t \PP \bigg[ I_t \PP \Big\{ d_t(S_t; \mu_t, \betainit) ( Y_{t+1} - p_t \gamma_t(S_t; \beta) - (1 - p_t)\mu_t(H_t, 1) \nonumber \\
& ~~~ - p_t \mu_t(H_t, 0)) \partial_\beta \gamma_t(S_t; \beta) \mid H_t, A_t = 1 \Big\} \bigg] \nonumber \\
		& ~~~- \sum_t \PP \bigg[ I_t \PP \Big\{ d_t(S_t; \mu_t, \betainit) ( Y_{t+1} - (p_t - 1) \gamma_t(S_t; \beta) - (1 - p_t)\mu_t(H_t, 1) \nonumber \\
  & ~~~~~~ - p_t \mu_t(H_t, 0)) \partial_\beta \gamma_t(S_t; \beta) \mid H_t, A_t = 0 \Big\} \bigg] \nonumber
  \end{align}
  \begin{align}
		& = \sum_t \PP \bigg[ d_t(S_t; \mu_t, \betainit) I_t \Big\{ \PP ( Y_{t+1} \mid H_t, A_t = 1 ) - \PP (Y_{t+1} \mid H_t, A_t = 0 ) \nonumber \\
& \qquad - \gamma_t(S_t; \beta) \Big\} \partial_\beta \gamma_t(S_t; \beta) \bigg] \label{A-eq:continuous-beta0-is-zero:proofuse1} \\
		& = \sum_t \PP \bigg[ d_t(S_t; \mu_t, \betainit) I_t \PP \Big\{ \PP ( Y_{t+1} \mid H_t, A_t = 1 ) - \PP (Y_{t+1} \mid H_t, A_t = 0 ) \nonumber \\
& \qquad - \gamma_t(S_t; \beta) \mid S_t, I_t \Big\} \partial_\beta \gamma_t(S_t; \beta) \bigg] \nonumber \\
		& = \sum_t \PP \bigg[ d_t(S_t; \mu_t, \betainit) I_t \PP \Big\{ \PP ( Y_{t+1} \mid H_t, A_t = 1 ) - \PP (Y_{t+1} \mid H_t, A_t = 0 ) \nonumber \\
& \qquad - \gamma_t(S_t; \beta) \mid S_t, I_t = 1 \Big\} \partial_\beta \gamma_t(S_t; \beta) \bigg], \label{A-eq:continuous-beta0-is-zero:proofuse2}
	\end{align}
	where \cref{A-eq:continuous-beta0-is-zero:proofuse1} follows from the fact that $\mu_t(H_t, 0)$ and $\mu_t(H_t, 1)$ got canceled, \cref{A-eq:continuous-beta0-is-zero:proofuse2} follows from the fact that $I_t$ can only take two values 0 and 1. Therefore, replacing $\beta$ by $\beta^\star$, the lemma result follows from the modeling assumption \cref{A-eq:beta0-continuous}.
\end{proof}

\begin{lem}
    \label{A-lem:continuous-ee-betainit-converge}
    Suppose \cref{A-asu:continuous-unique-zero,A-asu:continuous-donsker-nuisance,A-asu:continuous-compact-param-space,A-asu:continuous-bounded-obs} hold. Then $\betainithat \pto \beta^\star$ as $n \to \infty$.
\end{lem}

\begin{proof}[Proof of \cref{A-lem:continuous-ee-betainit-converge}]
    By \cref{A-lem:general-consistency}, it suffices to verify \cref{A-asu:unique-zero,A-asu:nuisance-conv-PPee-sup,A-asu:donsker,A-asu:reg-compact-param-space,A-asu:reg-bounded-obs,A-asu:reg-cont-PPee}.

    Verify \cref{A-asu:unique-zero}: This follows from \cref{A-asu:continuous-unique-zero} and \cref{A-lem:continuous-beta0-is-zero} by setting $d_t = 1$.

    Verify \cref{A-asu:reg-compact-param-space}: This is assumed in \cref{A-asu:continuous-compact-param-space}.
    
    Verify \cref{A-asu:reg-bounded-obs}: This is assumed in \cref{A-asu:continuous-bounded-obs}. 
    
    Verify \cref{A-asu:reg-cont-PPee}: This holds as $\phi_t(\mu_t, \beta)$ is linear in $\beta$. 

    Verify \cref{A-asu:nuisance-conv-PPee-sup}: This follows from \cref{A-lem:continuous-ee-sup-conv} by setting $d_t = 1$.

    Verify \cref{A-asu:donsker}: This is because $\mu_t(\cdot)$ and $d_t(\cdot)$ are from Donsker classes (\cref{A-asu:continuous-donsker-nuisance}), $m(\beta,\eta)$ is a polynomial of $\mu_t$ and $d_t$ and thus Lipschitz, and a Lipschitz transformation of Donsker classes is still Donsker \citep[][Theorem 9.31]{kosorok2008introduction}.
\end{proof}

\begin{lem}
    \label{A-lem:continuous-ee-phi-converge}
	Suppose \cref{A-asu:continuous-compact-param-space,A-asu:continuous-bounded-obs,A-asu:continuous-nuisance-converge,A-asu:continuous-bounded-gamma,A-asu:continuous-bounded-gamma-prime}, and \cref{asu:positivity} hold. Then for $\phic(\beta, \eta) := \sum_{t=1}^T \phic_t(\beta, \mu_t)$, we have
	\begin{align*}
		\|\phic(\beta^\star, \hat\eta) - \phic(\beta^\star, \eta') \|^2 = o_P(1).
	\end{align*}
	Under the same assumptions, the lemma statement also holds with $\hat\eta$ replaced by the cross-fitting version, $\hat\eta_k$.
\end{lem}

\begin{proof}[Proof of \cref{A-lem:continuous-ee-phi-converge}]
	We present the proof for the non-cross-fitting version. The proof for the cross-fitting version is identical with $\hat\eta$, $\hat\mu_t$ replaced by $\hat\eta_k$, $\hat\mu_{kt}$ respectively. By \cref{A-eq:ee-continuous}, we have $\phic(\beta,\eta) = \sum_{t=1}^T \phic_t(\beta, \mu_t)$ where 
	\begin{align*}
		\phic_t(\beta, \mu_t) & = \frac{A_t - p_t}{p_t (1 - p_t)} I_t ( Y_{t+1} - (A_t + p_t - 1) \gamma_t(S_t; \beta) - (1 - p_t)\mu_t(H_t, 1) \\
& \quad - p_t \mu_t(H_t, 0)) \partial_{\beta} \gamma_t(S_t;\beta).
	\end{align*}

	Using this definition we have
	\begin{align}
		& ~~~ \|\phic(\beta^\star, \hat\eta) - \phic(\beta^\star, \eta') \|^2 \nonumber \\
		& = \int |\phic(\beta^\star, \hat\eta) - \phic(\beta^\star, \eta')|^2 dP \nonumber \\
		& = \int \bigg|\sum_{t=1}^T \left\{ \phic_t(\beta^\star, \hat\mu_t) - \phic_t(\beta^\star, \mu_t') \right\} \bigg|^2 dP \nonumber \\
		& = \int \bigg|\sum_{t=1}^T \left[ \left\{ \phic_t(\beta^\star, \hat\mu_t) - \phic_t(\beta^\star, \mu_t') \right\} \right] \bigg|^2 dP \nonumber \\
		& \leq T \left\{\max_{1\leq t \leq T} \int \left| \phic_t(\beta^\star, \hat\mu_t) - \phic_t(\beta^\star, \mu_t') \right|^2 dP \right\}, \nonumber
	\end{align}
	Therefore, it suffices to show that for all $1 \leq t \leq T$,
	\begin{align}
		\int \left| \phic_t(\beta^\star, \hat\mu_t) - \phic_t(\beta^\star, \mu_t') \right|^2 dP = o_P(1), \label{A-eq:continuous-ee-phi-converge:proofuse1}
	\end{align}

	To show \cref{A-eq:continuous-ee-phi-converge:proofuse1}, we define
    \begin{align*}
        \psi_t(x, y) := (1 - p_t(H_t)) x + p_t(H_t)y, 
    \end{align*}
        and then we have
	\begin{align}
		& ~~~ \int \left| \phic_t(\beta^\star, \hat\mu_t) - \phic_t(\beta^\star, \mu_t') \right|^2 dP \nonumber \\
		& = \int \bigg| - \frac{A_t - p_t}{p_t (1 - p_t)} I_t ((1 - p_t)\hat \mu_t(H_t, 1) + p_t \hat \mu_t(H_t, 0) - (1 - p_t) \mu_t'(H_t, 1) \nonumber \\
& \qquad - p_t \mu_t'(H_t, 0)) \partial_{\beta} \gamma_t(S_t;\beta^\star) \bigg|^2 dP \nonumber \\
		& \leq C \int | \psi_t(\hat\mu_t(H_t, 1), \hat\mu_t(H_t, 0)) - \psi_t(\mu_t'(H_t, 1), \mu_t'(H_t, 0)) |^2 dP \label{A-eq:continuous-ee-phi-converge:proofuse3} \\
		& = o_P(1), \label{A-eq:continuous-ee-phi-converge:proofuse4}
	\end{align}
	where \cref{A-eq:continuous-ee-phi-converge:proofuse3} follows from \cref{A-asu:continuous-bounded-obs,A-asu:continuous-compact-param-space}, $A_t \in \{0, 1\}$, and \cref{asu:positivity}. Because $\psi_t$ is a weighted average of $x$ and $y$, that $p_t$ is bounded away from $0$ or $1$ (\cref{asu:positivity}), $\hat \mu_t$ and $\mu_t'$ are bounded (\cref{A-asu:continuous-bounded-gamma,A-asu:continuous-bounded-gamma-prime}), we have that $\psi_t(x, y)$ is bounded. Therefore, the last line \cref{A-eq:continuous-ee-phi-converge:proofuse4} follows from the convergence of $\hat\mu_t$ (\cref{A-asu:continuous-nuisance-converge}). This establishes \cref{A-eq:continuous-ee-phi-converge:proofuse1}.

	Having \cref{A-eq:continuous-ee-phi-converge:proofuse1} established, the proof is completed.
\end{proof}

\begin{lem}
    \label{A-lem:continuous-ee-d-converge}
    Suppose \cref{A-asu:continuous-nuisance-converge,A-asu:continuous-unique-zero,A-asu:continuous-donsker-nuisance,A-asu:continuous-compact-param-space,A-asu:continuous-bounded-obs,A-asu:continuous-continuous-d-hat,A-asu:continuous-uniform-integrable} hold. Then $\| \hat d_t(s_t; \hat \mu_t, \betainithat) - \hat d_t (s_t; \mu_t', \beta^\star)\|^2 = o_P(1)$.
    Under the same assumptions, the lemma statement also holds with $\hat d_t$ replaced by the cross-fitting version, $\hat d_{kt}$.

    \begin{proof}
        We present the proof for the non-cross-fitting version. The proof for the cross-fitting version is identical with $\hat\mu_t$, $\hat d_t$, $\betainithat$ replaced by $\hat\mu_{kt}$, $\hat d_{kt}$, $\betainithat_k$ respectively. We first show that $\hat d_t(s_t; \hat \mu_t, \betainithat) \pto \hat d_t (s_t; \mu_t', \beta^\star)$. This follows from that $\| \hat \mu_t - \mu_t' \|^2 = o_P(1)$ (\cref{A-asu:continuous-nuisance-converge}) implies $\hat \mu_t \pto \mu_t'$, $\betainithat \pto \beta^\star$ (\cref{A-lem:continuous-ee-betainit-converge}), and $\hat d_t$ is continuous in $\mu_t$ and $\betainit$ (\cref{A-asu:continuous-continuous-d-hat}). Because of the $L_2$ dominatedness of $\hat d_t(s_t; \hat \mu_t, \betainithat)$ (\cref{A-asu:continuous-uniform-integrable}), by the dominated convergence theorem (\citealp[Theorem 4.1.4 of Chapter 4.1]{chung2001course})
        implies that $\hat d_t(s_t; \hat \mu_t, \betainithat)$ converges to $\hat d_t(s_t; \mu_t', \beta^\star)$ in $L_2$.
    \end{proof}
    
\end{lem}

\begin{lem}
	\label{A-lem:continuous-ee-sup-conv}
	Suppose \cref{A-asu:continuous-bounded-obs,A-asu:continuous-d-converge,A-asu:continuous-compact-param-space} hold. Then for $\mc(\beta, \eta)$ defined in \cref{A-eq:ee-continuous}, we have
	\begin{align*}
		\sup_{\beta \in \Theta} | \PP \mc(\beta, \hat\eta) - \PP \mc(\beta, \eta') | = o_P(1).
	\end{align*}
	Under the same assumptions, the lemma statement also holds with $\hat\eta$ replaced by the cross-fitting version, $\hat\eta_k$.
\end{lem}

\begin{proof}[Proof of \cref{A-lem:continuous-ee-sup-conv}]
	We present the proof for the non-cross-fitting version. The proof for the cross-fitting version is identical with $\hat\eta$, $\hat\mu_t$, $\hat d_t$ replaced by $\hat\eta_k$, $\hat\mu_{kt}$, $\hat d_{kt}$, respectively. Here, we write out the dependency of $d_t$ on $O$. Let
	\begin{align*}
		f_t(\beta, O) & = I_t \left( \PP ( Y_{t+1} \mid H_t, A_t = 1 ) - \PP (Y_{t+1} \mid H_t, A_t = 0 ) - \gamma_t(S_t; \beta) \right) \nonumber \\
            g_t(\beta, O; d_t, \mu_t, \betainit) & = d_t(S_t; \mu_t, \betainit) \partial_{\beta} \gamma_t(S_t;\beta).
	\end{align*}
	and \cref{A-eq:continuous-beta0-is-zero:proofuse1} implies that
	\begin{align*}
		\PP\{\mc(\beta, \eta)\} = \PP\{ \sum_t f_t(\beta, O) g_t(\beta, O; d_t, \mu_t, \betainit) \}.
	\end{align*}
	Therefore, we have
	\begin{align}
		& ~ \sup_{\beta \in \Theta} | \PP\{\mc(\beta, \hat\eta)\} - \PP\{\mc(\beta, \eta')\} | \nonumber \\
		& = \sup_{\beta \in \Theta}  \left| \PP \left\{ \sum_t f_t(\beta, O) g_t(\beta, O; \hat d_t,\hat \mu_t, \betainithat) \right\} - \PP \left\{ \sum_t f_t(\beta, O) g_t(\beta, O; d') \right\} \right| \nonumber \\
		& \leq \sum_t \sup_{\beta \in \Theta} \left| \PP \left[ f_t(\beta, O) \{g_t(\beta, O; \hat d_t, \hat \mu_t, \betainithat) - g_t(\beta, O; d_t')\} \right] \right| \nonumber \\
		& \leq \sum_t \sup_{\beta \in \Theta} \| f_t(\beta, \cdot) \| \cdot \|g_t(\beta, \cdot; \hat d_t, \hat \mu_t, \betainithat) - g_t(\beta, \cdot; d_t')\|, \label{A-eq:thm:continuous-normality:proofuse1} \\
            & \leq \sum_t \sup_{\beta \in \Theta} \| f_t(\beta, \cdot) \| \cdot \Big(\|g_t(\beta, \cdot; \hat d_t, \hat \mu_t, \betainithat) - g_t(\beta, \cdot; \hat d_t, \mu_t', \beta^\star) \| \nonumber \\
& \qquad + \| g_t(\beta, \cdot; \hat d_t, \mu_t', \beta^\star) - g_t(\beta, \cdot; d_t')\| \Big) \nonumber
	\end{align}
	where the last two lines follow from Cauchy-Schwarz inequality. It remains to control the two terms from the last line. 
	To control the first term $\|g_t(\beta, \cdot; \hat d_t, \hat \mu_t, \betainithat) - g_t(\beta, \cdot; \hat d_t, \mu_t', \beta^\star) \|$, we have
	\begin{align}
		& \|g_t(\beta, \cdot; \hat d_t, \hat \mu_t, \betainithat) - g_t(\beta, \cdot; \hat d_t, \mu_t', \beta^\star) \|^2 \nonumber \\
           & = \int \left| g_t(\beta, \cdot; \hat d_t, \hat \mu_t, \betainithat) - g_t(\beta, \cdot; \hat d_t, \mu_t', \beta^\star) \right|^2 dP(o) \nonumber \\
           & \leq C \int \left| \hat d_t(s_t; \hat \mu_t, \betainithat) - \hat d_t(s_t; \mu_t', \beta^\star) \right|_F^2 dP(o) \label{A-eq:thm:continuous-normality:proofuse3.5} \\
		& = o_P(1), \label{A-eq:thm:continuous-normality:proofuse3}
	\end{align}
	where \cref{A-eq:thm:continuous-normality:proofuse3.5} follows from that all the observations are bounded (\cref{A-asu:continuous-bounded-obs}) and that the Frobenius norm is compatible with the Euclidean norm, and \cref{A-eq:thm:continuous-normality:proofuse3} follows from \cref{A-lem:continuous-ee-d-converge}. 
    To control the second term $\| g_t(\beta, \cdot; \hat d_t, \mu_t', \beta^\star) - g_t(\beta, \cdot; d_t') \|$, we have
    \begin{align}
		& \| g_t(\beta, \cdot; \hat d_t, \mu_t', \beta^\star) - g_t(\beta, \cdot; d_t') \|^2 \nonumber \\
           & = \int \left| g_t(\beta, \cdot; \hat d_t, \mu_t', \beta^\star) - g_t(\beta, \cdot; d_t') \right|^2 dP(o) \nonumber \\
           & \leq C \int \left| \hat d_t(s_t; \mu_t', \beta^\star) - d_t'(s_t) \right|_F^2 dP(o) \label{A-eq:thm:continuous-normality:proofuse4.5} \\
		& = o_P(1), \label{A-eq:thm:continuous-normality:proofuse4}
	\end{align}
    where \cref{A-eq:thm:continuous-normality:proofuse4.5} follows from that all the observations are bounded (\cref{A-asu:continuous-bounded-obs}) and that the Frobenius norm is compatible with the Euclidean norm, and \cref{A-eq:thm:continuous-normality:proofuse4} follows from \cref{A-asu:continuous-d-converge}. $\sup_{\beta \in \Theta} \| f_t(\beta, \cdot) \|$ is bounded because of \cref{A-asu:continuous-compact-param-space,A-asu:continuous-bounded-obs}.
	\cref{A-eq:thm:continuous-normality:proofuse3,A-eq:thm:continuous-normality:proofuse4} and the boundedness of $\sup_{\beta \in \Theta} \| f_t(\beta, \cdot) \|$ imply \cref{A-asu:nuisance-conv-PPee-sup}.
\end{proof}

\begin{lem}
	\label{A-lem:continuous-ee-l2-conv}
	Suppose \cref{A-asu:continuous-bounded-gamma-prime,A-asu:continuous-bounded-obs,A-asu:continuous-uniform-integrable,A-asu:continuous-d-converge} and \cref{asu:positivity} hold. Then for $\mc(\beta, \eta)$ defined in \cref{A-eq:ee-continuous}, we have
	\begin{align*}
		\|\mc(\beta^\star, \hat\eta) - \mc(\beta^\star, \eta') \|^2 = o_P(1).
	\end{align*}
	Under the same assumptions, the lemma statement also holds with $\hat\eta$, $\betainithat$, $\hat d$ replaced by the cross-fitting version, $\hat\eta_k$, $\betainithat_k$, $\hat d_k$, respectively.
\end{lem}

\begin{proof}[Proof of \cref{A-lem:continuous-ee-l2-conv}]
	We present the proof for the non-cross-fitting version. The proof for the cross-fitting version is identical with $\hat\eta$, $\hat\mu_t$, $\hat d_t$ replaced by $\hat\eta_k$, $\hat\mu_{kt}$, $\hat d_{kt}$, respectively. 

        Here, we write out the dependence of $\phi_t$ on the observation $O$.

	Using the definition in \cref{A-eq:ee-continuous} we have
	\begin{align}
		& ~~~ \|\mc(\beta^\star, \hat\eta) - \mc(\beta^\star, \eta') \|^2 \nonumber \\
		& = \int |\mc(\beta^\star, \hat\eta) - \mc(\beta^\star, \eta')|^2 dP \nonumber \\
		& = \int \bigg|\sum_{t=1}^T \left\{ \hat d_t(S_t; \hat \mu_t, \betainithat) \phi_t(\beta^\star, \hat \mu_t, o) -  d_t'(s_t) \phi_t(\beta^\star, \mu_t', o) \right\} \bigg|^2 dP \nonumber \\
		& = \int \bigg|\sum_{t=1}^T \bigg[ \left\{ \hat d_t(s_t; \hat \mu_t, \betainithat) \phi_t(\beta^\star, \hat \mu_t, o) - \hat d_t(s_t; \hat \mu_t, \betainithat) \phi_t(\beta^\star, \mu_t', o) \right\} \nonumber \\
            & ~~~ + \left\{ \hat d_t(s_t; \hat \mu_t, \betainithat) \phi_t(\beta^\star, \mu_t', o) - \hat d_t(s_t; \mu_t', \beta^\star) \phi_t(\beta^\star, \mu_t', o) \right\} \nonumber \\
            & ~~~ + \left\{ \hat d_t(s_t; \mu_t', \beta^\star) \phi_t(\beta^\star, \mu_t', o) -  d_t'(s_t) \phi_t(\beta^\star, \mu_t', o) \right\} \bigg] \bigg|^2 dP \nonumber \\
		& \leq 4 T^2 \bigg\{\max_{1\leq t \leq T} \int \left| \hat d_t(s_t; \hat \mu_t, \betainithat) \phi_t(\beta^\star, \hat \mu_t, o) - \hat d_t(s_t; \hat \mu_t, \betainithat) \phi_t(\beta^\star, \mu_t', o) \right|^2 dP \nonumber \\
            & ~~~ + \max_{1\leq t \leq T} \int \left|  \hat d_t(s_t; \hat \mu_t, \betainithat) \phi_t(\beta^\star, \mu_t', o) - \hat d_t(s_t; \mu_t', \beta^\star) \phi_t(\beta^\star, \mu_t', o) \right|^2 dP \nonumber \\
            & ~~~ + \max_{1\leq t \leq T} \int \left|  \hat d_t(s_t; \mu_t', \beta^\star) \phi_t(\beta^\star, \mu_t', o) -  d_t'(s_t) \phi_t(\beta^\star, \mu_t', o)\right|^2 dP \bigg\}, \nonumber
	\end{align}
	where the last inequality follows from \cref{A-lem:sum-of-squares-function}. Therefore, it suffices to show that for all $1 \leq t \leq T$,
	\begin{align}
		\int \left| \hat d_t(s_t; \hat \mu_t, \betainithat) \phi_t(\beta^\star, \hat \mu_t, o) - \hat d_t(s_t; \hat \mu_t, \betainithat) \phi_t(\beta^\star, \mu_t', o)  \right|^2 dP = o_P(1), \label{A-eq:continuous-ee-l2-conv:proofuse1} \\
            \int \left| \hat d_t(s_t; \hat \mu_t, \betainithat) \phi_t(\beta^\star, \mu_t', o) - \hat d_t(s_t; \mu_t', \beta^\star) \phi_t(\beta^\star, \mu_t', o) \right|^2 dP = o_P(1), \label{A-eq:continuous-ee-l2-conv:proofuse2} \\
		\int \left| \hat d_t(s_t; \mu_t', \beta^\star) \phi_t(\beta^\star, \mu_t', o) -  d_t'(s_t) \phi_t(\beta^\star, \mu_t', o) \right|^2 dP = o_P(1). \label{A-eq:continuous-ee-l2-conv:proofuse3}
	\end{align}

	To show \cref{A-eq:continuous-ee-l2-conv:proofuse1}, we 
    have
    \begin{align}
            & \int \left| \hat d_t(s_t; \hat \mu_t, \betainithat) \phi_t(\beta^\star, \hat \mu_t, o) - \hat d_t(s_t; \hat \mu_t, \betainithat) \phi_t(\beta^\star, \mu_t', o) \right|^2 dP \nonumber \\
            & \leq C \int \left| \phi_t(\beta^\star, \hat \mu_t, o) - \phi_t(\beta^\star, \mu_t', o) \right|^2 dP \label{A-eq:continuous-ee-l2-conv:proofuse3.5} \\
            & = o_P(1) \label{A-eq:continuous-ee-l2-conv:proofuse3.6}
        \end{align}
        where \cref{A-eq:continuous-ee-l2-conv:proofuse3.5} follows from that $\hat d_t(s_t; \hat \mu_t, \betainithat)$ is bounded (\cref{A-asu:continuous-uniform-integrable}) and that the Frobenius norm is compatible with the Euclidean norm, and \cref{A-eq:continuous-ee-l2-conv:proofuse3.6} follows from the convergence of $\phi_t(\beta^\star, \hat \mu_t)$ (\cref{A-lem:continuous-ee-phi-converge}). 
    
    To show \cref{A-eq:continuous-ee-l2-conv:proofuse2}, we have 
    \begin{align}
            & \int \left| \hat d_t(s_t; \hat \mu_t, \betainithat) \phi_t(\beta^\star, \mu_t', o) - \hat d_t(s_t; \mu_t', \beta^\star) \phi_t(\beta^\star, \mu_t', o) \right|^2 dP \nonumber \\
            & \leq C \int \left| \hat d_t(s_t; \hat \mu_t, \betainithat) - \hat d_t(s_t; \mu_t', \beta^\star) \right|_F^2 dP \label{A-eq:continuous-ee-l2-conv:proofuse4.5} \\
            & = o_P(1) \label{A-eq:continuous-ee-l2-conv:proofuse4}
        \end{align}
     where \cref{A-eq:continuous-ee-l2-conv:proofuse4.5} follows from that all the observations and $\mu_t'$ are bounded (\cref{A-asu:continuous-bounded-obs,A-asu:continuous-bounded-gamma-prime}), $A_t \in \{0, 1\}$, $p_t$ is bounded away from $0$ or $1$ (\cref{asu:positivity}) (thus $\phi_t(\beta^\star, \mu_t')$ is bounded), and that the Frobenius norm is compatible with the Euclidean norm, and \cref{A-eq:continuous-ee-l2-conv:proofuse4} follows from the convergence of $\hat d_t(s_t; \hat \mu_t, \betainithat)$ (\cref{A-lem:continuous-ee-d-converge}).  

	To show \cref{A-eq:continuous-ee-l2-conv:proofuse3}, we have
	\begin{align}
            & \int \left| \hat d_t(s_t; \mu_t', \beta^\star) \phi_t(\beta^\star, \mu_t', o) -  d_t'(s_t) \phi_t(\beta^\star, \mu_t', o) \right|^2 dP \nonumber \\
            & \leq C \int \left| \hat d_t(s_t; \mu_t', \beta^\star) -  d_t'(s_t) \right|_F^2 dP \label{A-eq:continuous-ee-l2-conv:proofuse5.5} \\
            & = o_P(1) \label{A-eq:continuous-ee-l2-conv:proofuse5}
	\end{align}
    where \cref{A-eq:continuous-ee-l2-conv:proofuse5.5} follows from that all the observations and $\mu_t'$ are bounded (\cref{A-asu:continuous-bounded-obs,A-asu:continuous-bounded-gamma-prime}), $A_t \in \{0, 1\}$, $p_t$ is bounded away from $0$ or $1$ (\cref{asu:positivity}) (thus $\phi_t(\beta^\star, \mu_t')$ is bounded), and that the Frobenius norm is compatible with the Euclidean norm, and \cref{A-eq:continuous-ee-l2-conv:proofuse5} follows from the convergence of $\hat d_t(s_t; \mu_t', \beta^\star)$ (\cref{A-asu:continuous-d-converge}).  

	Having \cref{A-eq:continuous-ee-l2-conv:proofuse1}, \cref{A-eq:continuous-ee-l2-conv:proofuse2}, and \cref{A-eq:continuous-ee-l2-conv:proofuse3} established, the proof is completed.
\end{proof}

\begin{lem}
	\label{A-lem:continuous-PPee-deriv-conv}
	Suppose \cref{A-asu:continuous-unique-zero,A-asu:continuous-bounded-obs,A-asu:continuous-compact-param-space,A-asu:continuous-bounded-gamma,A-asu:continuous-bounded-gamma-prime,A-asu:continuous-hat-tilde-d-bounded,A-asu:continuous-d-converge,A-asu:continuous-nuisance-converge,A-asu:continuous-donsker-nuisance,A-asu:continuous-continuous-d-hat,A-asu:continuous-uniform-integrable} and \cref{asu:positivity} hold. Then for $\mc(\beta, \eta)$ defined in \cref{A-eq:ee-continuous}, we have
	\begin{align*}
		\left| \PP \big\{ \partial_\beta \mc(\beta^\star, \hat\eta) \big\} - \PP \big\{ \partial_\beta \mc(\beta^\star, \eta') \big\} \right| = o_P(1).
	\end{align*}
	Under the same assumptions, the lemma statement also holds with $\hat\eta$ replaced by the cross-fitting version, $\hat\eta_k$.
\end{lem}

\begin{proof}[Proof of \cref{A-lem:continuous-PPee-deriv-conv}]
	We present the proof for the non-cross-fitting version. The proof for the cross-fitting version is identical with $\hat\eta$, $\hat\mu_t$, $\hat d_t$ replaced by $\hat\eta_k$, $\hat\mu_{kt}$, $\hat d_{kt}$, respectively. Here, we explicitly write out the dependency of $\mc$ on $\mu$, $d$, and $\betainit$. We have
	\begin{align*}
		& \left| \PP \big\{ \partial_\beta \mc(\beta^\star, \hat\eta) \big\} - \PP \big\{ \partial_\beta \mc(\beta^\star, \eta') \big\} \right| \nonumber \\
            & \leq \left| \PP \big\{ \partial_\beta \mc(\beta^\star, \hat\mu, \hat d(S_t; \betainithat, \hat \mu)) \big\} - \PP \big\{ \partial_\beta \mc(\beta^\star, \hat \mu, \hat d(S_t;\beta^\star, \mu')) \big\} \right| \nonumber \\
            & ~~~ + \left| \PP \big\{ \partial_\beta \mc(\beta^\star, \hat \mu, \hat d(S_t;\beta^\star, \mu')) \big\} - \PP \big\{ \partial_\beta \mc(\beta^\star, \mu', \hat d(S_t;\beta^\star, \mu')) \big\} \right| \nonumber \\
            & ~~~ + \left| \PP \big\{ \partial_\beta \mc(\beta^\star, \mu', \hat d(S_t; \beta^\star, \mu')) \big\} - \PP \big\{ \partial_\beta \mc(\beta^\star, \mu', d') \big\} \right| \nonumber
        \end{align*}
        Thus, it remains to show that
        \begin{align}
            \left| \PP \big\{ \partial_\beta \mc(\beta^\star, \hat\mu, \hat d(S_t; \betainithat, \hat \mu)) \big\} - \PP \big\{ \partial_\beta \mc(\beta^\star, \hat \mu, \hat d(S_t; \beta^\star, \mu')) \big\} \right|  & = o_P(1) \label{A-eq:lem:continuous-PPee-deriv-conv:proofuse1} \\
            \left| \PP \big\{ \partial_\beta \mc(\beta^\star, \hat \mu, \hat d(S_t; \beta^\star, \mu')) \big\} - \PP \big\{ \partial_\beta \mc(\beta^\star, \mu', \hat d(S_t;\beta^\star, \mu')) \big\} \right|  & = o_P(1) \label{A-eq:lem:continuous-PPee-deriv-conv:proofuse3} \\
            \left| \PP \big\{ \partial_\beta \mc(\beta^\star, \eta', \hat d(S_t; \beta^\star, \mu')) \big\} - \PP \big\{ \partial_\beta \mc(\beta^\star, \eta', d') \big\} \right| & = o_P(1) \label{A-eq:lem:continuous-PPee-deriv-conv:proofuse2} 
        \end{align}
        
        To show \cref{A-eq:lem:continuous-PPee-deriv-conv:proofuse1}, we have
        \begin{align}
            & \left| \PP \big\{ \partial_\beta \mc(\beta^\star, \hat\eta, \hat d(S_t; \betainithat, \hat \mu)) \big\} - \PP \big\{ \partial_\beta \mc(\beta^\star, \hat \eta, \hat d(S_t; \beta^\star, \mu')) \big\} \right| \nonumber \\
            & \leq \int \left| \partial_\beta \mc(\beta^\star, \hat\eta, \hat d(S_t; \betainithat, \hat \mu)) - \partial_\beta \mc(\beta^\star, \hat \eta, \hat d(S_t; \beta^\star, \mu')) \right| dP \nonumber \\
            & = \int \left| \sum_{t=1}^T \hat d_t(s_t; \hat\mu_t, \betainithat) \partial_\beta \phic_t(\beta^\star, \hat \mu_t) - \hat d_t(s_t; \mu_t', \beta^\star) \partial_\beta \phic_t(\beta^\star, \hat \mu_t) \right| dP \nonumber \\
            & \leq T \left\{\max_{1\leq t \leq T} \int \left| \hat d_t(s_t; \hat\mu_t, \betainithat) \partial_\beta \phic_t(\beta^\star, \hat \mu_t) - \hat d_t(s_t; \mu_t', \beta^\star) \partial_\beta \phic_t(\beta^\star, \hat \mu_t) \right| dP \right\}. \nonumber 
        \end{align}
        
        It suffices to show that for each $1\leq t \leq T$, 
        \begin{align*}
            \int \left| \hat d_t(s_t; \hat\mu_t, \betainithat) \partial_\beta \phic_t(\beta^\star, \hat \mu_t) - \hat d_t(s_t; \mu_t', \beta^\star) \partial_\beta \phic_t(\beta^\star, \hat \mu_t)  \right| dP = o_P(1)
        \end{align*}
        Because $A_t \in \{0, 1\}$, that all the observations and $\hat \mu_t$ are bounded (\cref{A-asu:continuous-bounded-obs,A-asu:continuous-bounded-gamma}), that $p_t$ is bounded away from $0$ or $1$ (\cref{asu:positivity}), and that the parameter space $\Theta$ is compact (\cref{A-asu:continuous-compact-param-space}) (thus $\partial_\beta \phic_t(\beta^\star, \hat \mu_t)$ is bounded). Because the Frobenius norm is compatible with the Euclidean norm, we have
        \begin{align}
            & \int \left| \hat d_t(s_t; \hat\mu_t, \betainithat) \partial_\beta \phic_t(\beta^\star, \hat \mu_t) - \hat d_t(s_t; \mu_t', \beta^\star) \partial_\beta \phic_t(\beta^\star, \hat \mu_t)  \right| dP \nonumber \\
            & \leq C \left(\int |\hat d_t(s_t; \hat\mu_t, \betainithat) - \hat d_t(s_t; \mu_t', \beta^\star) |_F^2 dP\right) \nonumber \\
            & \leq C \left(\int |\hat d_t(s_t; \hat\mu_t, \betainithat) - \hat d_t(s_t; \mu_t', \beta^\star) |_F^2 dP\right)^{\frac{1}{2}} \label{A-eq:lem:continuous-PPee-deriv-conv:proofuse3.4} \\
            & = o_P(1)\label{A-eq:lem:continuous-PPee-deriv-conv:proofuse4}
        \end{align}
        where \cref{A-eq:lem:continuous-PPee-deriv-conv:proofuse3.4} follows from the Cauchy-Schwarz inequality, and \cref{A-eq:lem:continuous-PPee-deriv-conv:proofuse4} follows from the convergence of $\hat d_t(s_t;\hat \mu_t,\betainithat)$ (\cref{A-lem:continuous-ee-d-converge}).  

    To show \cref{A-eq:lem:continuous-PPee-deriv-conv:proofuse3}, we have
    \begin{align}
        & \left| \PP \big\{ \partial_\beta \mc(\beta^\star, \hat \eta, \hat d(S_t; \beta^\star, \mu')) \big\} - \PP \big\{ \partial_\beta \mc(\beta^\star, \eta', \hat d(S_t; \beta^\star, \mu')) \big\} \right| \nonumber \\
        & \leq \int \left| \partial_\beta \mc(\beta^\star, \hat \eta, \hat d(S_t; \beta^\star, \mu')) - \partial_\beta \mc(\beta^\star, \eta', \hat d(S_t; \beta^\star, \mu')) \right| dP \nonumber \\
        & = \int \left| \sum_{t=1}^T \hat d_t(s_t; \mu_t', \beta^\star) \partial_\beta \phic_t(\beta^\star, \hat \mu_t) - \hat d_t(s_t; \mu_t', \beta^\star) \partial_\beta \phic_t(\beta^\star, \mu_t') \right| dP \nonumber \\
        & \leq T \left\{\max_{1\leq t \leq T} \int \left| \hat d_t(s_t; \mu_t', \beta^\star) \partial_\beta \phic_t(\beta^\star, \hat \mu_t) - \hat d_t(s_t; \mu_t', \beta^\star) \partial_\beta \phic_t(\beta^\star, \mu_t') \right| dP \right\}. \nonumber 
    \end{align}

    It suffices to show that for each $1\leq t \leq T$, 
        \begin{align*}
            \int \left| \hat d_t(s_t; \mu_t', \beta^\star) \partial_\beta \phic_t(\beta^\star, \hat \mu_t) - \hat d_t(s_t; \mu_t', \beta^\star) \partial_\beta \phic_t(\beta^\star, \mu_t')  \right| dP = o_P(1)
        \end{align*}
        Because $\hat d_t(s_t; \mu_t', \beta^\star)$ is bounded (\cref{A-asu:continuous-hat-tilde-d-bounded}), and then we have 
        \begin{align}
            & \int \left| \hat d_t(s_t; \mu_t', \beta^\star) \partial_\beta \phic_t(\beta^\star, \hat \mu_t) - \hat d_t(s_t; \mu_t', \beta^\star) \partial_\beta \phic_t(\beta^\star, \mu_t') \right| dP \nonumber \\
            & \leq C \left(\int |\partial_\beta \phic_t(\beta^\star, \hat \mu_t) - \partial_\beta \phic_t(\beta^\star, \mu_t') |^2 dP\right) \nonumber \\
            & \leq C \left(\int |\partial_\beta \phic_t(\beta^\star, \hat \mu_t) - \partial_\beta \phic_t(\beta^\star, \mu_t') |^2 dP\right)^{\frac{1}{2}} \label{A-eq:lem:continuous-PPee-deriv-conv:proofuse5} \\
            & = o_P(1)\label{A-eq:lem:continuous-PPee-deriv-conv:proofuse6}
        \end{align}
        where \cref{A-eq:lem:continuous-PPee-deriv-conv:proofuse5} follows from the Cauchy-Schwarz inequality, and \cref{A-eq:lem:continuous-PPee-deriv-conv:proofuse6} follows from the convergence of $\phic_t(\beta^\star, \hat \mu_t)$ (\cref{A-lem:continuous-ee-phi-converge}).  

       To show \cref{A-eq:lem:continuous-PPee-deriv-conv:proofuse2}, we have
        \begin{align}
            & \left| \PP \big\{ \partial_\beta \mc(\beta^\star, \eta', \hat d(S_t; \beta^\star, \mu')) \big\} -  \PP \big\{ \partial_\beta \mc(\beta^\star, \eta', d') \big\} \right| \nonumber \\
            & = \int \left| \partial_\beta \mc(\beta^\star, \eta', \hat d(S_t; \beta^\star, \mu')) - \partial_\beta \mc(\beta^\star, \eta', d') \right| dP \nonumber \\
            & = \int \left| \sum_{t=1}^T \hat d_t(s_t; \mu_t', \beta^\star) \partial_\beta \phic_t(\beta^\star, \mu_t') - d_t'(s_t) \partial_\beta \phic_t(\beta^\star, \mu_t') \right| dP \nonumber \\
            & \leq T \left\{\max_{1\leq t \leq T} \int \left| \hat d_t(s_t; \mu_t', \beta^\star) \partial_\beta \phic_t(\beta^\star, \mu_t') - d_t'(s_t) \partial_\beta \phic_t(\beta^\star, \mu_t') \right| dP \right\} \nonumber 
        \end{align}
        It suffices to show that for each $1\leq t \leq T$, 
        \begin{align*}
            \int \left| \hat d_t(s_t; \mu_t', \beta^\star) \partial_\beta \phic_t(\beta^\star, \mu_t') - d_t'(s_t) \partial_\beta \phic_t(\beta^\star, \mu_t') \right| dP = o_P(1)
        \end{align*}
        Because $A_t \in \{0, 1\}$, that all the observations and $\mu_t'$ are bounded (\cref{A-asu:continuous-bounded-obs,A-asu:continuous-bounded-gamma-prime}), that $p_t$ is bounded away from $0$ or $1$ (\cref{asu:positivity}), and that the parameter space $\Theta$ is compact (\cref{A-asu:continuous-compact-param-space}) (thus $\partial_\beta \phic_t(\beta^\star, \mu_t')$ is bounded). Because the Frobenius norm is compatible with the Euclidean norm, we have
        \begin{align}
            & \int \left| \hat d_t(s_t; \mu_t', \beta^\star) \partial_\beta \phic_t(\beta^\star, \mu_t') - d_t'(s_t) \partial_\beta \phic_t(\beta^\star, \mu_t') \right| dP \nonumber \\
            & \leq C \int \left| \hat d_t(s_t; \mu_t', \beta^\star) - d_t'(s_t) \right|_F dP \nonumber \\
            & \leq C \left( \int \left| \hat d_t(s_t; \mu_t', \beta^\star) - d_t'(s_t) \right|_F^2 dP \right)^{\frac{1}{2}} \label{A-eq:lem:continuous-PPee-deriv-conv:proofuse7} \\
            & = o_P(1) \label{A-eq:lem:continuous-PPee-deriv-conv:proofuse8}
        \end{align}
        where \cref{A-eq:lem:continuous-PPee-deriv-conv:proofuse7} follows from the Cauchy-Schwarz inequality, and \cref{A-eq:lem:continuous-PPee-deriv-conv:proofuse8} follows from the convergence of $\hat d_t(s_t;\mu_t', \beta^\star)$ (\cref{A-asu:continuous-d-converge}). This completes the proof.
\end{proof}

\begin{lem}
	\label{A-lem:continuous-PPee-meat-conv}
	Suppose \cref{A-asu:continuous-unique-zero,A-asu:continuous-bounded-obs,A-asu:continuous-compact-param-space,A-asu:continuous-bounded-gamma,A-asu:continuous-bounded-gamma-prime,A-asu:continuous-hat-tilde-d-bounded,A-asu:continuous-d-converge,A-asu:continuous-nuisance-converge,A-asu:continuous-donsker-nuisance,A-asu:continuous-continuous-d-hat,A-asu:continuous-uniform-integrable} and \cref{asu:positivity} hold. Then for $\mc(\beta, \eta)$ defined in \cref{A-eq:ee-continuous}, we have
	\begin{align*}
		\left| \PP \big\{ \mc(\beta^\star, \hat\eta) \mc(\beta^\star, \hat\eta)^T \big\} - \PP \big\{ \mc(\beta^\star, \eta') \mc(\beta^\star, \eta')^T \big\} \right| = o_P(1).
	\end{align*}
	Under the same assumptions, the lemma statement also holds with $\hat\eta$ replaced by the cross-fitting version, $\hat\eta_k$.
\end{lem}

\begin{proof}[Proof of \cref{A-lem:continuous-PPee-meat-conv}]
	The proof is almost identical to the proof of \cref{A-lem:continuous-PPee-deriv-conv} and is thus omitted.
\end{proof}

\subsection{Asymptotic Normality}
\label{A-subsec:normality-continuous}

\begin{thm}[Asymptotic normality for continuous outcome case]
	\label{A-thm:continuous-normality}
	Consider $\mc(\beta,\eta)$ defined in \cref{A-eq:ee-continuous}. Suppose \cref{A-asu:continuous-unique-zero,A-asu:continuous-regularity,A-asu:continuous-d-converge,A-asu:continuous-nuisance-converge,A-asu:continuous-donsker-nuisance} and \cref{asu:positivity} hold, then we have
	\begin{align*}
		\sqrt{n}(\hat\beta - \beta^\star) \dto N ( 0, V) \quad \text{as } n\to\infty,
	\end{align*}
	where
	\begin{align} 
		V = {} & \PP\Big\{\sum_{t=1}^T d_t'(S_t) \partial_\beta \phi_t(\beta^\star, \mu_t') \Big\}^{-1} \PP \Big[ \Big\{ \sum_{t=1}^T d_t'(S_t) \phi_t(\beta^\star, \mu_t') \Big\}^{\otimes 2} \Big] \nonumber \\
& \times \PP\Big\{\sum_{t=1}^T d_t'(S_t) \partial_\beta \phi_t(\beta^\star, \mu_t') \Big\}^{-1, T}. \label{A-eq:continuous-avar}
	\end{align}	
	In addition,
	\begin{align*}
		& ~~ \bigg[\PP_n\Big\{\sum_{t=1}^T \hat{d}_t(S_t; \hat\mu_t, \betainithat) \partial_\beta \phi_t(\hat\beta, \hat\mu_t) \Big\}\bigg]^{-1} \PP_n \Big[ \Big\{ \sum_{t=1}^T \hat{d}_t(S_t; \hat\mu_t, \betainithat) \phi_t(\hat\beta, \hat\mu_t) \Big\}^{\otimes 2} \Big] \\
		& \times \bigg[\PP_n\Big\{\sum_{t=1}^T \hat{d}_t(S_t; \hat\mu_t, \betainithat) \partial_\beta \phi_t(\hat\beta, \hat\mu_t) \Big\}\bigg]^{-1, T}.
	\end{align*}
	is a consistent estimator for the asymptotic variance $V$.	
\end{thm}

\begin{proof}[Proof of \cref{A-thm:continuous-normality}]
	By \cref{A-thm:general-normality}, it suffices to verify \cref{A-asu:unique-zero,A-asu:nuisance-conv-general,A-asu:donsker,A-asu:reg-general}.

	Verify \cref{A-asu:unique-zero}: This follows from \cref{A-asu:continuous-unique-zero} and \cref{A-lem:continuous-beta0-is-zero}.

	Verify \cref{A-asu:reg-general}: \cref{A-asu:reg-compact-param-space} is assumed in \cref{A-asu:continuous-compact-param-space}. \cref{A-asu:reg-bounded-obs} is assumed in \cref{A-asu:continuous-bounded-obs}. \cref{A-asu:reg-cont-PPee} holds as $\mc(\beta, \eta)$ is linear in $\beta$. \cref{A-asu:reg-bounded-and-cont-differentiable-ee,A-asu:reg-dominated-ee-deriv,A-asu:reg-dominated-ee-meat} hold because of the form of $\mc(\beta,\eta)$ and \cref{A-asu:continuous-bounded-obs}. \cref{A-asu:reg-invertible-ee-deriv} is assumed in \cref{A-asu:continuous-invertible-deriv}.
	
	Verify \cref{A-asu:nuisance-conv-PPee-sup}: This follows from \cref{A-lem:continuous-ee-sup-conv}.
	
	Verify \cref{A-asu:nuisance-conv-ee-l2}: This follows from \cref{A-lem:continuous-ee-l2-conv}.

	Verify \cref{A-asu:nuisance-conv-PPee-deriv}: This follows from \cref{A-lem:continuous-PPee-deriv-conv}.

	Verify \cref{A-asu:nuisance-conv-PPee-meat}: This follows from \cref{A-lem:continuous-PPee-meat-conv}.
	
	Verify \cref{A-asu:donsker}: This is because $\mu_t(\cdot)$ and $d_t(\cdot)$ are from Donsker classes (\cref{A-asu:continuous-donsker-nuisance}), $\mc(\beta,\eta)$ is a polynomial of $\mu_t$ and $d_t$ and thus Lipschitz, and a Lipschitz transformation of Donsker classes is still Donsker \citep[][Theorem 9.31]{kosorok2008introduction}.

	Thus, we have verified all assumptions in \cref{A-thm:general-normality}. This completes the proof.
\end{proof}

\begin{thm}[Asymptotic normality for continuous outcome case (cross-fitting)]
	\label{A-thm:continuous-normality-cf}
	Consider $\mc(\beta,\eta)$ defined in \cref{A-eq:ee-continuous}. Suppose \cref{A-asu:continuous-unique-zero,A-asu:continuous-nuisance-converge,A-asu:continuous-d-converge,A-asu:continuous-regularity} and \cref{asu:positivity} hold, then we have
	\begin{align*}
		\sqrt{n}(\check\beta - \beta^\star) \dto N ( 0, V) \quad \text{as } n\to\infty,
	\end{align*}
	where
	\begin{align*}
		V = \PP\{\partial_\beta \mc(\beta^\star, \eta')\}^{-1} ~ \PP\{ \mc(\beta^\star, \eta') \mc(\beta^\star, \eta')^T \} ~ \PP\{\partial_\beta \mc(\beta^\star, \eta')\}^{-1, T}.
	\end{align*}	
	In addition,
	\begin{align*}
		& \bigg[\frac{1}{K}\sum_{k=1}^K \PP_{n,k} \{ \partial_\beta \mc(\check\beta, \hat\eta_k) \}\bigg]^{-1}
		~ \bigg[\frac{1}{K} \sum_{k=1}^K \PP_{n,k} \{ \mc(\check\beta, \hat\eta_k) \mc(\check\beta, \hat\eta_k)^T \}\bigg] \\
		\times & \bigg[\frac{1}{K}\sum_{k=1}^K \PP_{n,k} \{ \partial_\beta \mc(\check\beta, \hat\eta_k) \}\bigg]^{-1, T}
	\end{align*}
	is a consistent estimator for the asymptotic variance $V$.	
\end{thm}

\begin{proof}[Proof of \cref{A-thm:continuous-normality-cf}]
	By \cref{A-thm:general-normality-cf-formal}, it suffices to verify \cref{A-asu:unique-zero,A-asu:nuisance-conv-general-cf,A-asu:reg-general,A-asu:reg-bounded-ee-deriv-and-meat}.

	Verify \cref{A-asu:unique-zero}: This follows from \cref{A-asu:continuous-unique-zero} and \cref{A-lem:continuous-beta0-is-zero}.

	Verify \cref{A-asu:reg-general}: \cref{A-asu:reg-compact-param-space} is assumed in \cref{A-asu:continuous-compact-param-space}. \cref{A-asu:reg-bounded-obs} is assumed in \cref{A-asu:continuous-bounded-obs}. \cref{A-asu:reg-cont-PPee} holds as $\mc(\beta, \eta)$ is linear in $\beta$. \cref{A-asu:reg-bounded-and-cont-differentiable-ee,A-asu:reg-dominated-ee-deriv,A-asu:reg-dominated-ee-meat} hold because of the form of $\mc(\beta,\eta)$ and \cref{A-asu:continuous-bounded-obs}. \cref{A-asu:reg-invertible-ee-deriv} is assumed in \cref{A-asu:continuous-invertible-deriv}.
	
	Verify \cref{A-asu:nuisance-conv-PPee-sup-cf}: This follows from \cref{A-lem:continuous-ee-sup-conv}.
	
	Verify \cref{A-asu:nuisance-conv-ee-l2-cf}: This follows from \cref{A-lem:continuous-ee-l2-conv}.

	Verify \cref{A-asu:nuisance-conv-PPee-deriv-cf}: This follows from \cref{A-lem:continuous-PPee-deriv-conv}.

	Verify \cref{A-asu:nuisance-conv-PPee-meat-cf}: This follows from \cref{A-lem:continuous-PPee-meat-conv}.
	
	Verify \cref{A-asu:reg-bounded-ee-deriv-and-meat}: This holds because $\mc(\beta, \eta)$ is linear in $\beta$ and because of \cref{A-asu:continuous-bounded-obs,A-asu:continuous-compact-param-space}. 

	Thus, we have verified all assumptions in \cref{A-thm:general-normality-cf-formal}. This completes the proof.
\end{proof}

\section{Double Robustness in Observational Settings}
\label{A-sec:dr-observational}

This section proves the double-robustness claim in \cref{rmk:double-robustness-observational}: when $p_t$ is unknown and replaced in the estimating function by a working/estimated $p_t$, the population estimating equation $\PP\{\phi_t(\beta^\star,p_t,\mu_t)\}=0$ holds whenever \emph{either} the propensity model or the outcome-regression model is correctly specified, for both the identity and log links.

\paragraph{Setup.} Let $p_t^\star(H_t)=P(A_t=1\mid H_t)$ and $\mu_t^\star(H_t,a)=\EE(Y_{t+1}\mid H_t,A_t=a,I_t=1)$ for $a\in\{0,1\}$ denote the truth, and let $p_t(H_t)$ and $\mu_t(H_t,a)$ denote the working models that the analyst plugs into the estimating function. The structural CEE model (\cref{eq:linear-model-cee}) is the marginalized identity $\EE\{\mu_t^\star(H_t,1)-\mu_t^\star(H_t,0)\mid S_t,I_t=1\}=\gamma_t(S_t;\beta^\star)$ for the identity link, and $\log\EE\{\mu_t^\star(H_t,1)\mid S_t,I_t=1\}-\log\EE\{\mu_t^\star(H_t,0)\mid S_t,I_t=1\}=\gamma_t(S_t;\beta^\star)$ for the log link; equivalently, $\EE\{\mu_t^\star(H_t,1)\mid S_t,I_t=1\}=e^{\gamma_t}\,\EE\{\mu_t^\star(H_t,0)\mid S_t,I_t=1\}$. We emphasize that the structural identity is at the $(S_t,I_t=1)$ level and need not hold pointwise in $H_t$. With these working models, the estimating function $\phi_t(\beta,p_t,\mu_t)$ is given by
\begin{align}
\phi_t \;=\; \frac{A_t-p_t}{p_t(1-p_t)}\,I_t\,\Big\{Y_{t+1}-(A_t+p_t-1)\gamma_t-(1-p_t)\mu_t(H_t,1)-p_t\mu_t(H_t,0)\Big\}\,\partial_\beta\gamma_t \label{A-eq:dr-phi-id}
\end{align}
for the identity link and
\begin{align}
\phi_t \;=\; \frac{A_t-p_t}{p_t(1-p_t)}\,I_t\,\Big\{e^{-A_t\gamma_t}Y_{t+1}-(1-p_t)e^{-\gamma_t}\mu_t(H_t,1)-p_t\mu_t(H_t,0)\Big\}\,\partial_\beta\gamma_t \label{A-eq:dr-phi-log}
\end{align}
for the log link.

\begin{prop}[Double robustness]
\label{A-prop:dr-observational}
Suppose the structural CEE model holds at $\beta=\beta^\star$. Then for both the identity and log link,
\begin{align}
\PP\{\phi_t(\beta^\star,p_t,\mu_t)\} \;=\; \PP\!\left\{I_t\,\partial_\beta\gamma_t\,\cdot\,\frac{p_t^\star-p_t}{p_t(1-p_t)}\,\cdot\,\Lambda_t(H_t)\right\}, \label{A-eq:dr-factorization}
\end{align}
where the outcome gap $\Lambda_t$ is
\begin{align}
\Lambda_t(H_t)
\;=\; \begin{cases}
(1-p_t)\{\mu_t^\star(H_t,1)-\mu_t(H_t,1)\}+p_t\{\mu_t^\star(H_t,0)-\mu_t(H_t,0)\} & \text{(identity link)},\\[2pt]
(1-p_t)\,e^{-\gamma_t}\{\mu_t^\star(H_t,1)-\mu_t(H_t,1)\}+p_t\{\mu_t^\star(H_t,0)-\mu_t(H_t,0)\} & \text{(log link)}.
\end{cases} \label{A-eq:dr-outcome-gap}
\end{align}
Consequently, $\PP\{\phi_t(\beta^\star,p_t,\mu_t)\}=0$ for each $t$ for which either $p_t(H_t)=p_t^\star(H_t)$ or $\mu_t(H_t,a)=\mu_t^\star(H_t,a)$ for $a\in\{0,1\}$. Hence, if for each $t\in[T]$ at least one of the two nuisance models is correctly specified (the choice of which model is correct may differ across $t$), then $\PP\{\sum_{t=1}^T \phi_t(\beta^\star,p_t,\mu_t)\}=0$. The same conclusion holds for the weighted estimating function $\sum_{t=1}^T d_t(S_t;\mu_t,\beta^\star)\,\phi_t(\beta^\star,p_t,\mu_t)$ for any $S_t$-measurable weight $d_t$.
\end{prop}

\begin{proof}
We drop $I_t$ and $\partial_\beta\gamma_t$ throughout the conditional computation and reinstate them at the end; both are $H_t$-measurable and pass through the conditional expectation in $A_t$.

\emph{Identity link.}
At $\beta=\beta^\star$, write $\phi_t = \{(A_t-p_t)/[p_t(1-p_t)]\}\,R_{A_t}$, where
\begin{align*}
R_a \;=\; \mu_t^\star(H_t,a)-(a+p_t-1)\gamma_t-(1-p_t)\mu_t(H_t,1)-p_t\mu_t(H_t,0),
\end{align*}
using $\EE(Y_{t+1}\mid H_t,A_t=a,I_t=1)=\mu_t^\star(H_t,a)$. Direct algebra gives
\begin{align*}
R_1-R_0 \;=\; \mu_t^\star(H_t,1)-\mu_t^\star(H_t,0)-\gamma_t, \qquad
(1-p_t)R_1+p_t R_0 \;=\; \Lambda_t(H_t),
\end{align*}
with $\Lambda_t$ as in \cref{A-eq:dr-outcome-gap}. Marginalizing over $A_t\mid H_t\sim\mathrm{Bernoulli}(p_t^\star)$,
\begin{align}
\EE(\phi_t\mid H_t)
&\;=\; \frac{p_t^\star(1-p_t)}{p_t(1-p_t)}\,R_1 - \frac{(1-p_t^\star)p_t}{p_t(1-p_t)}\,R_0 \nonumber\\
&\;=\; (R_1-R_0) + \frac{p_t^\star-p_t}{p_t(1-p_t)}\bigl[(1-p_t)R_1+p_t R_0\bigr] \nonumber\\
&\;=\; \bigl[\mu_t^\star(H_t,1)-\mu_t^\star(H_t,0)-\gamma_t\bigr] + \frac{p_t^\star-p_t}{p_t(1-p_t)}\,\Lambda_t(H_t). \label{A-eq:dr-Ht-decomp-id}
\end{align}
Reinstate $I_t\,\partial_\beta\gamma_t$ and take the outer expectation. The first summand on the right of \cref{A-eq:dr-Ht-decomp-id} integrates to zero by the marginal CEE model: $\partial_\beta\gamma_t$ is $S_t$-measurable, so iterated expectation conditional on $(S_t,I_t=1)$ gives
\begin{align*}
& \PP\bigl\{I_t\,\partial_\beta\gamma_t\,\bigl[\mu_t^\star(H_t,1)-\mu_t^\star(H_t,0)-\gamma_t\bigr]\bigr\} \\
& \quad = \PP\Bigl\{I_t\,\partial_\beta\gamma_t\,\EE\bigl[\mu_t^\star(H_t,1)-\mu_t^\star(H_t,0)-\gamma_t\,\big|\,S_t,I_t=1\bigr]\Bigr\} \\
& \quad = 0.
\end{align*}
Hence \cref{A-eq:dr-factorization} holds for the identity link.

\emph{Log link.}
At $\beta=\beta^\star$, write $\phi_t = \{(A_t-p_t)/[p_t(1-p_t)]\}\,\tilde R_{A_t}$, where
\begin{align*}
\tilde R_a \;=\; e^{-a\gamma_t}\mu_t^\star(H_t,a)-(1-p_t)e^{-\gamma_t}\mu_t(H_t,1)-p_t\mu_t(H_t,0),
\end{align*}
using $\EE(Y_{t+1}\mid H_t,A_t=a,I_t=1)=\mu_t^\star(H_t,a)$. Direct algebra gives $\tilde R_1-\tilde R_0 = e^{-\gamma_t}\mu_t^\star(H_t,1)-\mu_t^\star(H_t,0)$ and $(1-p_t)\tilde R_1+p_t\tilde R_0 = \Lambda_t(H_t)$ (log-link case of \cref{A-eq:dr-outcome-gap}). The same Bernoulli marginalization as in the identity-link case yields
\begin{align}
\EE(\phi_t\mid H_t)
\;=\; \bigl[e^{-\gamma_t}\mu_t^\star(H_t,1)-\mu_t^\star(H_t,0)\bigr] + \frac{p_t^\star-p_t}{p_t(1-p_t)}\,\Lambda_t(H_t). \label{A-eq:dr-Ht-decomp-log}
\end{align}
Reinstate $I_t\,\partial_\beta\gamma_t$ and take the outer expectation. The first summand integrates to zero: $e^{-\gamma_t}$ and $\partial_\beta\gamma_t$ are $S_t$-measurable, so by iterated expectation conditional on $(S_t,I_t=1)$,
\begin{align*}
&\PP\bigl\{I_t\,\partial_\beta\gamma_t\,\bigl[e^{-\gamma_t}\mu_t^\star(H_t,1)-\mu_t^\star(H_t,0)\bigr]\bigr\}\\
&\quad=\;\PP\Bigl\{I_t\,\partial_\beta\gamma_t\,\bigl(e^{-\gamma_t}\,\EE[\mu_t^\star(H_t,1)\mid S_t,I_t=1]-\EE[\mu_t^\star(H_t,0)\mid S_t,I_t=1]\bigr)\Bigr\}\;=\;0,
\end{align*}
where the last equality uses the log-link structural identity $\EE[\mu_t^\star(H_t,1)\mid S_t,I_t=1]=e^{\gamma_t}\,\EE[\mu_t^\star(H_t,0)\mid S_t,I_t=1]$. Hence \cref{A-eq:dr-factorization} holds for the log link.

\emph{Conclusion.}
\Cref{A-eq:dr-factorization} shows that $\PP\{\phi_t(\beta^\star,p_t,\mu_t)\}=0$ whenever $p_t=p_t^\star$ (the propensity gap vanishes at $t$) or $\mu_t=\mu_t^\star$ ($\Lambda_t$ vanishes at $t$), and the choice of which factor vanishes may differ across $t$. Summing over $t\in[T]$ gives $\PP\{\sum_t \phi_t(\beta^\star,p_t,\mu_t)\}=0$ provided one of the two nuisances is correct at every $t$. The weighted version follows from the same argument with $\phi_t$ replaced by $d_t\,\phi_t$: since $d_t(S_t;\mu_t,\beta^\star)$ is $S_t$-measurable, the conditional structural residual in \cref{A-eq:dr-Ht-decomp-id} (or \cref{A-eq:dr-Ht-decomp-log}) still integrates to zero against $d_t\,I_t\,\partial_\beta\gamma_t$ by iterated expectation over $(S_t,I_t=1)$, and the DR factor still vanishes at the same $t$'s.
\end{proof}

\section{Implementation}
\label{A-sec:implementation}

Here we provide additional details on the implementation of the non-cross-fitted estimator $\hat\beta$ (\cref{algo:estimator-ncf}) and the cross-fitted estimator $\check\beta$ (\cref{A-algo:estimator-cf}).

\subsection{Practical Implementation of the Proposed Algorithms}
\label{A-subsec:practical-implementation}

This subsection documents the implementation used in the simulations (\cref{sec:simulation}) and the data application (\cref{sec:application}). The corresponding code is available at \url{https://github.com/tqian/paper_CEE-efficient-ML}.

\paragraph{Data layout.}
The data are stored in long format as a single $nT$-row data frame with one row per $(i,t)$ pair. Required columns are: a participant identifier $i$, decision-point index $t$, treatment $A_t$, randomization probability $p_t = P(A_t = 1 \mid H_t)$, availability indicator $I_t$, the proximal outcome $Y_{t+1}$, the moderators $S_t$, and a working summary of the history $H_t$ (see Step 1 below). All Step 1 and Step 3 fits in our implementation operate on this single long-format data frame by pooling across $t$, although per-$t$ fitting is also a valid option (see the implementation note following \cref{algo:estimator-ncf}).

\paragraph{Step 1: fitting $\hat\mu_t(H_t, A_t)$.}
The history $H_t$ is summarized through a fixed-dimensional working covariate vector $X_t^\dagger$ that includes (i) baseline characteristics (e.g., age, AUDIT score), (ii) the decision-point index $t$, and (iii) current time-varying covariates and recent lags of the outcome, the treatment, and covariates (e.g., $Y_t$, $A_{t-1}$). $\hat\mu_t$ is fit by pooling across $t \in [T]$: a single regression of $Y_{t+1}$ on $X_t^\dagger$ is estimated using the full long-format data. To recover $\hat\mu_t(H_t, 1)$ and $\hat\mu_t(H_t, 0)$, we choose to fit two such pooled regressions, one on the rows with $A_t = 1$ and one on the rows with $A_t = 0$. One can alternatively fit a single model to obtain $\hat\mu_t(H_t, A_t)$, and subsequently obtain $\hat\mu_t(H_t, 1)$ and $\hat\mu_t(H_t, 0)$. The specific fitting procedures we use are:
\begin{itemize}
    \item \emph{GAM:} \texttt{mgcv::gam} with $Y_{t+1} \sim s(X_{t,1}^\dagger) + \cdots + s(X_{t,p}^\dagger)$, REML smoothing-parameter selection, and the family appropriate for the outcome (\texttt{gaussian} for continuous, \texttt{binomial} with logit link for binary, \texttt{poisson} with log link for count).
    \item \emph{Random forest:} \texttt{ranger::ranger} with default settings (500 trees, $\lfloor p/3 \rfloor$ candidate features per split for regression, minimum node size 5; for binary outcomes \texttt{probability = TRUE}).
    \item {\raggedright \emph{Super learner:} \texttt{SuperLearner::}\allowbreak\texttt{SuperLearner} with the library \texttt{c("SL.mean",} \texttt{"SL.glm",} \texttt{"SL.gam",} \texttt{"SL.earth",} \texttt{"SL.ranger",} \texttt{"SL.xgboost",} \texttt{"SL.nnet")} for continuous outcome, and \texttt{c("SL.mean",} \texttt{"SL.glm",} \texttt{"SL.gam",} \texttt{"SL.earth",} \texttt{"SL.ranger")} for binary outcome (XGBoost and neural network excluded due to numerical instability at the small sample sizes used in the binary simulation). Internal 10-fold cross-validation is used by the package to choose the convex combination of learners.\par}
    \item \emph{Stacking (count outcome only):} a weighted average of GLM, GAM, multivariate spline, and random forest, where weights are obtained by Poisson-regressing $Y_{t+1}$ on the predictions from the individual learners.
\end{itemize}
As noted in the implementation note following \cref{algo:estimator-ncf}, $\hat\mu_t$ in Step~1 may be obtained either by fitting separately for each $t$ or by pooling across $t$. Fitting a single pooled model with $t$ included as a covariate is mathematically equivalent to fitting models indexed by $t$ from a function class that nests $t$-additive structure, and is the implementation we use for both statistical and computational reasons: it borrows strength across decision points, accommodates small per-$t$ sample sizes, and keeps the per-fit cost growing in $nT$ rather than scaling with $T$ separately fit models.

\paragraph{Step 2: initial $\hat\beta^{\mathrm{init}}$.}
With $\hat\mu_t$ from Step 1 plugged into $\phi_t(\beta, \hat\mu_t)$ from \cref{eq:phi-definition}, $\hat\beta^{\mathrm{init}}$ is obtained by solving the unweighted estimating equation $\PP_n \sum_{t=1}^T \phi_t(\beta, \hat\mu_t) = 0$. We use \texttt{rootSolve::multiroot} with default tolerance and starting value $\beta = \mathbf{0}$.

\paragraph{Step 3: fitting $\hat d_t(S_t; \hat\mu_t, \hat\beta^{\mathrm{init}})$.}
The two conditional expectations in $d_t^\star(S_t)$ in \cref{eq:dt-definition} are estimated by pooling across $t$. Specifically, we fit one regression of $\partial_\beta \phi_t(\hat\beta^{\mathrm{init}}, \hat\mu_t)$ on $(t, S_t)$ and one regression of $\phi_t(\hat\beta^{\mathrm{init}}, \hat\mu_t)^{\otimes 2}$ on $(t, S_t)$, both on the long-format data. In the simulations, the regression is a $t$-stratified empirical mean (so the conditional expectations are estimated by per-$t$ averages); in the data application, both regressions are pooled-across-$t$ multivariate splines (\texttt{earth} or \texttt{mgcv::gam} with a tensor-product smooth in $(t, S_t)$). The remark below shows that when $\partial_\beta \gamma_t(S_t; \beta) = S_t^T$, the matrix-valued $d_t$ reduces to a scalar $\tilde d_t$ that does not require inverting an $\dim(\beta) \times \dim(\beta)$ matrix per decision point; we use this scalar form throughout to simplify implementation and improve numerical stability. The remark also gives the closed-form expression for $\partial_\beta \phi_t(\beta, \mu_t)$ used in evaluating $\hat d_t$.

\paragraph{Step 4: solving for $\hat\beta$.}
$\hat\beta$ is obtained by solving $\PP_n \sum_{t=1}^T \hat d_t(S_t) \phi_t(\beta, \hat\mu_t) = 0$, again via \texttt{rootSolve::multiroot} with default tolerance and starting value $\hat\beta^{\mathrm{init}}$.

\paragraph{Cross-fitting (\cref{A-algo:estimator-cf}).}
We use $K = 5$ folds throughout, with the partition obtained by a random shuffle of participant identifiers and a fixed seed. For each fold $k$, $\hat\mu_{kt}$ and $\hat d_{kt}$ are fit on $B_k^c$ following Steps 1--3 above, and the contributions $\PP_{n,k}\{\hat d_{kt}(S_t) \phi_t(\beta, \hat\mu_{kt})\}$ are evaluated on $B_k$. The cross-fitted estimating equation pools the fold-specific contributions as in Step~3 of \cref{A-algo:estimator-cf}.

\paragraph{Variance estimation.}
The asymptotic variance is estimated by the sandwich form in \cref{thm:normality-ncf} of the main paper (for $\hat\beta$) or \cref{A-thm:normality-cf} of this appendix (for $\check\beta$), with the small-sample correction described in \cref{A-sec:small-sample-correction}.

\paragraph{Wall-clock timings on the Drink Less data.}
\Cref{A-tab:wallclock} reports wall-clock times for the marginal continuous-outcome analysis on the Drink Less data ($n = 349$, $T = 30$). The bulk of the wall time in the full super-learner ensemble (61~s) is consumed by the XGBoost and neural-network learners; replacing the SL library with the smooth-only subset (\texttt{SL.mean}, \texttt{SL.glm}, \texttt{SL.gam}, \texttt{SL.earth}) reduces the wall time by roughly an order of magnitude (4.9~s) with limited efficiency loss in our experiments.
\begin{table}[h]
\centering
\begin{tabular}{lr}
\toprule
Method (5-fold cross-fitting except where noted) & Wall clock (s) \\
\midrule
WCLS (baseline)                                                & 0.20 \\
GAM (no cross-fitting)                                         & 0.38 \\
GAM                                                            & 0.58 \\
Random forest (\texttt{ranger})                                & 3.74 \\
Super learner, smooth library (no XGBoost/NN)                  & 4.93 \\
Super learner, full library (with XGBoost + NN)                & 61.21 \\
\bottomrule
\end{tabular}
\caption{Wall-clock time, in seconds, for one fit of the marginal CEE for the continuous outcome on the Drink Less data ($n = 349$, $T = 30$), measured on a 2024 M4 MacBook Pro using each package's default threading settings. Cross-fitting folds and super-learner library learners were run sequentially in R. All methods other than WCLS use \cref{algo:estimator-ncf} of the main paper (without cross-fitting) or \cref{A-algo:estimator-cf} of this appendix (with cross-fitting), with $\hat\mu_t$ pooled across decision points.}
\label{A-tab:wallclock}
\end{table}

\begin{rmk}[Analytical forms for $\partial_\beta \phi_t$ and the scalar $\tilde d_t$]
	\label{A-rmk:analytical-forms}
	\normalfont
	This remark collects two analytical results referenced in Step 3 above (\cref{A-subsec:practical-implementation}).
	$\partial_\beta \phi_t(\beta, \mu_t)$ has the following analytic form, which can be used in Step 3 of \cref{algo:estimator-ncf}:
	\begin{align*}
		& \partial_\beta \phi_t(\beta, \mu_t) \\
& = \begin{cases}
			-\frac{A_t - p_t}{p_t(1-p_t)} I_t (A_t + p_t - 1) S_t \partial_\beta \gamma_t(S_t;\beta) & \text{if $g$ is identity}, \\
			\frac{A_t - p_t}{p_t(1-p_t)} I_t \{-A_t e^{-A_t \gamma_t(S_t;\beta)}Y_{t+1} & \\
\quad + (1-p_t)e^{-\gamma_t(S_t;\beta)}\mu_t(H_t, 1)\} S_t \partial_\beta \gamma_t(S_t;\beta) & \text{if $g$ is log}.
		\end{cases} \nonumber
	\end{align*}
 We define matrix inverse as the Moore-Penrose generalized inverse. We could rewrite $\phi_t(\beta)$ defined in \cref{eq:phi-definition} as
 \begin{align*}
     \phi_t(\beta) = R_t(\beta) S_t
 \end{align*}
 where $R_t(\beta)$ is defined as
 \begin{align*}
	R_t(\beta) := \frac{A_t - p_t}{p_t(1-p_t)} I_t \Big[U_t(\beta) - (1-p_t)\EE\{ U_t(\beta) \mid H_t, A_t = 1 \} - p_t\EE\{ U_t(\beta) \mid H_t, A_t = 0 \}\Big],
 \end{align*}
 which under the identity and log links takes the form
 \begin{align*}
	& R_t(\beta, \mu_t) \\
& = \begin{cases}
		\frac{A_t - p_t}{p_t(1-p_t)} I_t \{Y_{t+1} - (A_t + p_t - 1)\gamma_t(S_t;\beta) & \\
\quad - (1-p_t)\mu_t(H_t, 1) - p_t \mu_t(H_t, 0) \} & \text{if $g$ is identity}, \\
		\frac{A_t - p_t}{p_t(1-p_t)} I_t \{e^{-A_t \gamma_t(S_t;\beta)}Y_{t+1} & \\
\quad - (1-p_t)e^{-\gamma_t(S_t;\beta)}\mu_t(H_t, 1) - p_t \mu_t(H_t, 0) \} & \text{if $g$ is log}.
	\end{cases} \nonumber
\end{align*}
 We show in the next subsection that when $\partial_\beta \gamma_t(S_t;\beta) = S_t^T$, $d_t(S_t;\mu_t,\beta)$ could be simplified to
 \begin{align*}
 \tilde{d}_t(S_t; \mu_t, \beta) = 
	 		\EE\bigg\{\partial_\beta R_t(\beta, \mu_t) \mid S_t \bigg\} ~ \EE\bigg\{ R_t(\beta, \mu_t)^2 \mid S_t \bigg\}^{-1}.
\end{align*}
Note, this reduced form of $d_t$ does not involve $S_t S_t^T$, which could simplify the implementation of $\hat{d}_t$ and allow more flexible modeling (e.g., $S_t$ could be high-dimensional) in $\hat{d}_t$. Additionally, $\tilde{d}_t(S_t; \mu_t, \beta)$ has the following analytical form 
\begin{align*}
		& \tilde{d}_t(S_t; \mu_t, \beta) \\
& = \begin{cases}
			 \frac{-1}{\EE\bigg\{ \frac{(A_t - p_t)^2}{p_t^2(1-p_t)^2} I_t^2 (Y_{t+1} - (A_t + p_t - 1)\gamma_t(S_t;\beta) - (1-p_t)\mu_t(H_t, 1) - p_t \mu_t(H_t, 0))^2 \mid S_t \bigg\}} & \text{if $g$ is identity}, \\
			\frac{\EE\bigg\{ -A_t e^{-A_t \gamma_t(S_t;\beta)}Y_{t+1} + (1-p_t)e^{-\gamma_t(S_t;\beta)}\mu_t(H_t, 1) \mid S_t \bigg\}}{\EE\bigg\{ \frac{(A_t - p_t)^2}{p_t^2(1-p_t)^2} I_t^2 (e^{-A_t \gamma_t(S_t;\beta)}Y_{t+1} - (1-p_t)e^{-\gamma_t(S_t;\beta)}\mu_t(H_t, 1) - p_t \mu_t(H_t, 0))^2 \mid S_t \bigg\}} & \text{if $g$ is log}.
		\end{cases} \qedhere
\end{align*}

\end{rmk}

\subsection{Technical Derivation}

Consider the following estimating function for $\beta$:
\begin{align}
    \sum_{t=1}^T d_t(S_t;\mu_t,\betainit)\phi_t(\beta,\mu_t). \label{A-eq:implementation-matrix-form}
\end{align}
where
\begin{align*}
	& \phi_t(\beta, \mu_t) \\
& = \begin{cases}
		\frac{A_t - p_t}{p_t(1-p_t)} I_t \{Y_{t+1} - (A_t + p_t - 1)\gamma_t(S_t;\beta) & \\
\quad - (1-p_t)\mu_t(H_t, 1) - p_t \mu_t(H_t, 0) \} S_t & \text{if $g$ is identity}, \\
		\frac{A_t - p_t}{p_t(1-p_t)} I_t \{e^{-A_t \gamma_t(S_t;\beta)}Y_{t+1} & \\
\quad - (1-p_t)e^{-\gamma_t(S_t;\beta)}\mu_t(H_t, 1) - p_t \mu_t(H_t, 0) \} S_t & \text{if $g$ is log}.
	\end{cases} \nonumber
\end{align*}
and
\begin{align*}
    d_t(S_t;\mu_t,\beta) = \EE\{\partial_\beta \phi_t(\beta, \mu_t) \mid S_t \} ~ [\EE\{ \phi_t(\beta, \mu_t)^{\otimes 2} \mid S_t \}]^{+}.
\end{align*}
Here we use $^+$ to denote the Moore-Penrose generalized inverse.

\begin{asu}[Effect Linear in $\beta$]
    \label{A-asu:effect-linear-in-beta}
    Suppose for each $1 \leq t \leq T$, $\gamma_t(S_t;\beta)$ is linear in $\beta$.
\end{asu}

\begin{lem}
    \label{A-lem:d-scalar-implementation}
    Suppose \cref{A-asu:effect-linear-in-beta} holds, then the estimating function defined in \cref{A-eq:implementation-matrix-form} is equivalent to the following estimating function
    \begin{align}
        \sum_{t=1}^T \td_t(S_t;\mu_t,\betainit)\phi_t(\beta,\mu_t)
    \end{align}
    where
\begin{align*}
		& \tilde{d}_t(S_t; \mu_t, \beta) := \\
            & \begin{cases}
            \EE\bigg\{ -\frac{A_t - p_t}{p_t(1-p_t)} I_t (A_t + p_t - 1) \mid S_t \bigg\} \times \\
			 \EE\bigg\{ \bigg[\frac{(A_t - p_t)}{p_t(1-p_t)} I_t (Y_{t+1} - (A_t + p_t - 1)\gamma_t(S_t;\beta) \\
\quad - (1-p_t)\mu_t(H_t, 1) - p_t \mu_t(H_t, 0))\bigg]^{\otimes2} \mid S_t \bigg\}^+ \\
            & \text{if $g$ is identity}, \\
		\EE\bigg\{ -A_t e^{-A_t \gamma_t(S_t;\beta)}Y_{t+1} + (1-p_t)e^{-\gamma_t(S_t;\beta)}\mu_t(H_t, 1) \mid S_t \bigg\} \times \\
            ~~~ \EE\bigg\{ \bigg[\frac{(A_t - p_t)}{p_t(1-p_t)} I_t (e^{-A_t \gamma_t(S_t;\beta)}Y_{t+1} \\
\quad - (1-p_t)e^{-\gamma_t(S_t;\beta)}\mu_t(H_t, 1) - p_t \mu_t(H_t, 0))\bigg]^{\otimes2} \mid S_t \bigg\}^+ \\
            & \text{if $g$ is log},
		\end{cases} \nonumber
\end{align*} 
\end{lem}

\begin{proof}
    Let
    \begin{align*}
        W_t & = \frac{A_t-p_t}{p_t(1-p_t)} I_t
    \end{align*}
    and 
     \begin{align*}
	R_t(\beta, \mu_t) = \begin{cases}
		Y_{t+1} - (A_t + p_t - 1)\gamma_t(S_t;\beta) - (1-p_t)\mu_t(H_t, 1) - p_t \mu_t(H_t, 0) & \text{if $g$ is identity}, \\
		e^{-A_t \gamma_t(S_t;\beta)}Y_{t+1} - (1-p_t)e^{-\gamma_t(S_t;\beta)}\mu_t(H_t, 1) - p_t \mu_t(H_t, 0) & \text{if $g$ is log}.
	\end{cases} \nonumber
\end{align*}
    \cref{A-eq:implementation-matrix-form} implies that, for any $\beta \in \Theta$, we have
    \begin{align}
    & \sum_{t=1}^T d_t(S_t;\mu_t,\betainit)\phi_t(\beta,\mu_t) \nonumber \\
    & = \sum_{t=1}^T d_t(S_t;\mu_t,\betainit) W_t R_t \partial_{\beta} \gamma_t(S_t;\beta) \nonumber \\
        & = \sum_{t=1}^T \td_t(S_t;\mu_t,\beta) ~\partial_{\beta} \gamma_t(S_t;\beta) \partial_{\beta} \gamma_t(S_t;\beta)^T ~ \left( \partial_{\beta} \gamma_t(S_t;\beta) \partial_{\beta} \gamma_t(S_t;\beta)^T \right)^+ ~ W_t R_t \partial_{\beta} \gamma_t(S_t;\beta) \label{A-eq:lem:d-scalar-implementation-proofuse1} \\
        & = \sum_{t=1}^T \td_t(S_t;\mu_t,\beta) ~ \partial_{\beta} \gamma_t(S_t;\beta) \partial_{\beta} \gamma_t(S_t;\beta)^T ~ \left( \partial_{\beta} \gamma_t(S_t;\beta)^T \right)^+ \partial_{\beta} \gamma_t(S_t;\beta)^+ ~ W_t R_t \partial_{\beta} \gamma_t(S_t;\beta) \label{A-eq:lem:d-scalar-implementation-proofuse2} \\
        & = \sum_{t=1}^T \td_t(S_t;\mu_t,\beta) ~\partial_{\beta} \gamma_t(S_t;\beta) \partial_{\beta} \gamma_t(S_t;\beta)^T ~ \left( \partial_{\beta} \gamma_t(S_t;\beta)^T \right)^+ \partial_{\beta} \gamma_t(S_t;\beta)^+ ~ W_t R_t \partial_{\beta} \gamma_t(S_t;\beta) \label{A-eq:lem:d-scalar-implementation-proofuse3} \\
        & = \sum_{t=1}^T \td_t(S_t;\mu_t,\beta) ~\partial_{\beta} \gamma_t(S_t;\beta) \partial_{\beta} \gamma_t(S_t;\beta)^+ ~ W_t R_t \partial_{\beta} \gamma_t(S_t;\beta) \nonumber \\
        & = \sum_{t=1}^T \td_t(S_t;\mu_t,\beta) W_t R_t \partial_{\beta} \gamma_t(S_t;\beta) \label{A-eq:lem:d-scalar-implementation-proofuse4} 
    \end{align}
    where \cref{A-eq:lem:d-scalar-implementation-proofuse1} follows from the specific form of $\phi_t(\beta, \mu_t)$ and that $\gamma_t(S_t;\beta)$ is linear in $\beta$ (\cref{A-asu:effect-linear-in-beta}); \cref{A-eq:lem:d-scalar-implementation-proofuse2} follows from the property of Moore-Penrose inverse; \cref{A-eq:lem:d-scalar-implementation-proofuse3} follows from that $\gamma_t(S_t;\beta)$ is linear in $\beta$ (\cref{A-asu:effect-linear-in-beta}) and $\beta$ is a $p$-dimensional vector that $\left\{\partial_{\beta} \gamma_t(S_t;\beta)\right\}^T$ is a row vector and thus has full row rank; \cref{A-eq:lem:d-scalar-implementation-proofuse4} follows from the property of Moore-Penrose inverse. This completes the proof.
\end{proof}

\section{Bias-Corrected Variance Estimator for Globally Robust Estimating Equations}
\label{A-sec:small-sample-correction}

\subsection{Setup}

Consider a class of globally robust estimating functions that can be written as
\begin{align}
	m(\beta, \eta) = D(\eta)^T r(\beta, \eta). \label{A-eq:decompose-ee}
\end{align}

The proposed estimating equations for $\hat\beta$ can be written in this form. For notation simplicity, we define $\mu_{t1} := \mu_t(H_t, 1)$, $\mu_{t0} := \mu_t(H_t, 0)$, $p_t = p_t(H_t)$, $\tp_t = \tp_t(S_t)$, $f_t = f_t(S_t)$, $d_t = d_t(S_t)$, $W_t = W_t\{\tp_t(S_t)\}$. For identity link, the estimating function is
\begin{align*}
	\mc(\beta, \eta) & = \sum_{t=1}^T d_t I_t W_t \{ Y_{t+1} - \mu_{t1}(1 - p_t) - \mu_{t0} p_t - (A_t + p_t - 1) f_t^T\beta\} (A_t - \tp_t) f_t \\
	& = \sum_{t=1}^T D_t(\eta)^T r_t(\beta, \eta),
\end{align*}
where $D_t(\eta) = d_t I_t W_t (A_t - \tp_t) f_t^T$ and $r_t(\beta, \eta) = Y_{t+1} - \mu_{t1}(1 - p_t) - \mu_{t0} p_t - (A_t + p_t - 1) f_t^T\beta$. For log link, the estimating function is
\begin{align*}
	\mb(\beta, \eta) & = \sum_{t=1}^T d_t I_t W_t \{ e^{-A_t f_t^T\beta} Y_{t+1} - e^{-f_t^T\beta}\mu_{t1}(1 - p_t) - \mu_{t0} p_t\} (A_t - \tp_t) f_t \\
	& = \sum_{t=1}^T D_t(\eta)^T r_t(\beta, \eta),
\end{align*}
where $D_t(\eta) = d_t I_t W_t (A_t - \tp_t) f_t^T$ and $r_t(\beta, \eta) = e^{-A_t f_t^T\beta} Y_{t+1} - e^{-f_t^T\beta}\mu_{t1}(1 - p_t) - \mu_{t0} p_t$. For both cases, let $D(\eta) = \{D_t(\eta)\}_c$ and $r(\eta) = \{r_t(\beta, \eta)\}_c$, and we have \cref{A-eq:decompose-ee} hold. (Subscript $c$ means column stacking.)

In the following, we derive the bias correction for the variance estimator for this class of estimating equations. To make the derivation rigorous, we make the following assumptions. Note that these assumptions are only for the purpose of deriving the bias correction, and the correction formula will be applied even when these assumptions do not hold.

\begin{asu}
	\label{A-asu:ssc-exogeneous}
	Assume that $I_t$ is exogenous, the randomization probability $p_t(H_t)$ is a constant, and $\tp_t(S_t)$ is set to a constant, and $A_t$ is independent of $\{Y_{s+1}: t \leq s \leq T\}$.
\end{asu}

Even though the last part of \cref{A-asu:ssc-exogeneous} is quite restrictive, it does hold under the setting where $A_t$ has no effect on all future proximal outcomes, which is a plausible scenario under the null. Therefore, the small sample correction we derived is particularly useful for improving type I error rate control in small sample sizes.

\subsection{Derivation of the small sample correction with fixed \texorpdfstring{$\eta$}{eta}}

For the globally robust estimating equation $\mathbb{P}_{n}D(\eta)^{T}r(\beta,\eta)=0$, we have established in the paper that
\begin{align}
	\hat{\beta}-\beta^\star = -\left\{ \PP\partial_{\beta}m(\beta^\star,\eta)\right\} ^{-1}\PP_n m(\beta^\star,\eta)+o_P(n^{-1/2}). \label{A-eq:ssc-proofuse-1}
\end{align}
This implies that the asymptotic variance of $\hat{\beta}$ is
\begin{align}
	\left[\PP\partial_{\beta}m(\beta^\star,\eta)\right]^{-1}\PP\left[m(\beta^\star,\eta)m(\beta^\star,\eta)^T \right]\left[\PP\partial_{\beta}m(\beta^\star,\eta)\right]^{-1,T}. \label{A-eq:ssc-proofuse-2}
\end{align}
and can be estimated by
\begin{align}
	\left[\PP_n \partial_{\beta}m(\hat{\beta},\eta)\right]^{-1}\PP_n \left[m(\hat{\beta},\eta)m(\hat{\beta},\eta)^T \right]\left[\PP_n \partial_{\beta}m(\hat{\beta},\eta)\right]^{-1,T}. \label{A-eq:ssc-proofuse-3}
\end{align}

Define $M_{\beta^\star, \eta} := \PP\partial_{\beta}m(\beta^\star, \eta)$ and
\begin{align}
	H_{ij}(\beta, \eta): = \frac{1}{n}\frac{\partial r_i (\beta, \eta)}{\partial\beta^T }M_{\beta, \eta}^{-1}D_j^T. \label{A-eq:ssc-proofuse-4}
\end{align}
Then \cref{A-eq:ssc-proofuse-1} implies that
\begin{align}
	\hat{\beta}-\beta^\star = -M_{\beta^\star}^{-1}\PP_n m(\beta^\star)+o_P(n^{-1/2}). \label{A-eq:ssc-proofuse-5} 
\end{align}

The middle term in the asymptotic variance \cref{A-eq:ssc-proofuse-2} is
\begin{align*}
	\PP\left[m(\beta^\star,\eta)m(\beta^\star,\eta)^T \right] = \PP\left[D(\eta)^T r(\beta^\star,\eta)r(\beta^\star,\eta)^T D(\eta)\right].
\end{align*}
In estimating the asymptotic variance (\cref{A-eq:ssc-proofuse-3}), we used $\PP_n \left[D(\eta)^T r(\hat{\beta},\eta)r(\hat{\beta},\eta)^T D(\eta)\right]$ to approximate $\PP\left[D(\eta)^T r(\beta^\star,\eta)r(\beta^\star,\eta)^T D(\eta)\right]$. Following the rationale in \citet{mancl2001}, we derive the small sample correction by replacing $r(\hat{\beta},\eta)$ in \cref{A-eq:ssc-proofuse-3} with a better approximation of $r(\beta^\star,\eta)$. 

We make the following additional working assumption to facilitate a closed-form small sample correction formula.
\begin{asu}[Conditional independence]
    \label{A-asu:zopt-reg-independence}
    For any $1 \leq u < t \leq T$, \[ \{Y_{t+1}, A_t, \gamma_t(H_t), p_t(H_t)\} \perp \{Y_{u+1}, A_u, \gamma_u(H_u), p_u(H_u), I_u\} \mid S_t, I_t = 1. \]
\end{asu}
Because of \cref{A-asu:zopt-reg-independence}, we have
\begin{align}
	& ~~~~ \PP\left[D(\eta)^T r(\beta^\star,\eta)r(\beta^\star,\eta)^T D(\eta)\right] \nonumber \\
	& = \PP\left[ \PP \{ D(\eta)^T r(\beta^\star,\eta)r(\beta^\star,\eta)^T D(\eta) \mid I_1, \ldots, I_T, A_1, \ldots, A_T\} \right] \nonumber \\
	& = \PP \left[ D(\eta)^T \PP \{ r(\beta^\star,\eta)r(\beta^\star,\eta)^T \mid I_1, \ldots, I_T, A_1, \ldots, A_T\} D(\eta) \right].
\end{align}
From now on in this subsection, we will use $\EE$ to denote expectations conditional on $\{I_t, A_t: 1 \leq t \leq T\}$, and we omit writing this conditional set for notation simplicity. We will derive the small sample correction that provides a better approximation of $\EE\{ r(\beta^\star,\eta)r(\beta^\star,\eta)^T \}$. For notation simplicity, we will also ignore writing $\eta$ in this subsection because we will always consider a fixed $\eta$. 

By the Mean Value Theorem we have
\begin{align*}
	r_i(\hat\beta) = r_i(\beta^\star) + \frac{\partial r_i (\tilde{\beta})}{\partial\beta^T }(\hat\beta - \beta^\star)\qquad \text{for some } \tilde\beta \in (\beta^\star, \hat\beta).
\end{align*}
Therefore,
\begin{align}
	& ~~~~ \EE \{r_i (\hat{\beta})r_i (\hat{\beta})^T \} \nonumber \\
	&  = \EE \bigg\{ r_i (\beta^\star)r_i (\beta^\star)^T  \bigg\}+\EE \bigg[ r_i (\beta^\star)(\hat{\beta}-\beta^\star)^T \bigg\{ \frac{\partial r_i (\tilde{\beta})}{\partial\beta^T } \bigg\}^T  \bigg] \nonumber \\
	& + \EE \bigg\{ \frac{\partial r_i (\tilde{\beta})}{\partial\beta^T }(\hat{\beta}-\beta^\star)r_i (\beta^\star)^T  \bigg\} + \EE \bigg[  \frac{\partial r_i (\tilde{\beta})}{\partial\beta^T }(\hat{\beta}-\beta^\star)(\hat{\beta}-\beta^\star)^T \bigg\{ \frac{\partial r_i (\tilde{\beta})}{\partial\beta^T } \bigg\}^T  \bigg] \nonumber \\
	&  = B_1 + B_2 + B_3 + B_4 \label{A-eq:ssc-proofuse-5.1}
\end{align}

By \cref{A-eq:ssc-proofuse-5}, we have
\begin{align}
	B_{2} & = \EE \bigg[ r_i (\beta^\star)(\hat{\beta}-\beta^\star)^T \bigg\{ \frac{\partial r_i (\tilde{\beta})}{\partial\beta^T }\bigg\}^T  \bigg]\nonumber \\
	&  = \EE \bigg[ r_i (\beta^\star)\bigg\{ -M_{\beta^\star}^{-1}\PP_n m(\beta^\star)+o_P(n^{-1/2})\bigg\}^T \bigg\{ \frac{\partial r_i (\beta^\star)}{\partial\beta^T }+o_P(1)\bigg\}^T  \bigg] \nonumber \\
	&  = \frac{1}{n}\EE \bigg[ r_i (\beta^\star)\bigg\{ -M_{\beta^\star}^{-1}m_i (\beta^\star)\bigg\}^T \bigg\{ \frac{\partial r_i (\beta^\star)}{\partial\beta^T }\bigg\}^T  \bigg]+o_P(1) \label{A-eq:ssc-proofuse-6} \\
	&  = -\EE \bigg[ r_i (\beta^\star)r_i (\beta^\star)^T \frac{1}{n} D_i M_{\beta^\star}^{-1,T}\bigg\{ \frac{\partial r_i (\beta^\star)}{\partial\beta^T }\bigg\}^T  \bigg]+o_P(1) \nonumber \\
	&  = -\EE \{ r_i (\beta^\star)r_i (\beta^\star)^T H_{ii}(\beta^\star)^T \} + o_P(1). \label{A-eq:ssc-proofuse-7}
\end{align}
Here \cref{A-eq:ssc-proofuse-6} is because $\PP\{D_{j}^T r_{j}(\beta^\star)\} = 0$ and the independence between data from participant $i$ and $j$. 

Similarly, we have
\begin{align}
	B_3 &  = \EE \bigg[ \frac{\partial r_i (\tilde{\beta})}{\partial\beta^T }(\hat{\beta}-\beta^\star)r_i (\beta^\star)^T  \bigg]\\
 	&  = -\EE \{ H_{ii}(\beta^\star)r_i (\beta^\star)r_i (\beta^\star)^T \} + o_P(1). \label{A-eq:ssc-proofuse-8}
\end{align}

Lastly, we have
\begin{align}
	B_4 &  = \EE \bigg[ \frac{\partial r_i (\tilde{\beta})}{\partial\beta^T }(\hat{\beta}-\beta^\star)(\hat{\beta}-\beta^\star)^T \bigg\{ \frac{\partial r_i (\tilde{\beta})}{\partial\beta^T }\bigg\}^T  \bigg] \nonumber \\
	 & = \EE \bigg[ \bigg\{ \frac{\partial r_i (\beta^\star)}{\partial\beta^T }+o_P(1)\bigg\} \bigg\{ -M_{\beta^\star}^{-1}\PP_n m(\beta^\star)+o_P(n^{-1/2})\bigg\} \nonumber \\
	 & ~~ \times \bigg\{ -M_{\beta^\star}^{-1}\PP_n m(\beta^\star)+o_P(n^{-1/2})\bigg\}^T \bigg\{ \frac{\partial r_i (\beta^\star)}{\partial\beta^T }+o_P(1)\bigg\}^T  \bigg] \nonumber \\
	 & = \EE \bigg[ \frac{\partial r_i (\beta^\star)}{\partial\beta^T }M_{\beta^\star}^{-1}\PP_n m(\beta^\star)\PP_n m(\beta^\star)^T M_{\beta^\star}^{-1,T}\bigg\{ \frac{\partial r_i (\beta^\star)}{\partial\beta^T }\bigg\}^T  \bigg]+o_P(1) \nonumber \\
	 & = \frac{1}{n^{2}}\EE \bigg[ \frac{\partial r_i (\beta^\star)}{\partial\beta^T }M_{\beta^\star}^{-1}\sum_{j,k = 1}^n m_j (\beta^\star)m_k(\beta^\star)^T M_{\beta^\star}^{-1,T}\bigg\{ \frac{\partial r_i (\beta^\star)}{\partial\beta^T }\bigg\}^T  \bigg]+o_P(1) \nonumber \\
	 & = \frac{1}{n^{2}}\EE \bigg[ \frac{\partial r_i (\beta^\star)}{\partial\beta^T }M_{\beta^\star}^{-1}\sum_{j = 1}^n m_j(\beta^\star)m_j(\beta^\star)^T M_{\beta^\star}^{-1,T}\bigg\{ \frac{\partial r_i (\beta^\star)}{\partial\beta^T }\bigg\}^T  \bigg]+o_P(1) \label{A-eq:ssc-proofuse-9} \\
	 & = \frac{1}{n^{2}}\sum_{j = 1}^n \EE \bigg[ \frac{\partial r_i (\beta^\star)}{\partial\beta^T }M_{\beta^\star}^{-1}D_{j}^T r_{j}(\beta^\star)r_{j}(\beta^\star)^T D_{j}M_{\beta^\star}^{-1,T}\bigg\{ \frac{\partial r_i (\beta^\star)}{\partial\beta^T }\bigg\}^T  \bigg]+o_P(1) \nonumber \\
	 & = \sum_{j = 1}^n \EE \{ H_{ij}(\beta^\star)r_{j}(\beta^\star)r_{j}(\beta^\star)^T H_{ij}(\beta^\star)^T \} + o_P(1). \label{A-eq:ssc-proofuse-10}
\end{align}
Here, \cref{A-eq:ssc-proofuse-9} is because $\PP\{D_{j}^T r_{j}(\beta^\star)\} = 0$ and the independence between data from participant $i$ and $j$.  

Plugging \cref{A-eq:ssc-proofuse-7,A-eq:ssc-proofuse-8,A-eq:ssc-proofuse-10} into \cref{A-eq:ssc-proofuse-5.1} and we have
\begin{align}
	\EE \{ r_i (\hat{\beta})r_i (\hat{\beta})^T  \} & = \EE \{ r_i (\beta^\star)r_i (\beta^\star)^T  \} -\EE \{ r_i (\beta^\star)r_i (\beta^\star)^T H_{ii}(\beta^\star)^T  \} \nonumber\\
	& ~~ - \EE \{ H_{ii}(\beta^\star)r_i (\beta^\star)r_i (\beta^\star)^T  \} \nonumber \\
& ~~ + \sum_{j = 1}^n \EE \{ H_{ij}(\beta^\star)r_{j}(\beta^\star)r_{j}(\beta^\star)^T H_{ij}(\beta^\star)^T  \}+o_P(1) \nonumber \\
	& = \EE [ \{ \mathrm{Ind}_i -H_{ii}(\beta^\star) \} r_i (\beta^\star)r_i (\beta^\star)^T \{ \mathrm{Ind}_i -H_{ii}(\beta^\star)\}^T  ] \nonumber \\
	& ~~ + \sum_{j\in[n]\setminus\{i\}}\EE \{ H_{ij}(\beta^\star)r_{j}(\beta^\star)r_{j}(\beta^\star)^T H_{ij}(\beta^\star)^T \} + o_P(1), \label{A-eq:ssc-proofuse-11}
\end{align}
where $\mathrm{Ind}_i$ is an identity matrix of the same dimension as $H_{ii}$.
To derive a tractable small sample correction formula, we assume that the contribution to the bias of the sum $\sum_{j\in[n]\setminus\{i\}}$ in \cref{A-eq:ssc-proofuse-11} is negligible. (A similar assumption was made in \citet{mancl2001}.) Then \cref{A-eq:ssc-proofuse-11} implies
\begin{align}
	\EE \{ r_i (\hat{\beta})r_i (\hat{\beta})^T \} & \approx \EE [ \{ \mathrm{Ind}_i -H_{ii}(\beta^\star) \} r_i (\beta^\star)r_i (\beta^\star)^T \{ \mathrm{Ind}_i -H_{ii}(\beta^\star) \}^T ] \nonumber \\
	& = \{ \mathrm{Ind}_i -H_{ii}(\beta^\star) \} \EE \{ r_i (\beta^\star)r_i (\beta^\star)^T \} \{ \mathrm{Ind}_i -H_{ii}(\beta^\star) \}^T. \label{A-eq:ssc-proofuse-12}
\end{align}
Here, \cref{A-eq:ssc-proofuse-12} follows because $\EE$ denotes the conditional expectation given $\{I_t, A_t: 1 \leq t \leq T\}$.

\cref{A-eq:ssc-proofuse-12} implies that 
\begin{align*}
	\EE \{r_i (\beta^\star)r_i (\beta^\star)^T \}\approx\{ \mathrm{Ind}_i -H_{ii}(\beta^\star)\} ^{-1}\EE [r_i (\hat{\beta})r_i (\hat{\beta})^T ]\{ \mathrm{Ind}_i -H_{ii}(\beta^\star)\} ^{-1,T}.
\end{align*}
This suggests that the meat term in the sandwich asymptotic variance $\PP[m(\beta^\star,\eta)m(\beta^\star,\eta)^T ]$ should be estimated using
\begin{align*}
	\PP_n [D_i^T \{ \mathrm{Ind}_i -H_{ii}(\beta^\star)\} ^{-1}r_i (\hat{\beta})r_i (\hat{\beta})^T \{ \mathrm{Ind}_i -H_{ii}(\beta^\star)\} ^{-1,T}D_i ]. \label{A-eq:ssc-proofuse-13}
\end{align*}

Therefore, the bias corrected estimator for the asymptotic variance of $\hat\beta$ in \cref{A-eq:ssc-proofuse-2} is
\begin{align*}
	& \bigg[\PP_n \partial_{\beta}m(\hat{\beta},\eta)\bigg]^{-1} \ \PP_n \bigg[D_i^T \{ \mathrm{Ind}_i -H_{ii}(\beta^\star)\} ^{-1}r_i (\hat{\beta})r_i (\hat{\beta})^T \{ \mathrm{Ind}_i -H_{ii}(\beta^\star)\} ^{-1,T}D_i \bigg] \\
& \times \bigg[\PP_n \partial_{\beta}m(\hat{\beta},\eta)\bigg]^{-1,T}.
\end{align*}

\section{Simulations for Binary and Count Outcomes}
\label{A-sec:appendix-additional-sims}

The simulation studies in \cref{sec:simulation} of the main text focus on the continuous outcome. Here we report parallel simulations for binary and count outcomes; the qualitative conclusions match those for the continuous outcome.

\subsection{Binary Outcome}
\label{A-subsec:appendix-binary}

The generation of covariate $Z_t$ and treatment $A_t$ for the binary outcome simulation is the same as the continuous outcome generative model in \cref{subsec:simulation-continuous}. The binary proximal outcome $Y_{t,1}$ is generated from $\text{Bernoulli}(\mu_t(H_t, A_t))$, where $\mu_t(H_t, A_t) = \exp\{A_t (\beta_0 + \beta_1 Z_t)\} \times \mu_t(H_t, 0)$ and $\mu_t(H_t, 0)$ can take one of four forms: Log-linear, where $\log\mu_t(H_t, 0) = \alpha_0 + \alpha_1 t/T + \alpha_2 (Z_t/6 + 1/2) + \rho Y_{t-1, 1} + \alpha_3 (t-1)/T$; simple nonlinear, where $\log\mu_t(H_t, 0) = \alpha_0 + 2 \times (1-\lambda) +  2/3 \times \lambda\{ q_{2,2}(Z_t/6 + 1/2) + q_{2,2}(t/T) + \rho Y_{t-1, 1}\} + \alpha_3 (t-1)/T$ with $q_{2,2}$ being the probability density function of $\text{Beta}(2,2)$ distribution; periodic, where $\log \mu_t(H_t, 0) = \alpha_0 + 2 \times (1-\lambda) +  1/2 \times \lambda \{\sin(t / 5) + \sin(Z_t) + 2\} + \rho Y_{t-1, 1} + \alpha_3 (t-1)/T$; step function, where $\log\mu_t(H_t, 0) = \alpha_0 + 2 \times (1-\lambda) + \lambda \{ \indic(\lfloor t/5 \rfloor \text{ is even}) + \indic(\lfloor 2Z_t \rfloor \text{ is even}) \} + \rho Y_{t-1, 1} + \alpha_3 (t-1)/T$ with $\lfloor \cdot \rfloor$ being the floor function.  We set $\beta_0 = 0.225$, $\beta_1 = 0.025$, $\alpha_0 = -2.5$, $\alpha_1 = \alpha_2 = 1$, and $\alpha_3 = 0.05$. The tuning parameters $\lambda$ and $\rho$ capture different magnitudes of nonlinearity in $\mu_t(H_t, 0)$ and within-subject correlation in the outcome, respectively.  The particular form of $\mu_t(H_t, 0)$ and the parameter values are chosen so that $\mu_t(H_t, A_t) \in (0,1)$.

We aim to estimate the marginal CEE with $S_t = \emptyset$ and $g$ being the log link, with true value equal to $\beta_0$.
We set $\lambda \in [0,1]$, $\rho = 0.1$. The choice of control variables and the implementations of $\hat\beta$ and $\check\beta$ are the same as \cref{subsec:simulation-continuous}, except that Step 3 of \cref{algo:estimator-ncf} uses a pooled-across-time multivariate spline to approximate the two expectations.

The simulation results are presented in \cref{A-fig:appendix-binary-control_pattern}. The conclusions are qualitatively similar to the continuous outcome simulation (\cref{subsec:simulation-continuous}). All estimators are consistent and robust against misspecified Step 1 working model; all estimators have good confidence interval coverage except for non-cross-fitting RF and SL; all the proposed estimators improve efficiency over EMEE and the efficiency gain increases as the nonlinearity parameter $\lambda$ increases, with the oracle being the most efficient; and using cross-fitting does not affect the efficiency of the proposed estimator. The efficiency gain by the proposed estimator here is smaller in magnitude compared to the continuous outcome simulation, possibly due to the limited variability and nonlinearity under the constraint $\mu_t(H_t, A_t) \in (0,1)$. Nonetheless, the proposed estimator still provides meaningful efficiency gain with a relative efficiency up to 1.16.

\begin{figure}
        \centering
        \includegraphics[width = \textwidth]{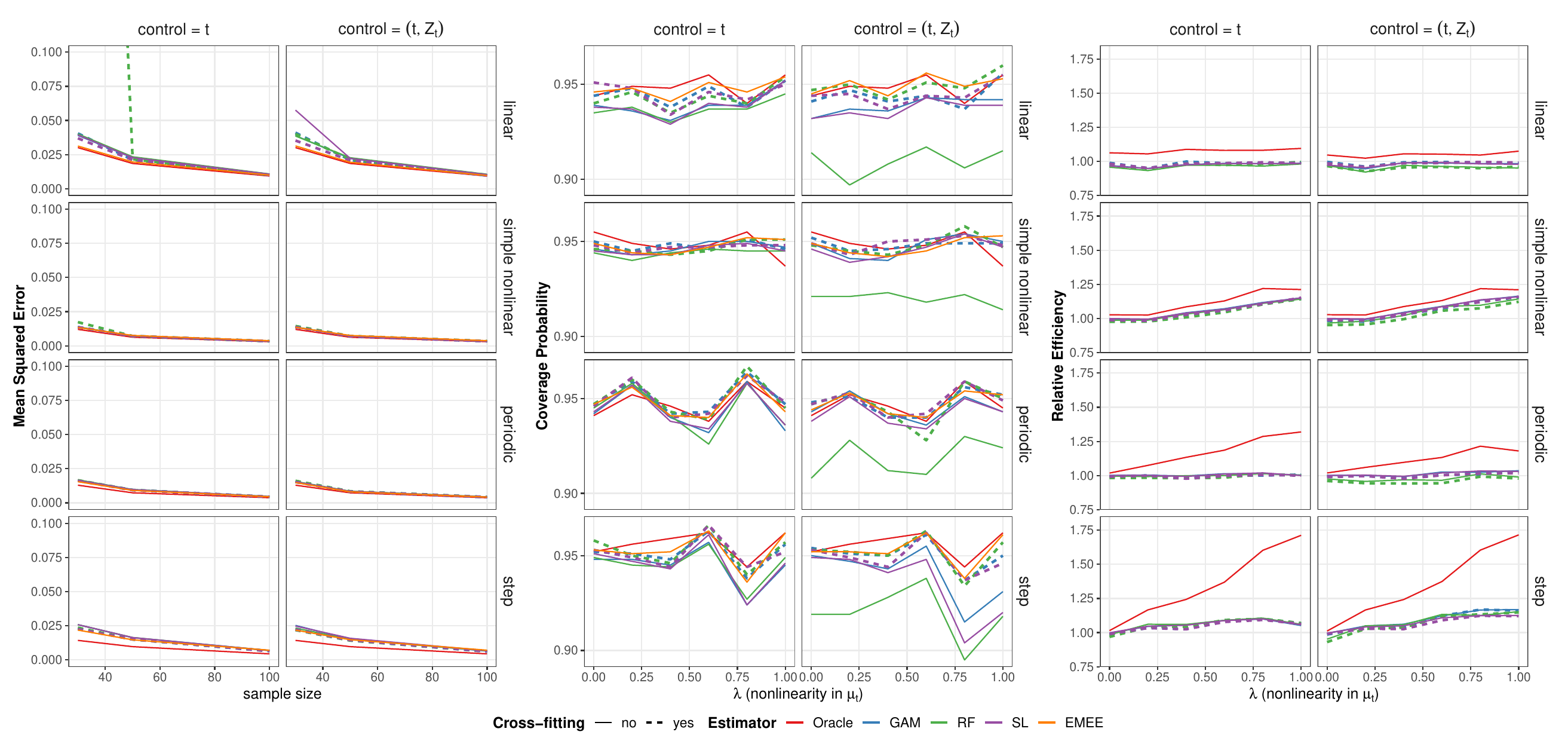}
	\caption{Simulation results for binary outcome where we varied $\lambda$, the nonlinearity in $\mu_t$. The other tuning parameter was fixed at $\rho = 0.1$. In the left panel, $\lambda$ is fixed at 0.8 and the sample size $n$ varies from 30 to 100. In the middle and the right panels, the sample size $n$ is fixed at 100 and $\lambda$ varies from 0 to 1. The large Mean Squared Error for cross-fitting RF is due to not finding the most optimal RF algorithm parameters on small sample size after sample splitting. We elaborate in more detail in \cref{sec:discussion}. }
	\label{A-fig:appendix-binary-control_pattern}
\end{figure}

\subsection{Count Outcome}
\label{A-subsec:appendix-count}

The generation of covariate $Z_t$ and treatment $A_t$ for the count outcome simulation is the same as the continuous outcome generative model in \cref{subsec:simulation-continuous}. The count proximal outcome $Y_{t,1}$ is generated from $\text{Poisson}(\mu_t(H_t, A_t))$, where $\mu_t(H_t, A_t) = \exp\{A_t (\beta_0)\} \times \mu_t(H_t, 0)$ and $\mu_t(H_t, 0)$ can take one of four forms: Log-linear, where $\log\mu_t(H_t, 0) = \alpha_0 + \alpha_1 t + \rho Y_{t-1, 1}$; simple nonlinear, where $\log\mu_t(H_t, 0) = \alpha_2 + \lambda q_{2,2}(t/T) + \rho Y_{t-1, 1}$ with $q_{2,2}$ being the probability density function of $\text{Beta}(2,2)$ distribution; periodic, where $\log \mu_t(H_t, 0) = \alpha_2 + \lambda \sin(t) + \rho Y_{t-1, 1}$; step function, where $\log\mu_t(H_t, 0) = \alpha_2 + \lambda \indic(t \text{ is even}) + \rho Y_{t-1, 1}$.  We set $\beta_0 = 0.1$, $\alpha_0 = -5$, $\alpha_1 = 0.8$, and $\alpha_2 = 0.5$. The tuning parameters $\lambda$ and $\rho$ capture different magnitudes of nonlinearity in $\mu_t(H_t, 0)$ and within-subject correlation in the outcome, respectively.

We aim to estimate the marginal CEE with $S_t = \emptyset$ and $g$ being the log link, with true value equal to $\beta_0$.
We set $\lambda \in [0,1]$, $\rho = 0.01$. The choice of control variables and the implementations of $\hat\beta$ and $\check\beta$ are the same as \cref{subsec:simulation-continuous}, except that SL-based $\hat\beta$ and $\check\beta$ are replaced by STACK-based ones, where Step 1 of \cref{algo:estimator-ncf} uses a variation of stacking ensemble that takes a weighted average of the individual learners that include generalized linear model, GAM, multivariate spline, and random forest. After each individual learner is trained, a Poisson regression regressing the outcome on the predictions from individual learners is fitted to determine the weight for each individual learner.
We implemented the STACK-based estimators because the super learner package does not tailor to count outcomes. Step 3 of \cref{algo:estimator-ncf} uses a pooled-across-time multivariate spline to approximate the two expectations.

The simulation results are presented in \cref{A-fig:appendix-count-control_pattern}. Qualitatively, the conclusions are similar to the continuous outcome simulation (\cref{subsec:simulation-continuous}). All estimators are consistent and robust against misspecified Step 1 working model; the confidence interval coverage of the cross-fitting version of GAM, RF, and STACK is better than that of the non-cross-fitting version; all the proposed estimators improve efficiency over EMEE and the efficiency gain increases as the nonlinearity parameter $\lambda$ increases, with the oracle being the most efficient; and using cross-fitting does not affect the efficiency of the proposed estimator. Across all settings, the proposed estimator shows substantial efficiency gain with a relative efficiency up to 1.75.

\begin{figure}
        \centering
        \includegraphics[width = \textwidth]{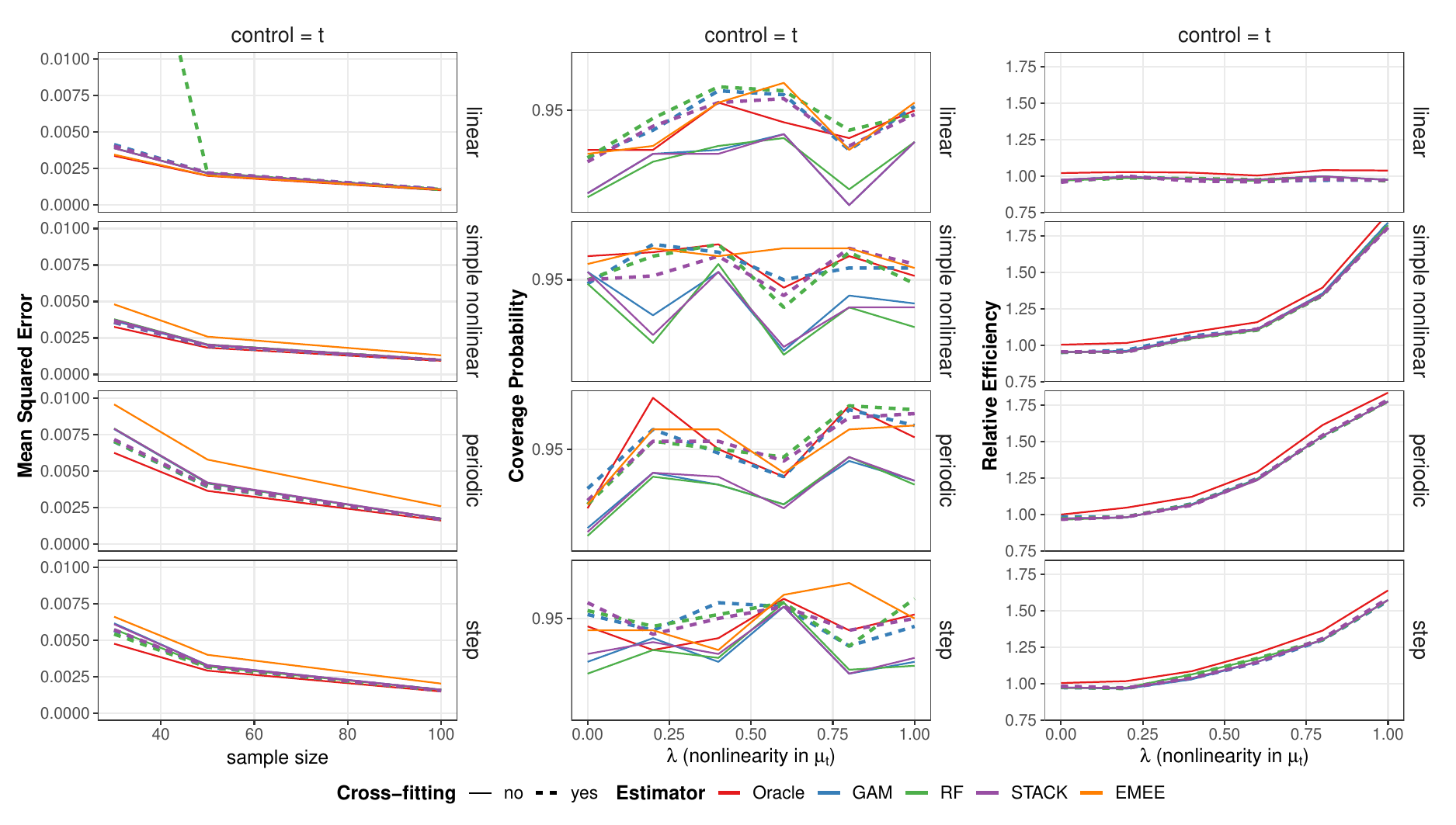}
	\caption{Simulation results for count outcome where we varied $\lambda$, the nonlinearity in $\mu_t$. The other tuning parameter was fixed at $\rho = 0.01$. In the left panel, $\lambda$ is fixed at 0.8 and the sample size $n$ varies from 30 to 100. In the middle and the right panels, the sample size $n$ is fixed at 100 and $\lambda$ varies from 0 to 1. The large Mean Squared Error for cross-fitting RF is due to not finding the most optimal RF algorithm parameters on small sample size after sample splitting. We elaborate in more detail in \cref{sec:discussion}. }
	\label{A-fig:appendix-count-control_pattern}
\end{figure}

\end{appendix}

\end{document}